\documentclass{article}

\usepackage{arxiv}
\pdfoutput=1 
\usepackage[utf8]{inputenc} 
\usepackage[T1]{fontenc}    
\usepackage{hyperref}       
\usepackage{url}            
\usepackage{booktabs}       
\usepackage{amsfonts}       
\usepackage{nicefrac}       
\usepackage[pdftex]{graphicx}
\usepackage{bm}
\usepackage{float}
\usepackage{amsmath}
\usepackage{subcaption}
\usepackage{enumitem}
\usepackage[english]{babel}
\usepackage{placeins}
\usepackage[boxruled, lined, vlined]{algorithm2e}
\usepackage{caption}
\usepackage{amsthm}

\DeclareMathSizes{10}{8}{5}{2}
\DeclareMathOperator*{\tr}{\text{Tr}}

\DeclareMathOperator*{\diag}{diag}
\DeclareMathOperator*{\argmin}{argmin}
\DeclareMathOperator*{\argmax}{argmax}
\allowdisplaybreaks

\theoremstyle{theorem}
\newtheorem{theorem}{Theorem}

\theoremstyle{lemma}
\newtheorem{lemma}{Lemma}[section]

\title{Parsimonious Feature Extraction Methods: Extending Robust Probabilistic Projections with Generalized Skew-t}

\author{
Dorota Toczydlowska \\
  School of Mathematical and Physical Sciences\\
  University of Technology Sydney\\
  \texttt{dtoczydlowska@gmail.com} \\
   \And
Gareth W. Peters \\
Department of Actuarial Mathematics and Statistics\\
Heriot-Watt University\\
  \texttt{garethpeters78@gmail.com} \\
     \And
Pavel V. Shevchenko \\
Department of Actuarial Studies and Business Analytics\\
Macquarie University\\
  \texttt{ pavel.shevchenko@mq.edu.au}
}

\begin{document}
\maketitle

\begin{abstract}
We propose a novel generalisation to the Student-t Probabilistic Principal Component methodology which: (1) accounts for an asymmetric distribution of the observation data; (2) is a framework for grouped and generalised multiple-degree-of-freedom structures, which provides a more flexible approach to modelling groups of marginal tail dependence in the observation data; and (3) separates the tail effect of  the error terms and factors. The new feature extraction methods are derived in an incomplete data setting to efficiently handle the presence of missing values in the observation vector. We discuss various special cases of the algorithm being a result of simplified assumptions on the process generating the data.   The applicability of the new framework is illustrated on a data set that consists of crypto currencies with the highest market capitalisation.
\end{abstract}

\keywords{probabilistic PCA; EM algorithm; robust orthogonal projections; skew grouped t-Copula; missing data; tail dependence; dependence modelling;}

\section{Introduction}
The study focuses on extension to the approach of Principal Component Analysis (PCA), as defined in \cite{Hotelling1933} , \cite{Jolliffe02} or  \cite{Vidal2016}. PCA and related matrix factorisation methodologies are widely used in data-rich environments for dimensionality reduction, data compression, feature-extraction techniques or data de-noising. The methodologies identify a lower-dimensional linear subspace to represent the data, which captures second-order dominant information contained in high-dimensional data sets. PCA can be viewed as a matrix factorisation problem which aims to learn the lower-dimensional representation of the data, preserving its Euclidean structure. However, in the presence of either a non-Gaussian distribution of the data generating distribution or in the presence of outliers which corrupt the data, the standard PCA methodology provides biased information about the lower-rank representation.

In many applications, the stochastic noise or observation errors in the data set are assumed to be, in some sense, ``well-behaved''; for instance, additive, light-tailed, symmetric and zero-mean. When non-robust feature extraction methods are naively utilised in the presence of violations of these implicit statistical assumptions, the information contained in the extracted features cannot be relied upon, resulting in misleading inference. Therefore, it is critical to ensure that the feature extraction captures information about correct characteristics of the process generating the data. In the following study, we relax the inherent assumption of ``well-behaved'' observation noise by developing a class of robust estimators that can withstand violations of such assumptions, which routinely arise in real data sets.

The investigated framework facilitates incorporation of prior assumptions about the data distribution into the model to ensure robust analysis of large incomplete datasets. 

Many improvements to the classical PCA methodology have been introduced in the literature to accommodate various characteristics that may deviate from the standard assumptions on the data when applying classical PCA methods. For instance robust variants of standard PCA modify the distance measure between each observation and its lower-rank approximation.     

Since the standard problem of PCA corresponds to finding the directions which maximise the covariance of the projected data, one group of improvements focuses on robustifying the calculation of the estimators of mean and covariance matrices, see \cite{Croux2000}, \cite{TorreBlack2001}, \cite{Croux2005}, \cite{Hubert2005} or \cite{Croux2007}.  

The use of a robust estimator is, in fact, equivalent to introducing observation-specific weights to the loss function of the standard PCA. The majority of these approaches are based on down-weighting, or even removing, observations with outlying distances; this categorises them as a non-probabilistic set of methodologies, as they do not directly incorporate any assumption about the noise distribution. The concept of observation-specific weights is also addressed in \cite{ParkZhangZhaKasturi2004} with weights, which are inversely proportional to the distance between the data points and which specify the distance measure on the sample. The study of \cite{ZhaoLinTang2007} assumes the local representation of the data as Gaussian.

The other class of approaches to address robust PCA investigates different types of measures which assess the distance between a set of observations and its projection. Consequently, new frameworks provide procedures which are efficient in the presence of various assumptions on the noise. The framework of PCA in the presence of sparse and of arbitrary amplitude has been studied by \cite{Tibshirani1996}, \cite{CrouxFilzmoser1998}, \cite{UkkelbergBorgen1993} and \cite{KeKanade2005}, \cite{DingZhouHeZha2006}, among many others, where authors proposed various PCA algorithms which incorporate regularised covariance shrinkage of an $L_1$-type loss function. 

A set of alternative methodologies, called projection pursuit methods, address the problem of representing a data matrix by sparse and dense components, see the review in \cite{Bouwmans2017}. Starting with the studies of \cite{WrightYang2009}, \cite{CandesLiMaWright2009} or \cite{ZhouLiWrightCandesMa2010}, the alternative PCA methodology aims to combine both sparse and dense noise patterns. The studies investigate a model for PCA which decomposes the data into a lower-rank matrix comprised of a small and dense noise matrix and a large and sparse noise matrix. Most of these methods are very sensitive to the initialisation step. For instance, the frameworks investigated by \cite{ZhouLiWrightCandesMa2010} requires a good estimate of the magnitude of the dense noise, which is usually difficult to obtain. The work of \cite{XuCaramanisSanghav2012} proposes a formulation of the problem which focuses on the exact recovery of the eigenvector representation of the data matrix, rather than the recovery of the data matrix as is broadly approached in other studies.

The other class of methodologies, which extends PCA to its probabilistic interpretation, were introduced by \cite{TippingBishop1999} and \cite{Roweis1998} as Probabilistic Principal Component Analysis (PPCA). In its first formulation, the standard PPCA assumes the observation vector to be Gaussian what allows for a straightforward interpretation of representation obtained by PCA in terms of PPCA, see \cite{Vidal2016}. PPCA can be easily tailored to handle incomplete information in the sample data and allows the utilisation of the probabilistic assumptions about both the type of missingness as well as the distribution of missing values. The Expectation-Maximisation (EM) algorithm, formalised by \cite{Dempster1977} and discussed in detail, with its extensions, in \cite{McLachanKrishnan1997}, is especially suited for inference of probabilistic models with unobserved or hidden variables.

One of the natural extensions of the standard formulation of PPCA is to introduce a heavy-tailed distribution to the process generating the data.  \cite{RidderFranc2003}, \cite{KhanDellaert2004}, \cite{Archambeau2006}, \cite{Fang2008} and \cite{Chen2009} address this problem and explore the use of the Student-t assumption on the noise distribution and its impact on the robustness of the methodology with respect to dense noise. \cite{Gao2008} formulates the PPCA problem with the Laplace error term and Gaussian latent variables. Other works introducing sparsity to the solution of PPCA are \cite{Guan2009} and \cite{BachArchambeau2009}, both of which incorporate a sparse prior distribution on the model's parameters via a variational Expectation-Maximisation. \cite{khanna2015} add sparse domain constraints on the distribution of latent variables.

\cite{XieXing2015} improve the robustness of PPCA to both sparse and dense outliers of significant magnitude by assuming that the error term and latent random vector follow a Cauchy distribution. Another flexible framework for PPCA is introduced in works of \cite{RednerWalker1984}, \cite{PeelMcLachlan2000}, \cite{Archambeau2005_phd}, \cite{ChenMorris2006} where they propose PPCA frameworks for mixture models, that follow Gaussian or Student-t distributions, in order to model arbitrary probability density functions of the observation process or the distribution of the noise which corrupts the data. 

The reviewed robust feature-extraction methods are primarily based on the assumption that observations are independent over time and that the marginal distributions of their components have the same profile of heavy tails. This reasoning might be criticised for having a limited ability to capture various tail-dependence patterns in multivariate data analysis. Therefore, we want to investigate an approach that will be able to accommodate a broader range of dependence and marginal distribution assumptions in the data-generating mechanism.  We comment that this new model for PPCA can easily be reduced to the simpler representation such as Gaussian PPCA and Student-t PPCA if the data reveals such characteristics. 

Therefore, our first contribution is to separate the tail effect of the error terms and factors that reflects the representation of the original data in the new basis. It allows for independent assumptions about the profiles of heavy tails of the error term and the original representation.

Secondly, we show how to employ Grouped t-copula into the PPCA framework.  The Grouped t-copula allows for a grouped or individual degrees of freedom parameter per marginal of a random vector. It has been explored in \cite{Daul2003}, \cite{Luo2010} and \cite{LuoShevchenko2009} in the context of risk modelling. The new component allows the marginal elements of unobserved vectors to have individual or grouped profiles of heavy-tails dependency structures and, consequently, provides greater flexibility in capturing second order characteristics of the data set.

Our next contribution is to combine the described concepts with the flexibility of modelling an asymmetric correlation and heavy- tail dependence in a multivariate setting.  We focus on the skewed Student-t distribution from the Generalized Hyperbolic family of distributions as defined and discussed in \cite{BarndorffNielsen1981} or \cite{Demarta2005},  and comprehensively compared with other families of skewed distributions in \cite{Aas2006}.  The type of a Student-t copula that accommodate skewness and individual degrees of freedom is studied in \cite{Church2012}.  

In addition, we study the robustness of the developed class of the PPCA as defined by \cite{Hampel1986}. The introduced structural components of the representation of the data generating process increase the flexibility of PPCA frameworks to take into account different features of the data. These features may impact on the obtained projection as well as on its rotation.  Given numerical examination, we show that this flexibility allows capturing complex characteristic of the data generating process, when they appear. Also, it results in an accurate estimation in the presence of the dynamics that are consistent with the assumptions of standard PCA or PPCA approaches.

Lastly, we develop an efficient Expectation-Maximisation (EM) algorithm of \cite{Dempster1977} that estimates the parameters of this new class of PPCA methods.  The framework handles the presence of missing data, and we comment how the procedure can be adjusted to various assumptions about the patterns of missing data. 

We apply our framework to cryptocurrencies data, and show how the new methodology can be accommodated to guide portfolio construction by measuring market concentration, the potential for diversification or hedging.

\section{Introduction to Probabilistic Principal Component Analysis}\label{sec:GaussianPPCA}
Let the $d$-dimensional random vector $\mathbf{Y}_t$ represents a process that generates the observation data with a realisation $\mathbf{y}_t$  at time $t$ . We observe $N$ realisations of $\mathbf{Y}_t$ ,  $\mathbf{y}_{1:N} = \big\{\mathbf{y}_{1},\ldots,\mathbf{y}_{N} \big\}$.  The standard PPCA, that assumes Gaussian distribution of $\mathbf{Y}_t$, has been introduced by \cite{TippingBishop1999}. The method seeks $k$ - dimensional uncorrelated latent vector $\mathbf{X}_t$ which provides the most meaningful model of $\mathbf{Y}_t$,
\begin{equation}\label{eq:ppca_model}
\mathbf{Y}_t = \bm{\mu} + \mathbf{X}_t \mathbf{W}_{d \times k}^T + \bm{\epsilon}_t,
\end{equation}
for a vector of constants $\bm{\mu} \in \mathbb{R}^d$ and $d$-dimensional error term random vector $\bm{\epsilon}_t$.  As remarked in \cite{TippingBishop1999}, in contrary to the standard PCA, the probabilistic version does not requires the orthogonality condition of $\mathbf{W}$, that is $\mathbf{W}^T \mathbf{W} = \mathbb{I}_k$, where $\mathbb{I}_k$ denotes a $k$ by $k$ identity matrix. This condition was essential in non-probabilistic PCA in order to impose a restriction or identification of a unique solution. In the optimisation problem in the classical PCA, the condition limits the space of possible solutions that minimize the distance between the original data and its projection to the new orthonormal space.  On the other hand, the objective function of the probabilistic PPCA can be represented by the likelihood  of the considered model and when it is combined in PPCA with a distribution on the factors, the marginal likelihood having integrated out the random factors removes the need for such a constraint.

In the classical Gaussian PPCA formulation it is assumed that the latent random vectors being sought in the feature extraction, that characterise the data, are distributed according to a multivariate normal distributions, that is $\mathbf{X}_t \sim \mathcal{N}\left(0, \mathbb{I}_k \right)$  and $\bm{\epsilon}_t \sim \mathcal{N}\left(0, \sigma^2 \mathbb{I}_d\right)$ where $\sigma^2 \geq 0$. They are also assumed to be mutually independent and independent over time. Given the model in  \eqref{eq:ppca_model}, $\mathbf{Y}_t$ is also $d-$dimensional random vector which follows Gaussian distribution and is independent over time, with mean $\bm{\mu}$ and the covariance matrix $\mathbf{C} = \mathbf{W} \mathbf{W}^T + \sigma^2 \mathbb{I}_d$. 

In the Gaussian PPCA,  the objective is to estimate the projection matrix $\mathbf{W}$, the vector $\bm{\mu}$ and the scalar $\sigma^2$ given the marginal distribution of $\mathbf{Y}_t$ 
\begin{equation*}
\mathbf{Y}_t | \Psi \sim \mathcal{N} \left(\bm{\mu}, \mathbf{W}\mathbf{W}^T + \sigma^2 \mathbb{I}_d  \right),
\end{equation*}
for the static parameters $\Psi = \left[\mathbf{W}, \bm{\mu},\sigma^2 \right]$ of the model in  \eqref{eq:ppca_model}. Given $N$ realisations of $\mathbf{Y}_t |\Psi$, the marginal likelihood $L(\Psi;\mathbf{y}_{1:N}) : =  \pi_{\mathbf{Y}_{1:N}|\Psi}(\mathbf{y}_{1:N})$ of the model under the Gaussian case can be factorized as 
\begin{equation}\label{eq:EM_loglik}
L(\Psi;\mathbf{y}_{1:N}) = \big( 2\pi \big)^{- \frac{N}{2}} \Big| \mathbf{W}\mathbf{W}^T + \sigma^2\mathbb{I}_d  \Big|^{- \frac{N}{2}}  \exp \bigg\{ - \frac{1}{2}\sum_{t = 1}^N (\mathbf{y}_t - \bm{\mu} )(\mathbf{W}\mathbf{W}^T + \sigma^2\mathbb{I}_d)^{-1}(\mathbf{y}_t - \bm{\mu} )^T	\bigg\},
\end{equation}
where the marginalisation is undertaken with regards to $\mathbf{X}_t$ random vectors.  

In order to calculate the covariance matrix we have to estimate the parameters $\Psi$ and marginalise $\mathbf{X}_t$, achieved by the iterative procedure of the EM algorithm of \cite{Dempster1977}. The steps and derivation of the algorithm have been described in \cite{RubinThayer1982} or \cite{TippingBishop1999} where no missingness is assumed. The  EM algorithm finds Maximum Likelihood Estimation (MLE) estimates of parameters in probabilistic models when the direct optimisation of a likelihood function is not feasible. The MLE estimates of $\Psi \in \Omega$ are computed by maximising the marginalized likelihood function $L(\Psi;\mathbf{y}_{1:N})$ which in the Gaussian PPCA model is given in \eqref{eq:EM_loglik}. The space $\Omega$ represents the parameter space of $\Psi$.

In order to iteratively find a stationary point of the function in \eqref{eq:EM_loglik}, the EM algorithms exploits the artificial formulation of probabilistic models with regards to $\mathbf{Y}_t$ as incomplete information about the studied model with the latent vector $\mathbf{X}_t$ being assumed to be a missing part of the complete random vector $(\mathbf{Y}_t,\ \mathbf{X}_t$). The joint model is know for certain assumptions about the distribution of error and $\mathbf{X}_t$ and so the joint likelihood of $\mathbf{Y}_t$  and $\mathbf{X}_t$ is known given its probability density function $\pi_{\mathbf{Y}_t,\mathbf{X}_t|\Psi}(\mathbf{y}_t,\mathbf{x}_t)$. The random vector $\mathbf{Y}_t$ acts as an observable elements of this vector. 

Each iteration of the EM algorithm seeks maximizers of $L(\Psi;\mathbf{y}_{1:N})$ with respect to $\Psi$, and consists of two steps: an expectation step (E-step) and a maximisation step (M-step).  The E-step infers missing values or latent variables, $\mathbf{X}_{1:N}$, by finding their distribution given the known observed values $\mathbf{Y}_{1:N}$, and current estimates of parameters. It then integrates the joint log-likelihood or complete data likelihood with regards to the distribution of these latent random vectors. At the $i$-th iteration, the E-step specifies an estimate of complete information formulated as a function of parameters, that is
\begin{equation}\label{eq:EM_intro_Q}
Q(\Psi^*,\Psi):=\mathbb{E}_{\mathbf{X}_{1:N}|\mathbf{Y}_{1:N},\Psi^*} \Big[ \log \pi_{\mathbf{Y}_{1:N},\mathbf{X}_{1:N}|\Psi}(\mathbf{y}_{1:N},\mathbf{x}_{1:N})\Big] \text{ for } \Psi^* = \Psi^{(i)}.
\end{equation}
Next, the M-step maximises the marginalised complete data likelihood obtained from the E-step in \eqref{eq:EM_intro_Q} which is now just a function of observed $\mathbf{Y}_{1:N}$ and parameters $\Psi$
\begin{equation*}
\Psi^{(i+1)} = \argmax_{\Psi \in \Omega} Q(\Psi^{(i)},\Psi).
\end{equation*}
The key idea behind the steps of the algorithm is to use the following representation of the logarithm of the likelihood function, which exploits Bayes' rule applied to $\pi_{\mathbf{Y}_t,\mathbf{X}_t|\Psi}(\mathbf{y}_t,\mathbf{x}_t)$, that is
\begin{align} \label{eq:Em_loglik_formulation}
l(\Psi,\mathbf{y}_{1:N})& :  = \log L(\Psi;\mathbf{y}_{1:N}) =   \log \pi_{\mathbf{Y}_{1:N},\mathbf{X}_{1:N}|\Psi}(\mathbf{y}_{1:N},\mathbf{x}_{1:N}) - \log \pi_{\mathbf{X}_{1:N}|\mathbf{Y}_{1:N},\Psi}(\mathbf{x}_{1:N}).
\end{align} 
By noting that the term $l(\Psi,\mathbf{y}_{1:N})$ is invariant under the expectation with respect to the conditional distribution $\mathbf{X}_{1:N}|\mathbf{Y}_{1:N},\Psi^*$ we have that
\begin{align*} 
\log \pi_{\mathbf{Y}_{1:N}|\Psi}(\mathbf{y}_{1:N}) & = \mathbb{E}_{\mathbf{X}_{1:N}|\mathbf{Y}_{1:N},\Psi^*} \Big[ \log \pi_{\mathbf{Y}_{1:N},\mathbf{X}_{1:N}|\Psi}(\mathbf{y}_{1:N},\mathbf{x}_{1:N})\Big] - \mathbb{E}_{\mathbf{X}_{1:N}|\mathbf{Y}_{1:N},\Psi^*}\Big[\log \pi_{\mathbf{X}_{1:N}|\mathbf{Y}_{1:N},\Psi}(\mathbf{x}_{1:N})\Big],
\end{align*} 
for some $\Psi^* \in \Omega$. In their study, \cite{Dempster1977} shows that the maximizers of $\log \pi_{\mathbf{Y}_{1:N}|\Psi}(\mathbf{y}_{1:N})$ can be specified by iteratively optimising the first component of the representation in \eqref{eq:Em_loglik_formulation}, the expectation $Q(\Psi^*,\Psi)$ defined in \eqref{eq:EM_intro_Q}, using the steps of the EM algorithm. \cite{Dempster1977} shows that the sequence of the log-likelihood function evaluations, obtained iteratively in EM algorithm updates  of the parameters, denoted by $\Big\{ l^{(i)} \Big\}_{i\in \mathbb{N}_{0}}$ for $l^{(i)} = l(\Psi^{(i)},\mathbf{y}_{1:N})$ ,  is non-decreasing and consequently, each iteration of the EM algorithm results in the update of the parameter $\Psi$ which increases the loglikelihood in \eqref{eq:EM_loglik} or leaves it unchanged.  Therefore, the EM algorithm monotonically increases the likelihood function during each iteration. The studies of \cite{Dempster1977}, \cite{Wu1983} and \cite{Boyles1983} investigate additional assumptions  such as monotonicity of the sequence $\Big\{ l^{(i)} \Big\}_{i\in \mathbb{N}_{0}}$ or the smoothness of the objective function  which need to be satisfied in order to ensure that the sequence  $\Big\{ l^{(i)} \Big\}_{i\in \mathbb{N}_{0}}$ converges to a stationary point or, more specifically, a local or global maximum. 

\section{Generalized Skew-t Probabilistic Principal Component} \label{sec:GStSPPCA}
Following the concept of combining Skew-t and Grouped t-copula distributions discussed in  \cite{Church2012}, we introduce a PPCA model which allows one to develop tail dependence structures more flexible than the ones under Gaussian PPCA. Our novel proposed PPCA model will allow a greater degree of flexibility, especially when asymmetry is present in tail dependence between pairs or sub-sets of the multivariate random observation vectors.  We achieve this by developing novel extensions based on grouped and generalised Student-t PPCA models. 

Consider the stochastic representation of the Student-t random variables, which can be expressed as a scale mixture of a Gaussian random vector and Gamma variable, as formulated in \cite{Gupta1999} or \cite{Kotz2004}. Our extension to PPCA assumes representing the scale mixtures of $\mathbf{X}_t$ and $\bm{\epsilon}_t$ by independent Gamma random variables. Consequently, the vectors themselves are mutually independent and have individual dependency structures. The assumption provides the model with the additional flexibility to determine which component of the model impacts on marginal tails behaviour of the observation vector. We introduce the coefficient of skewness which specifies the strength of asymmetry in the distribution of the unobserved random vectors using the definition of the hyperbolic Skew-t distribution as introduced in \cite{BarndorffNielsen1981} or \cite{Demarta2005}. We choose this  definition of Skew-t distribution due to two reasons: simplicity of conditional distributions given by the stochastic representation; and the appealing property of the hyperbolic Skew-t distribution remarked in  \cite{Aas2006}, that the tails of corresponding distributions have different behaviours, polynomial  and exponential. Therefore,  the tails of the distribution can have different magnitude of heaviness. 

Lastly, we want to highlight that under appropriate assumptions on the deterministic parameters of the introduced models, the generalized PPCA framework reduces to various special cases such as the PPCA model under a Grouped t-copula distribution (when skewness is equal to zero), or the PPCA model under a Skew-t distribution (when degrees of freedom are equal per marginal). 

We consider two cases of the distributions: the first in which the random vectors $\mathbf{X}_t$ and $\bm{\epsilon}_t$ are independently and non-identically distributed, and the second in which they are identically and conditionally independently distributed over time -- the assumptions and the derivations of the latter model are given in \cite{Toczydlowska_thesis}.

\subsection{Independent Generalized Skew-t Probabilistic Principal Component} \label{ssec:PPCA_GStS_ind}
Let us denote  two mutually independent and identically distributed over time uniform random variables $S_{\epsilon,t}, S_{x,t} \sim \mathcal{U}\left(0,1 \right)$. For convenience of the notation, we denote $d$- and $k$-dimensional random vectors $\mathbf{U}_{t}$ and $\mathbf{V}_{t}$, respectively,
\begin{equation}\label{eq:latentProcesStochRep_U_V_ind}
\mathbf{U}_t  = \left( \frac{\chi_{\nu_{\epsilon}^1}^{-1} (S_{\epsilon,t})}{\nu_{\epsilon}^1}, \ldots, \frac{\chi_{\nu_{\epsilon}^d}^{-1} (S_{\epsilon,t})}{\nu_{\epsilon}^d} \right)_{1 \times d} \ \text{ and } \ \mathbf{V}_t = \left( \frac{\chi_{\nu_x^1}^{-1} (S_{x,t})}{\nu_x^1}, \ldots, \frac{\chi_{\nu_x^k}^{-1} (S_{x,t})}{\nu_x^k} \right)_{1 \times k},
\end{equation}
for vectors of non-negative real numbers $\bm{\nu}_{\epsilon} = \lbrace \nu_{\epsilon}^1, \ldots, \nu_{\epsilon}^d \rbrace$ and $\bm{\nu}_x = \lbrace \nu_x^1, \ldots, \nu_x^k \rbrace$ and $\chi_{\nu}^{-1}$ denoting the quantile function of the Chi-square distribution with $\nu$ degrees of freedom. Note, that each of the vectors  $\mathbf{U}_t$ and $\mathbf{V}_t$ follows a multivariate Gamma distribution, are mutually independent and independent in time. However, the components of the vectors are dependent. In fact, they are co-monotonic, since they are constructed as transformations of a common uniform variable at time $t$.

Let us denote $d$-dimensional and $k$-dimensional real valued model parameter vectors, $\bm{\delta}_\epsilon$ and $\bm{\delta}_x$. The stochastic representation of the $d$-dimensional error term $\bm{\epsilon}_t$ and the $k$-dimensional latent variable $\mathbf{X}_t$ is given by 
\begin{align}\label{eq:latentProcesStochRep_X_epsilon_ind}
 \mathbf{X}_t = \mathbf{V}_t^{-1} \circ \bm{\delta}_x + \sqrt{\mathbf{V}_t^{-1}} \circ \mathbf{Z}_{x, t} \ \text{ and } \ \bm{\epsilon}_t = \mathbf{U}_t^{-1} \circ \bm{\delta}_\epsilon + \sqrt{\sigma^2 \mathbf{U}_t^{-1}} \circ \mathbf{Z}_{\epsilon,t},
\end{align}
for $\mathbf{Z}_{x, t}$ and $\mathbf{Z}_{\epsilon, t}$ being mutually independent standard normal multivariate variables, $k$- and $d$-dimensional respectively, and being independent of $\mathbf{U}_t$ and $\mathbf{V}_t$ (or $S_{x,t}$ and $S_{\epsilon,t}$ likewise). The operator $\circ$ denotes the Hadamard product, that is, for two $d$-dimensional vectors $\mathbf{a}$ and $\mathbf{b}$, the product of thee vectors results in the $d$-dimensional vector $\mathbf{a} \circ \mathbf{b} = \big(a^1b^1,\ldots, a^db^d\big)$. Consequently, we have the following joint probability density function of the Generalized Skew-t PPCA (GSt PPCA) model given $N$ realisations of the random vectors at times $t = 1, \ldots, N$, that is
\begin{align} \label{eq:joinProbability_ind}
& \pi_{\mathbf{Y}_{1:N},\mathbf{X}_{1:N}, \mathbf{U}_{1:N}, \mathbf{V}_{1:N}|\Psi} ( \mathbf{y}_{1:N},\mathbf{x}_{1:N}, \mathbf{u}_{1:N}, \mathbf{v}_{1:N} )
= \prod_{t=1}^N \bigg\{\pi_{\mathbf{Y}_t|\mathbf{X}_t, \mathbf{U}_t, \mathbf{V}_t,\Psi} ( \mathbf{y}_t) \cdot \pi_{\mathbf{X}_t| \mathbf{U}_t, \mathbf{V}_t,\Psi} (\mathbf{x}_t) \cdot \pi_{\mathbf{U}_t | \Psi} ( \mathbf{u}_t) \cdot \pi_{\mathbf{V}_t | \Psi} ( \mathbf{v}_t) \bigg\},
\end{align}
for $\Psi = \left[ \mathbf{W},\ \bm{\mu}, \ \sigma^2,\ \bm{\delta}_\epsilon, \ \bm{\delta}_x, \ \bm{\nu}_{x}, \ \bm{\nu}_{\epsilon} \right]$ being a vector which consists of all static parameters in the model specified in \eqref{eq:ppca_model} under GSt PPCA assumptions. Recall that the distributions of  $\mathbf{Y}_t, \ \mathbf{X}_t$ and $\bm{\epsilon}_t$ are conditionally multivariate Gaussian, such that
\begin{align*}\notag
& \mathbf{Y}_t | \mathbf{X}_t, \mathbf{U}_t, \mathbf{V}_t, \Psi   \sim  \mathcal{N}\left(\bm{\mu} +  \bm{\delta}_\epsilon \mathbf{D}_{\epsilon,t}^{-1} +  \mathbf{X}_t \mathbf{W} ^T,  \sigma^2 \mathbf{D}_{\epsilon,t}^{-1} \right), \ \mathbf{X}_t | \mathbf{V}_t, \Psi    \sim  \mathcal{N}\left( \bm{\delta}_x \mathbf{D}_{x,t}^{-1}, \mathbf{D}_{x,t}^{-1} \right), \  \ \bm{\epsilon}_t | \mathbf{U}_t, \Psi    \sim  \mathcal{N}\left( \bm{\delta}_\epsilon \mathbf{D}_{\epsilon,t}^{-1}, \sigma^2 \mathbf{D}_{\epsilon,t}^{-1} \right), 
\end{align*}
where
\begin{equation*}
\mathbf{D}_{\epsilon,t} =  \begin{pmatrix}
U_t^1 &0&0 \\
0& \ddots &0 \\
0& 0& U_t^d
\end{pmatrix}_{d \times d} , \
\mathbf{D}_{x,t} =  \begin{pmatrix}
V_t^1 &0&0 \\
0& \ddots &0 \\
0&0 & V_t^k
\end{pmatrix}_{k \times k}.
\end{equation*}

\subsubsection{Special Cases: PPCA with Skew-t Distribution and Grouped-t Distribution}\label{sssec:PPCA_GStS_special}
Under appropriate assumptions on the deterministic parameters of the model, the GSt PPCA reduces to simpler special cases. We show how these special cases are defined using independent PPCA model assumptions from Subsection \ref{ssec:PPCA_GStS_ind} as the identical and conditionally distributed case proceeds analogously.  We can straightforwardly obtain the following three models
\begin{enumerate}[itemindent = 5em, leftmargin = *,font={\bfseries},label=Special Case \arabic*:]
\item $\bm{\delta}_\epsilon = \bm{\delta}_x = \mathbf{0}$.

Given the skewness coefficients equal to zero, the generalized PPCA model reduces to \textbf{Grouped-t GSt PPCA} case with the representation of the latent processes
\begin{align*}
\mathbf{X}_t = \sqrt{\mathbf{V}_t^{-1}} \circ \mathbf{Z}_{x, t}
\ \text{ and } \ \bm{\epsilon}_t = \sqrt{\sigma^2 \mathbf{U}^{-1}_t} \circ \mathbf{Z}_{\epsilon,t},
\end{align*}
results in the conditional distributions given by 
\begin{align*}\notag
& \mathbf{Y}_t | \mathbf{X}_t, \mathbf{U}_t, \mathbf{V}_t, \Psi    \sim  \mathcal{N}\left(\bm{\mu} +  \mathbf{X}_t \mathbf{W} ^T,  \sigma^2 \mathbf{D}_{\epsilon,t}^{-1} \right), \ \mathbf{X}_t | \mathbf{V}_t, \Psi    \sim  \mathcal{N}\left( \mathbf{0}, \mathbf{D}_{x,t}^{-1} \right) ,\ \ \bm{\epsilon}_t | \mathbf{U}_t, \Psi    \sim  \mathcal{N}\left( \mathbf{0}, \sigma^2 \mathbf{D}_{\epsilon,t}^{-1} \right). 
\end{align*}
 
\item $\nu_\epsilon^1 = \ldots = \nu_\epsilon^d = \nu_\epsilon $ and $\nu_x^1 = \ldots = \nu_x^k = \nu_x $.

When we assume that the marginal distributions of $\mathbf{X}_t$ and $\bm{\epsilon}_t$ are characterized by the same heavy tails, the scaling variable $\mathbf{U}_t$ and $\mathbf{V}_t$ become one dimensional as they components are identical, $U_t = \frac{\chi_{\nu_{\epsilon}}^{-1} (S_{\epsilon,t})}{\nu_{\epsilon}}$ and $V_t = \frac{\chi_{\nu_x}^{-1} (S_{x,t})}{\nu_x}$. Therefore, the generalized PPCA model reduces to \textbf{Skew-t GSt PPCA} case with the representation of the latent processes
\begin{align*}
 \mathbf{X}_t = V_t^{-1}  \bm{\delta}_x + \sqrt{V_t^{-1}} \mathbf{Z}_{x, t} \ \text{ and } \ \bm{\epsilon}_t = U_t^{-1} \bm{\delta}_\epsilon + \sqrt{\sigma^2 U_t^{-1}} \mathbf{Z}_{\epsilon,t},
\end{align*}
results in the conditional distributions given by 
\begin{align*}\notag
& \mathbf{Y}_t | \mathbf{X}_t, U_t, V_t, \Psi    \sim  \mathcal{N}\left(\bm{\mu} +  U_t^{-1}  \bm{\delta}_\epsilon +  \mathbf{X}_t \mathbf{W} ^T,  \sigma^2 U_t^{-1} \mathbb{I}_d \right), \ \mathbf{X}_t | V_t, \Psi   \sim  \mathcal{N}\left( V_t^{-1} \bm{\delta}_x , V_t^{-1} \mathbb{I}_k \right), \  \ \bm{\epsilon}_t | U_t, \Psi    \sim  \mathcal{N}\left( U_t^{-1} \bm{\delta}_\epsilon, \sigma^2 U_t^{-1} \right).
\end{align*}

\item $\bm{\delta}_\epsilon = \bm{\delta}_x = \mathbf{0}$,  $\nu_\epsilon^1 = \ldots = \nu_\epsilon^d = \nu_\epsilon $ and $\nu_x^1 = \ldots = \nu_x^k = \nu_x $.

When we assume that the marginal distributions of $\mathbf{X}_t$ and $\bm{\epsilon}_t$ are characterized by the same heavy tails and the skewness coefficients are equal to zero, the generalized PPCA model reduces to \textbf{Student-t GSt PPCA} case which separates the tail effect of the error term $\bm{\epsilon}_t$ and the latent process $\mathbf{X}_t$ on the observation vector $\mathbf{Y}_t$. Given the independent case from Subsection  \ref{ssec:PPCA_GStS_ind} , the representation of the latent processes
\begin{align*}
\mathbf{X}_t = \sqrt{V_t^{-1}}  \mathbf{Z}_{x, t} \ \text{ and } \ \bm{\epsilon}_t = \sqrt{\sigma^2 U_t^{-1}}  \mathbf{Z}_{\epsilon,t},
\end{align*}
results in the conditional distributions given by 
\begin{align*}
& \mathbf{Y}_t | \mathbf{X}_t, U_t, V_t, \Psi    \sim  \mathcal{N}\left(\bm{\mu} +   \mathbf{X}_t \mathbf{W} ^T,  \sigma^2 U_t^{-1} \mathbb{I}_d \right), \  \mathbf{X}_t | V_t, \Psi    \sim  \mathcal{N}\left( \mathbf{0}, V_t^{-1} \mathbb{I}_k \right), \  \
\bm{\epsilon}_t | U_t, \Psi   \sim  \mathcal{N}\left( \mathbf{0}, \sigma^2 U_t^{-1} \right).
\end{align*}
Recall that the above formulation is an additional variant to the Student-t PPCA model derived in  \cite{RidderFranc2003}, \cite{KhanDellaert2004} or \cite{ToczydlowskaPeters2018}. The difference lies in the stochastic representation of the random vectors $\bm{\epsilon}_t$ and $\mathbf{X}_t$. In contrary to the existing models, our formulation assumes the independence of the scaling random variables, $U_t$ and $V_t$, and therefore allows to model separately the tails of $\bm{\epsilon}_t$ and $\mathbf{X}_t$.
\end{enumerate}

\subsection{Formulation and Eigendecomposition of the Covariance Matrix of $\mathbf{Y}_t$}
\label{ssec:comments_on_eigen}
We want to define eigenvectors and eigenvalues of the covariance matrix of $\mathbf{Y}_t$ given its formulation according to the model in \eqref{eq:ppca_model} under the model assumption discussed in Section \ref{ssec:PPCA_GStS_ind}, that is
\begin{align*} \notag
\mathbf{Cov}_{\mathbf{Y}_t|\Psi}\big[ \mathbf{Y}_t\big] &: = \mathbb{E}_{\mathbf{Y}_t|\Psi}\big[ \mathbf{Y}_t^T\mathbf{Y}_t \big] - \mathbb{E}_{\mathbf{Y}_t|\Psi}\big[ \mathbf{Y}_t \big]^T \mathbb{E}_{\mathbf{Y}_t|\Psi}\big[ \mathbf{Y}_t\big]  = \mathbf{W} \mathbf{Cov}_{\mathbf{X}_t|\Psi}\big[ \mathbf{X}_t\big]\mathbf{W}^T + \mathbf{Cov}_{\bm{\epsilon}_t|\Psi}\big[ \bm{\epsilon}_t\big].
\end{align*}
We discuss the parametrisation of the covariance matrix of the observation process in terms of the static parameters of the model defined in Section \ref{ssec:PPCA_GStS_ind}. We highlight how assumptions of the special cases of the GSt PPCA model listed in Subsection \ref{sssec:PPCA_GStS_special} impact on an eigendecomposition of $\mathbf{Cov}_{\mathbf{Y}_t|\Psi}\big[ \mathbf{Y}_t\big] $ as well as the characteristics of the representation determined by latent vector $\mathbf{X}_t$. We show that 
\begin{enumerate}
\item[(1)] when skewness is not present in the model (Cases 1 \& 3), the eigenvectors of $\mathbf{Cov}_{\mathbf{Y}_t|\Psi}\big[ \mathbf{Y}_t\big]$ are defined by the singular vectors of $\mathbf{W}$. In the presence of the skewness, the eigenvectors are also defined by the parameters of skewness, $\bm{\delta}_x$ and $\bm{\delta}_\epsilon$;
\item[(2)] when skewness is present in the model, the expectations and covariance matrices of the random vectors $\mathbf{X}_t$ and  $\bm{\epsilon}_t$ are defined by the corresponding parameters of skewness and might not be zero mean or diagonal, respectively. Therefore, the new representation determined by $\mathbf{X}_t$ might not be orthogonal;
\item[(3)]  the parameters of degrees of freedom influence the eigenvalues of $\mathbf{Cov}_{\mathbf{Y}_t|\Psi}\big[ \mathbf{Y}_t\big]$.
\end{enumerate}

\subsubsection{Solution for GSt PPCA Without Grouped Heavy Tails of Marginals (Cases 2 \& 3) }
The latent variable and the error terms has the following stochastic representation
\begin{align*}
\mathbf{X}_t = V_t^{-1}  \bm{\delta}_x + \sqrt{V_t^{-1}} \mathbf{Z}_{x, t} \ \text{ and } \ \bm{\epsilon}_t = U_t^{-1} \bm{\delta}_\epsilon + \sqrt{\sigma^2 U_t^{-1}} \mathbf{Z}_{\epsilon,t}.
\end{align*}
Firstly, let us recall that since $V_t \sim \Gamma \big(\frac{v_x}{2},\frac{v_x}{2} \big)$ and $U_t \sim \Gamma \big(\frac{v_\epsilon}{2},\frac{v_\epsilon}{2} \big)$, the moments of their inverses are 
\begin{align*} \notag
 \mathbb{E}_{V_t|\Psi}\big[ V_t^{-1}  \big]  = \frac{v_x}{v_x-2} \text{ and }  \mathbf{Cov}_{V_t|\Psi}\big[ V_t^{-1}  \big]  = \frac{2v_x}{(v_x-2)(v_x-4)^2}, \\\notag
  \mathbb{E}_{U_t|\Psi}\big[ U_t^{-1}  \big]  = \frac{v_\epsilon}{v_\epsilon-2}  \text{ and }  \mathbf{Cov}_{U_t|\Psi}\big[ U_t^{-1}  \big]  = \frac{2v_\epsilon}{(v_\epsilon-2)(v_\epsilon-4)^2}.
\end{align*}
The moments of the marginal distributions of $\mathbf{X}_t$ and $\bm{\epsilon}_t$, which are needed to specify the covariance matrix of $\mathbf{Y}_t$, are then
\begin{align*}\notag
& \mathbb{E}_{\mathbf{X}_t|\Psi}\big[ \mathbf{X}_t \big] = \mathbb{E}_{V_t|\Psi}\Big[ \mathbb{E}_{\mathbf{X}_t|V_t, \Psi }\big[ \mathbf{X}_t \big]  \Big] =   \frac{v_x}{v_x-2} \bm{\delta}_x,  \\\notag
& \mathbf{Cov}_{\mathbf{X}_t|\Psi}\big[ \mathbf{X}_t \big] =  \mathbb{E}_{V_t|\Psi}\Big[ \mathbf{Cov}_{\mathbf{X}_t|V_t, \Psi }\big[ \mathbf{X}_t \big]  \Big] +  \mathbf{Cov}_{V_t|\Psi}\Big[ \mathbb{E}_{\mathbf{X}_t|V_t, \Psi }\big[ \mathbf{X}_t \big]  \Big] =   \frac{v_x}{v_x-2}  \mathbb{I}_k +   \bm{\delta}_x^T   \bm{\delta}_x  \frac{2v_x}{(v_x-2)^2(v_x-4)},
\end{align*}
and 
\begin{align*}\notag
& \mathbb{E}_{\bm{\epsilon}_t\Psi}\big[\bm{\epsilon}_t \big] = \mathbb{E}_{U_t|\Psi}\Big[ \mathbb{E}_{\bm{\epsilon}_t|U_t, \Psi }\big[ \bm{\epsilon}_t \big]  \Big] = \frac{v_\epsilon}{v_\epsilon-2}  \bm{\delta}_\epsilon,  \\\notag
& \mathbf{Cov}_{\bm{\epsilon}_t|\Psi}\big[ \bm{\epsilon}_t \big] =  \mathbb{E}_{U_t|\Psi}\Big[ \mathbf{Cov}_{\bm{\epsilon}_t|U_t, \Psi }\big[ \bm{\epsilon}_t \big]  \Big] +  \mathbf{Cov}_{U_t|\Psi}\Big[ \mathbb{E}_{\bm{\epsilon}_t|U_t, \Psi }\big[\bm{\epsilon}_t \big]  \Big] =  \sigma^2 \frac{v_\epsilon}{v_\epsilon-2}  \mathbb{I}_d  + \bm{\delta}_\epsilon^T \bm{\delta}_\epsilon  \frac{2v_\epsilon}{(v_\epsilon-2)^2(v_\epsilon-4)}.
\end{align*}
Hence, the covariance matrix of the observation process $\mathbf{Y}_t$ is given by 
\begin{align*} \notag
& \mathbf{Cov}_{\mathbf{Y}_t|\Psi}\big[ \mathbf{Y}_t\big]  = \mathbf{W} \bigg(   \frac{v_x}{v_x-2}  \mathbb{I}_k +   \bm{\delta}_x^T   \bm{\delta}_x  \frac{2v_x}{(v_x-2)^2(v_x-4)}  \bigg) \mathbf{W}^T + \sigma^2 \frac{v_\epsilon}{v_\epsilon-2}  \mathbb{I}_d  + \bm{\delta}_\epsilon^T \bm{\delta}_\epsilon  \frac{2v_\epsilon}{(v_\epsilon-2)^2(v_\epsilon-4)}.
\end{align*}
The $d\times k$ matrix $\mathbf{W}$ has the following singular value decomposition $\mathbf{W} = \mathbf{M}_{d \times d} \mathbf{D}_{d \times k} \mathbf{N}_{k \times k}$.  If $\bm{\delta}_\epsilon = \bm{\delta}_x = \mathbf{0}$ we can show that
\begin{align*} \notag
& \mathbf{Cov}_{\mathbf{Y}_t|\Psi}\big[ \mathbf{Y}_t\big]  = \mathbf{M} \bigg(  \mathbf{D}\mathbf{D}^T \frac{v_x}{v_x-2}  +  \sigma^2 \frac{v_\epsilon}{v_\epsilon-2}  \mathbb{I}_d  \bigg) \mathbf{M}^T.
\end{align*}
Since the matrix $\mathbf{D}$ has only non-zero diagonal elements, the eigenvectors of the covariance matrix correspond to the left singular-vectors of the matrix $\mathbf{W}$. However, when we relax the assumption about the parameters which control the skewness of the distribution, and allow $\bm{\delta}_\epsilon  \neq \mathbf{0}$  and $ \bm{\delta}_x \neq \mathbf{0}$ then the eigenvectors are again dependent on the outer products of the parameters of skewness. 

\subsubsection{Solution for GSt PPCA with Grouped Heavy Tails of Marginals (Cases 1 \& general GSt PPCA) }
The latent variable and the error terms have the following stochastic representation as in \eqref{eq:latentProcesStochRep_X_epsilon_ind} and their moments are determined as
\begin{align*}\notag
& \mathbb{E}_{\mathbf{X}_t|\Psi}\big[ \mathbf{X}_t \big] = \mathbb{E}_{\mathbf{V}_t|\Psi}\Big[ \mathbb{E}_{\mathbf{X}_t|\mathbf{V}_t, \Psi }\big[ \mathbf{X}_t \big]  \Big] =  \mathbb{E}_{\mathbf{V}_t|\Psi}\big[ \mathbf{V}_t^{-1} \big] \circ \bm{\delta}_x, \\\notag
& \mathbf{Cov}_{\mathbf{X}_t|\Psi}\big[ \mathbf{X}_t \big] =  \mathbb{E}_{\mathbf{V}_t|\Psi}\Big[ \mathbf{Cov}_{\mathbf{X}_t|\mathbf{V}_t, \Psi }\big[ \mathbf{X}_t \big]  \Big] +  \mathbf{Cov}_{\mathbf{V}_t|\Psi}\Big[ \mathbb{E}_{\mathbf{X}_t|\mathbf{V}_t, \Psi }\big[ \mathbf{X}_t \big]  \Big] =  \mathbb{E}_{\mathbf{V}_t|\Psi}\big[ \mathbf{V}_t^{-1}  \big] \circ \mathbb{I}_k +\Big(  \bm{\delta}_x^T  \bm{\delta}_x  \Big) \circ   \mathbf{Cov}_{\mathbf{V}_t|\Psi}\big[ \mathbf{V}_t^{-1} \big] ,
\end{align*}
and 
\begin{align*}\notag
& \mathbb{E}_{\bm{\epsilon}_t\|\Psi}\big[\bm{\epsilon}_t\big] = \mathbb{E}_{\mathbf{U}_t|\Psi}\Big[ \mathbb{E}_{\bm{\epsilon}_t\|\mathbf{U}_t, \Psi }\big[ \bm{\epsilon}_t\ \big]  \Big] =  \mathbb{E}_{\mathbf{U}_t|\Psi}\big[ \mathbf{U}_t^{-1} \big] \circ \bm{\delta}_\epsilon, \\\notag
& \mathbf{Cov}_{\bm{\epsilon}_t |\Psi}\big[ \bm{\epsilon}_t\big] =  \mathbb{E}_{\mathbf{U}_t|\Psi}\Big[ \mathbf{Cov}_{\bm{\epsilon}_t|\mathbf{U}_t, \Psi }\big[ \bm{\epsilon}_t\big]  \Big] +  \mathbf{Cov}_{\mathbf{U}_t|\Psi}\Big[ \mathbb{E}_{\bm{\epsilon}_t |\mathbf{U}_t, \Psi }\big[ \bm{\epsilon}_t \big]  \Big] =  \mathbb{E}_{\mathbf{U}_t|\Psi}\big[ \mathbf{U}_t^{-1}  \big] \circ \mathbb{I}_d +\Big(  \bm{\delta}_\epsilon^T  \bm{\delta}_\epsilon  \Big) \circ   \mathbf{Cov}_{\mathbf{U}_t|\Psi}\big[ \mathbf{U}_t^{-1} \big] ,
\end{align*}
Hence, the covariance matrix of the observation process $\mathbf{Y}_t$ is given by 
\begin{align*} \notag
& \mathbf{Cov}_{\mathbf{Y}_t|\Psi}\big[ \mathbf{Y}_t\big]  = \mathbf{W} \bigg(  \mathbb{E}_{\mathbf{V}_t|\Psi}\big[ \mathbf{V}_t^{-1}  \big] \circ \mathbb{I}_k +\Big(  \bm{\delta}_x^T  \bm{\delta}_x  \Big) \circ   \mathbf{Cov}_{\mathbf{V}_t|\Psi}\big[ \mathbf{V}_t^{-1} \big]  \bigg) \mathbf{W}^T +\mathbb{E}_{\mathbf{U}_t|\Psi}\big[ \mathbf{U}_t^{-1}  \big] \circ \mathbb{I}_d +\Big(  \bm{\delta}_\epsilon^T  \bm{\delta}_\epsilon  \Big) \circ   \mathbf{Cov}_{\mathbf{U}_t|\Psi}\big[ \mathbf{U}_t^{-1} \big] ,
\end{align*}
that is defined by 
\begin{align*}
& \mathbb{E}_{\mathbf{V}_t|\Psi}\big[ \mathbf{V}_t^{-1}  \big]_{1 \times k} = \mathbb{E}_{S_{x,t}|\Psi}\big[T_x( s_{x,t})^{-1}   \big] = \bigg(  \frac{v^1_x}{v^1_x-2}  , \ldots ,  \frac{v^k_x}{v^k_x-2}   \bigg)_{1 \times k } ,  \\
& \mathbb{E}_{\mathbf{U}_t|\Psi}\big[ \mathbf{U}_t^{-1}  \big] _{1 \times d} =  \mathbb{E}_{S_{\epsilon,t}|\Psi}\big[T_\epsilon( s_{\epsilon,t})^{-1}   \big]  =  \bigg(  \frac{v^1_\epsilon}{v^1_\epsilon-2}  , \ldots ,  \frac{v^d_\epsilon}{v^d_\epsilon-2}   \bigg)_{1 \times d }.
\end{align*}
Evaluating this requires solving the following integration problems
\begin{align*}
 \mathbf{Cov}_{\mathbf{V}_t|\Psi}\big[ \mathbf{V}_t^{-1}  \big]_{k \times k } &  =  \int_0^1 \Big( T_x( s_{x,t})^{-1} - \mathbb{E}_{S_{x,t}|\Psi}\big[ T_x( s_{x,t})^{-1} \big]  \Big)^T  \Big( T_x( s_{x,t})^{-1} - \mathbb{E}_{S_{x,t}|\Psi}\big[ T_x( s_{x,t})^{-1} \big]  \Big) \ d s_{x,t},   \\
 \mathbf{Cov}_{\mathbf{U}_t|\Psi}\big[ \mathbf{U}_t^{-1}  \big] _{d \times d}&  =  \int_0^1 \Big( T_x( s_{\epsilon,t})^{-1} - \mathbb{E}_{S_{\epsilon,t}|\Psi}\big[ T_\epsilon( s_{\epsilon,t})^{-1} \big]  \Big)^T  \Big( T_\epsilon( s_{\epsilon,t})^{-1} - \mathbb{E}_{S_{\epsilon,t}|\Psi}\big[ T_\epsilon( s_{\epsilon,t})^{-1} \big]  \Big) \ d s_{\epsilon,t} ,  
\end{align*}
where we applied the Jacobian of the transformation and the  following relations between the probability density functions  for the random variables $S_{\epsilon,t}$ and $S_{x,t}$ defined in Subsection \ref{ssec:PPCA_GStS_ind}, that is
\begin{align*} \notag
& \pi_{\mathbf{U}_t| \Psi} (\mathbf{u}_t) = \pi_{S_{\epsilon,t}| \Psi} (s_{\epsilon,t}) \Big|\frac{\partial T_\epsilon( s_{\epsilon,t})}{\partial s_{\epsilon,t}}\Big|^{-1} = \mathbf{1}_{[0,1]} (s_{\epsilon,t})  \Big|\frac{\partial T_\epsilon( s_{\epsilon,t})}{\partial s_{\epsilon,t}}\Big|^{-1}, \\
& \pi_{\mathbf{V}_t| \Psi} (\mathbf{v}_t) = \pi_{S_{x,t}| \Psi} (s_{x,t}) \Big|\frac{\partial T_x( s_{x,t})}{\partial s_{x,t}}\Big|^{-1} = \mathbf{1}_{[0,1]} (s_{x,t})  \Big|\frac{\partial T_x( s_{x,t})}{\partial s_{x,t}}\Big|^{-1}.
\end{align*}
Let the $d\times k$ matrix $\mathbf{W}$ has the singular decomposition $\mathbf{W} = \mathbf{M}_{d \times d} \mathbf{D}_{d \times k} \mathbf{N}_{k \times k}$.  If $\bm{\delta}_\epsilon = \bm{\delta}_x = \mathbf{0}$ we can show that
\begin{align*} \notag
& \mathbf{Cov}_{\mathbf{Y}_t|\Psi}\big[ \mathbf{Y}_t\big]  = \mathbf{M} \bigg(  \mathbf{D}\mathbf{D}^T \circ \mathbb{E}_{S_{x,t}|\Psi}\big[T_x( S_{x,t})^{-1}   \big]  +  \sigma^2 \mathbb{E}_{S_{\epsilon,t}|\Psi}\big[T_\epsilon( S_{\epsilon,t})^{-1}   \big] \circ  \mathbb{I}_d  \bigg) \mathbf{M}^T.
\end{align*}
Since the matrix $\mathbf{D}$ has only non-zero diagonal elements, the eigenvectors of the covariance matrix correspond to the left singular-vectors of the matrix $\mathbf{W}$. However, when we relax the assumption about the parameters which control the skewness of the distribution, and allow that $\bm{\delta}_\epsilon  \neq \mathbf{0}$  and $ \bm{\delta}_x \neq \mathbf{0}$, then again the eigenvectors are again dependent on the outer products of the parameters of skewness but also are defined by  $\mathbf{Cov}_{\mathbf{U}_t|\Psi}\big[ \mathbf{U}_t^{-1}  \big] $ and $ \mathbf{Cov}_{\mathbf{V}_t|\Psi}\big[ \mathbf{V}_t^{-1}  \big]$ that are not necessary diagonal matrices.

\section{EM Algorithm for GSt PPCA  in the Presence of Missing Data}\label{ssec:EM_PPCA_GStS_ind}
Until now, we assumed that the data set, which we analyse, does not contain any missing observations.  We start with discussing the concept of missingness, and address the questions of  how such characteristics of the data can be incorporated.  To achieve this we introduce the new notation where the random vector $\mathbf{Y}_t $ is partitioned into two subvectors, one which contains observed values $\mathbf{Y}_t ^o$ and the second which indicates missing entries $\mathbf{Y}_t^m$, such that $\mathbf{Y}_t = \big[ \mathbf{Y}_t^o ,\mathbf{Y}^m_t \big]$. We denote $d_o$ as the number of observed elements of the vector $\mathbf{Y}_t$ and $d_m = d - d_o$ the number of missing entries at time $t$. The numbers $d_o$ and $d_m$ can vary over time.  

In the incomplete-data case related sections we denote by $\mathbf{W}_o$ and $\mathbf{W}_m$ the $d_o \times  k$ and $d_m \times k$ non-square submatrices of $\mathbf{W}$ with corresponding rows to the elements of the vector $\mathbf{Y}_t$ which are observed and missing, respectively. In general, by lower index $o$ and $m$, we further refer to the elements of some objects corresponding to observed and missing values of $\mathbf{Y}_t$, respectively. 

Let us define the random vector $\mathbf{R}_t$ which indicates which entries of $\mathbf{Y}_t$ are missing and denotes them by $1$, otherwise $0$. Recall, that a single observation consists of the pair $\left[ \mathbf{Y}_t^o,\mathbf{R}_t \right]$ with distribution parameters $\left[ \Psi, \Theta \right]$ respectively. We assume the parameters to be distinct. The likelihood of parameters is proportional to the conditional probability  $\mathbf{Y}_t^o,\mathbf{R}_t |  \Psi, \Theta  $ that is
\begin{equation}\label{eq:MissingIntegral}
\pi_{\mathbf{Y}_t^o, \mathbf{R}_t|\Psi, \Theta} \left( \mathbf{y}_t^o,\mathbf{r}_t  \right) = \int \pi_{\mathbf{Y}_t^o, \mathbf{Y}_t^m,\mathbf{R}_t |  \Psi, \Theta } \left( \mathbf{y}_t^o, \mathbf{y}_t^m,\mathbf{r}_t \right) d \mathbf{y}^m_t = \int \pi_{\mathbf{R}_t | \mathbf{Y}_t, \Psi, \Theta} \left( \mathbf{r}_t \right) \pi_{ \mathbf{Y}_t| \Psi, \Theta} \left( \mathbf{y}_t \right) d \mathbf{y}^m_t .
\end{equation}
In our study, we assume the pattern of missing data to be MAR - missing at random as defined in  \cite{Little2002}. The assumptions imposes the indicator variable $\mathbf{R}_t$ to be independent of the value of missing data. Then the vector $\mathbf{Y}_t$ which is MAR satisfies $\pi_{\mathbf{R}_t | \mathbf{Y}_t, \Psi}(\mathbf{r}_t ) = \pi_{\mathbf{R}_t | \mathbf{Y}_t^o, \Psi} (\mathbf{r}_t )$
giving
\begin{equation*}
\pi_{\mathbf{Y}_t^o,\mathbf{R}_t | \Psi, \Theta } \left( \mathbf{y}_t^o \right) = \pi_{\mathbf{R}_t | \mathbf{Y}_t^o, \Theta} \left( \mathbf{r}_t  \right) \int \pi_{\mathbf{Y}_t| \Psi} \left( \mathbf{y}_t \right) d \mathbf{y}^m_t  = \pi_{\mathbf{R}_t | \mathbf{Y}_t^o, \Theta} \left( \mathbf{r}_t \right) \pi_{\mathbf{Y}^o_t| \Psi} \left( \mathbf{y}^o_t\right).
\end{equation*}
Under the MAR assumption, the estimation of $\Psi$ via maximum likelihood of the joint distribution  $\mathbf{Y}_t^o,\mathbf{R}_t | \Psi, \Theta$ is equivalent to the maximisation of the likelihood of the marginal distribution $\mathbf{Y}_t^o| \Psi$. Hence, we do not worry about the distribution of the indicator random variable $\mathbf{R}_t$ and the joint distribution of $\mathbf{Y}_t^o$ and $\mathbf{R}_t$. If the assumption about MAR does not hold, one needs to solve the integral from  \eqref{eq:MissingIntegral} in order to maximize the joint likelihood in the corresponding EM algorithm.

\subsection{The E-step and M-step of EM algorithm for GSt PPCA}
The two iterative steps of the EM algorithm for Generalized Skew-t PPCA which jointly maximize the expected log-likelihood of the observed and hidden variables are the following

\begin{enumerate}[label= \textbf{Step \arabic*:},start=1,leftmargin=0.6in]
\item Expectation Step (E-step) \mbox{}\\
We calculate the expectation of the conditional distribution $\mathbf{Y}_{1:N}^m,\mathbf{X}_{1:N},\mathbf{U}_{1:N},\mathbf{V}_{1:N}|\mathbf{Y}_{1:N}^o,\Psi$ over the joint distribution likelihood function of the model \eqref{eq:ppca_model}. Given $N$ realisations of the variables, the expectation is a function of two vectors with static parameters $\Psi = [\mathbf{W}, \bm{\mu}, \sigma^2, \bm{\delta}_\epsilon,\bm{\delta}_x]$ and $\Psi^* = [\mathbf{W}^*, \bm{\mu}^*, \sigma^{*2},\bm{\delta}_\epsilon^*,\bm{\delta}_x^*]$, that is  
\begin{equation*}
Q\big(\Psi, \Psi^* \big) = \mathbb{E}_{\mathbf{Y}_{1:N}^m,\mathbf{X}_{1:N},\mathbf{U}_{1:N},\mathbf{V}_{1:N}|\mathbf{Y}_{1:N}^o,\Psi} \Big[\log \pi_{\mathbf{Y}_{1:N},\mathbf{X}_{1:N},\mathbf{U}_{1:N},\mathbf{V}_{1:N}|\Psi^*} \big(\mathbf{Y}_{1:N},\mathbf{X}_{1:N},\mathbf{U}_{1:N},\mathbf{V}_{1:N}\big)\Big],
\end{equation*}
\item Maximisation Step (M-step)\mbox{}\\
We update the vector of static parameters $\Psi$ via maximisation of the resulting $Q_{GSt,ind}$ function with respect to the vector $\Psi^*$, that is
\begin{equation*}
\hat{\Psi}^* = \argmax_{\Psi^*} Q\big(\Psi, \Psi^* \big) .
\end{equation*}
\end{enumerate}
The $Q$ function of the E-step for the Generalized Skew-t PPCA is provided in Theorem \ref{th:Estep_GStS_missing_ind} whereas the derivation of its maximizers is given in Theorem \ref{th:Mstep_explicitFormulas_missing_ind}. The step of the EM alogirthm for special cases of the GSt PPCA are provided in \cite{Toczydlowska_thesis}. 

\begin{theorem}\label{th:Estep_GStS_missing_ind}
Consider observation vector $\mathbf{Y}_t$ modelled according to  \eqref{eq:ppca_model} with the latent processes $\bm{\epsilon}_t, \ \mathbf{X}_t,  \ \mathbf{U}_t$ and $\mathbf{V}_t$ following the assumptions of GSt PPCA given in \eqref{eq:latentProcesStochRep_U_V_ind} and  \eqref{eq:latentProcesStochRep_X_epsilon_ind}. Given $N$ realisations of the observed entries of the random vector $\mathbf{Y}_t^o$, denoted by $\mathbf{y}_{1:N}^o:= \big[\mathbf{y}_{1}^o, \ldots, \mathbf{y}_{N}^o \big]$, the E-step of the Expectation-Maximisation algorithm for Generalized Skew-t PPCA in the incomplete data setting is specified as follows
\begin{align*} 
Q(\Psi, \Psi^* )& = \frac{1}{\pi_{\mathbf{Y}_{1:N}^o| \Psi}(\mathbf{y}_{1:N}^o)} \sum_{t=1}^N \Big\{ I_1(\mathbf{y}_{t}^o;\Psi,\Psi^*) \prod_{s=1, s \neq t}^N  I_2(\mathbf{y}_s^o;\Psi) \Big\}. 
\end{align*}
The functions $I_1, \ I_2:\mathbb{R}^{d_o} \rightarrow \mathbb{R}$ are defined as 
\begin{align*}\notag
& I_1(\mathbf{y}_{t}^o;\Psi,\Psi^*) :=  \int_0^1 \int_0^1 \tilde{v}(\mathbf{y}_{t}^o,s_{\epsilon,t},s_{x,t};\Psi,\Psi^*) \  m(\mathbf{y}_{t}^o,s_{\epsilon,t},s_{x,t};\Psi) \ d s_{\epsilon,t} \ d s_{x,t} ,\\
& I_2(\mathbf{y}_{t}^o;\Psi) :=  \int_0^1 \int_0^1    m(\mathbf{y}_{t}^o,s_{\epsilon,t},s_{x,t};\Psi) \ d s_{\epsilon,t} \ d s_{x,t}. 
\end{align*}
where $\tilde{v}\big(\mathbf{y}_{t}^o,s_{\epsilon,t},s_{x,t};\Psi,\Psi^*\big):= v\big(\mathbf{y}_{t}^o,T_\epsilon(s_{\epsilon,t}),T_x(s_{x,t});\Psi,\Psi^*\big)$ specified in Lemma~\ref{lemma:integral_w_ind} in Appendix~\ref{appendix:proofs_GStS} and the function $m: \mathbb{R}^{d_o} \times [0,1]^2 \longrightarrow \mathbb{R}$ is given by 
\begin{align*}
m(\mathbf{y}_{t}^o,s_{\epsilon,t},s_{x,t};\Psi) = &  e^{ - \frac{1}{2} \bm{\delta}_x \big( \mathbf{D}_{x,t}^{-1} - \sigma^2  \mathbf{M}_t^{-1}\mathbf{W}^T\mathbf{D}_{\epsilon,t} \mathbf{N}_t^{-1}\mathbf{D}_{\epsilon,t} \mathbf{W}\mathbf{M}_t^{-1}  - \sigma^2  \mathbf{M}_t^{-1} \big)\bm{\delta}_x^T } \pi_{\mathbf{Y}_t^o |S_{\epsilon,t},S_{x,t},\Psi}(\mathbf{y}_t^o) \big(\sigma^2\big)^{\frac{k}{2}}  \Big| \mathbf{D}_{x,t} \Big|^{\frac{1}{2}}  \Big| \mathbf{M}_t\Big|^{- \frac{1}{2}} \Big| \mathbf{N}_t\Big|^{-\frac{1}{2}}.
\end{align*}
The probability density function $\pi_{\mathbf{Y}_t^o |S_{\epsilon,t},S_{x,t},\Psi}(\mathbf{y}_t^o)$ equals to the density function $\pi_{\mathbf{Y}_t^o |T_\epsilon(S_{\epsilon,t}),T_x(S_{x,t}),\Psi}(\mathbf{y}_t^o)$ specified in Lemma~\ref{lemma:tildePi_Ytmissing_Conditional_ind} in Appendix~\ref{appendix:proofs_GStS} where $\mathbf{M}_t = \sigma^2 \mathbf{D}_{x,t} + \mathbf{W}^T \mathbf{D}_{\epsilon,t} \mathbf{W}$ and $ \mathbf{N}_t = \mathbb{I}_d - \mathbf{W}\mathbf{M}_t^{-1} \mathbf{W}^T \mathbf{D}_{\epsilon,t}$.
\end{theorem}
\begin{proof}
The proof of Theorem \ref{th:Estep_GStS_missing_ind} is given in  Subsection \ref{proof:th_Estep_GStS_missing_ind} in Appendix~\ref{appendix:proofs_GStS}.
\end{proof}

\begin{theorem} \label{th:Mstep_explicitFormulas_missing_ind}
The solution to the system of equation $\nabla_{\Psi^*} Q(\Psi, \Psi^* ) = \mathbf{0}$  which determines the maximizers of the function $Q$ from Theorem \ref{th:Estep_GStS_missing_ind} with respect to the parameters $\bm{\mu}^*, \bm{\delta}_\epsilon^* , \bm{\delta}_x^*, \sigma^{*2}$ are given by explicit formulas defined by two-dimensional integration problems on the hypercube $[0,1]^2$ as follows
{\footnotesize 
\begin{align*}
\begin{cases}
& \bm{\mu}^*  = \bigg[  \mathbf{A}_6 (\mathbf{y}_{1:N}^o;\Psi) -  \mathbf{A}_{10} (\mathbf{y}_{1:N}^o;\Psi,\Psi^*)  - \bm{\delta}_\epsilon^*  A_0 (\mathbf{y}_{1:N}^o;\Psi) + \bm{\mu}  \mathbf{A}_{13} (\mathbf{y}_{1:N}^o;\Psi,\Psi^*) + \Big(\bm{\delta}_\epsilon \mathbf{W}  -  \sigma^2 \bm{\delta}_x\Big) \mathbf{A}_{12} (\mathbf{y}_{1:N}^o;\Psi,\Psi^*)^T \bigg]\\
& \hspace{1cm} \times \mathbf{A}_1(\mathbf{y}_{1:N}^o;\Psi)^{-1}, \\\notag
& \bm{\delta}_\epsilon^* =   \bigg[ \mathbf{A}_5(\mathbf{y}_{1:N}^o;\Psi) - \bm{\mu}^*A_0(\mathbf{y}_{1:N}^o;\Psi) - \mathbf{A}_8(\mathbf{y}_{1:N}^o;\Psi)\mathbf{W}^{*T}
 +  \bm{\mu} \mathbf{A}_{11}(\mathbf{y}_{1:N}^o;\Psi)  \mathbf{W}^{*T}+ \big( \bm{\delta}_\epsilon \mathbf{W} -  \sigma^2 \bm{\delta}_x \big) \mathbf{A}_4(\mathbf{y}_{1:N}^o;\Psi) \mathbf{W}^{*T} \bigg] \\
 & \hspace{1cm}\mathbf{A}_2(\mathbf{y}^{o}_{1:N};\Psi)^{-1} ,\\\notag
&\bm{\delta}_x^* = \bigg[ \mathbf{A}_8(\mathbf{y}_{1:N}^o)- \bm{\mu} \mathbf{A}_{11}(\mathbf{y}_{1:N}^o) - \big( \bm{\delta}_\epsilon \mathbf{W} -  \sigma^2 \bm{\delta}_x \big) \mathbf{A}_4(\mathbf{y}_{1:N}^o) \bigg] \mathbf{A}_3 (\mathbf{y}_{1:N},\Psi)^{-1}, \\\notag
& \sigma^{*2} = \frac{1}{d A_0(\mathbf{y}_{1:N}^o;\Psi)} \bigg[  A_{20}(\mathbf{y}_{1:N}^o;\Psi)  + \bm{\mu}^* \mathbf{A}_{1}(\mathbf{y}_{1:N}^o;\Psi) \bm{\mu}^{*T} + \bm{\delta}_\epsilon^{*} \mathbf{A}_{2}(\mathbf{y}_{1:N}^o;\Psi)\bm{\delta}_\epsilon^{*T} + 2   \mathbf{A}_{9}(\mathbf{y}_{1:N}^o;\Psi,\Psi^*) \bm{\mu}^T   \\\notag
& \hspace{1cm} + 2 \Big( \mathbf{A}_{10}(\mathbf{y}_{1:N}^o;\Psi,\Psi^*) -  \mathbf{A}_{6}(\mathbf{y}_{1:N}^o;\Psi)- \bm{\mu} \mathbf{A}_{13}(\mathbf{y}_{1:N}^o;\Psi,\Psi^*)\Big) \bm{\mu}^{*T} - 2  A_{21}(\mathbf{y}_{1:N}^o;\Psi,\Psi^*) +   \sigma^2 \tr \Big\{\mathbf{A}_{12}(\mathbf{y}_{1:N}^o;\Psi)^T \mathbf{W}^* \Big\}  \\\notag
&\hspace{1cm}  +   \tr \Big\{\mathbf{A}_{17}(\mathbf{y}_{1:N}^o;\Psi)^T \mathbf{W}^* \Big\} - \tr \Big\{\mathbf{A}_{18.1}(\mathbf{y}_{1:N}^o;\Psi)^T \mathbf{W}^* \Big\} - \tr \Big\{\mathbf{A}_{18.2}(\mathbf{y}_{1:N}^o;\Psi)^T \mathbf{W}^* \Big\} + \tr \Big\{\mathbf{A}_{19}(\mathbf{y}_{1:N}^o;\Psi)^T \mathbf{W}^* \Big\}   \\\notag
&\hspace{1cm} + 2 \Big( \bm{\mu}^{*}  A_{0}(\mathbf{y}_{1:N}^o;\Psi)-  \mathbf{A}_{5}(\mathbf{y}_{1:N}^o;\Psi) + \mathbf{A}_{8}(\mathbf{y}_{1:N}^o;\Psi) \mathbf{W}^{*T} -  \bm{\mu}  \mathbf{A}_{11}(\mathbf{y}_{1:N}^o;\Psi) \mathbf{W}^{*T} \Big) \bm{\delta}_\epsilon^{*T} \\\notag
&\hspace{1cm}+ 2 \Big(  \mathbf{A}_{7}(\mathbf{y}_{1:N}^o;\Psi,\Psi^*)  -  \bm{\mu}^{*} \mathbf{A}_{12}(\mathbf{y}_{1:N}^o;\Psi,\Psi^*) -   \bm{\delta}_\epsilon^{*}  \mathbf{W}^{*} \mathbf{A}_{4}(\mathbf{y}_{1:N}^o;\Psi) \Big) \big( \bm{\delta}_\epsilon \mathbf{W} - \sigma^2 \bm{\delta}_x \big)^T  \bigg], \\\notag
& f(\mathbf{W}^*;\Psi,\Psi^*) = \frac{\partial Q(\Psi,\Psi^*)}{\partial \mathbf{W}^*} = \mathbf{0}.
\end{cases}
\end{align*}
}
The solutions with respect to  $\bm{\mu}^*, \ \bm{\delta}_\epsilon^*$ and $\sigma^{*2}$ are linear function of the parameter $\mathbf{W}^*$. The maximizer of the function $Q$ with respect to $\mathbf{W}^*$ is determined by the function $f: \mathbb{R}^{d \times k} \rightarrow  \mathbb{R}^{d \times k} $ 
\begin{align*}
f(\mathbf{W}^*;\Psi,\Psi^*) &=  \mathbf{A}_{14} (\mathbf{y}_{1:N}^o;\Psi) - \mathbf{A}_{15} (\mathbf{y}_{1:N}^o;\Psi,\Psi^*) - \bm{\delta}_\epsilon^{*T} \mathbf{A}_{8} (\mathbf{y}_{1:N}^o;\Psi)- \mathbf{A}_{16} (\mathbf{y}_{1:N}^o;\Psi,\Psi^*)  + \bm{\delta}_\epsilon^{*T}\bm{\mu} \mathbf{A}_{11} (\mathbf{y}_{1:N}^o;\Psi) \\\notag
&\hspace{0.1cm}+ \bm{\delta}_\epsilon^{*T} \big(\bm{\delta} \mathbf{W} - \sigma^2 \mathbf{W} \big) \mathbf{A}_{4} (\mathbf{y}_{1:N}^o;\Psi)- \mathbf{A}_{17} (\mathbf{y}_{1:N}^o;\Psi,\Psi^*) +  \mathbf{A}_{18.1} (\mathbf{y}_{1:N}^o;\Psi,\Psi^*) + \mathbf{A}_{18.2} (\mathbf{y}_{1:N}^o;\Psi,\Psi^*) \\\notag
&\hspace{0.1cm}- \mathbf{A}_{19} (\mathbf{y}_{1:N}^o;\Psi,\Psi^*) - \sigma^2 \mathbf{A}_{12} (\mathbf{y}_{1:N}^o;\Psi,\Psi^*)^T.
\end{align*}
The function $f$ is linear with respect to the parameters $\bm{\mu}^*$ and $\bm{\delta}_\epsilon^*$. The two-dimensional integrals $\mathbf{A}_i$ on the hypercube $[0,1]^2$, for $i \in \Big\{ 0, \ldots, 21 \Big\}$, are defined in Subsection~\ref{proof:cor_Mstep_explicitFormulas_missing_ind} in Appendix~\ref{appendix:proofs_GStS}.
\end{theorem}
\begin{proof}
The proof of Theorem \ref{th:Mstep_explicitFormulas_missing_ind} is given in Subsection~\ref{proof:cor_Mstep_explicitFormulas_missing_ind} in Appendix~\ref{appendix:proofs_GStS}.
\end{proof}

\section{Comments on Influence Function of GSt PPCA}\label{sec:IF}
Here we explore our proposed methodology from the perspective of robust estimation rather than a model selection. The introduced models provide estimators of parameters, such as mean and covariance, that contain additional flexibility to accommodate characteristics of the data such as skewness or various patterns of the marginal tail dependence.  The proposed estimators are obtained as solutions to the estimation equations that may be seen as 'distorted' in comparison to the equations for the same parameters in the standard PPCA.  We argue that this new class of estimators is more robust than the standard PPCA frameworks.  

In general, to verify this conjecture, one may analyse a non-linear system of equations with weighting that specify contributions of sample points to the formulation of an estimator and study the properties of the robustness imposed by these weights. However, our challenge is that for the new class of the estimators under GSt PPCA model, the weighting functions do not have analytic closed forms. Instead, we propose to study the characteristics of robustness by the notion of the influence functions and study them numerically in three different ways: by the asymptotic bias of an estimator, the asymptotic variance of an estimator, ie, the precision of the estimation, and by the sensitivity of an estimator to the effect of outliers. 

As it is shown in the final part of this section, the GSt PPCA is characterized by the highest robustness in comparison to the standard PPCA methods according to these three measures.  If one suspects that the data might reveal characteristics captured by the GSt PPCA family of models, the best choice is to use these PPCA frameworks as we show the significant loss of accuracy and robustness for simpler standard approaches such as Gaussian PPCA of \cite{TippingBishop1999} or Student-t PPCA of \cite{RidderFranc2003} or\cite{KhanDellaert2004}, when these characteristics appear in the data. On the other hand, if the data follows simpler distributions, the class of GSt PPCA methods is flexible enough to also accommodate simple Guassian and student-t structures. 

\subsection{Solutions to the Estimating Equations}
Under the Gaussian PPCA model from Section \ref{sec:GaussianPPCA}, we can specify  the marginal distribution of the observation vector $\mathbf{Y}_t$ to be Gaussian with the mean vector $\bm{\mu}$ and the covariance matrix $\mathbf{C} = \mathbf{W} \mathbf{W}^T + \sigma^2 \mathbb{I}_d$. The MLE estimates using this marginalised likelihood are then
\begin{align*}
\hat{\bm{\mu}}, \ \hat{\mathbf{C}} =  \argmax_{\bm{\mu}, \mathbf{C}} \sum_{t= 1}^N \log \big( \pi_{\mathbf{Y}_t|\Psi}(\mathbf{y}_{t}) \big) = \argmin_{\bm{\mu}, \mathbf{C}} \bigg\{N \log \big| \mathbf{C} \big| + \sum_{t= 1}^N \big( \mathbf{y}_t - \bm{\mu} \big) \mathbf{C}^{-1} \big( \mathbf{y}_t - \bm{\mu} \big)^T \bigg\}.
\end{align*}
After differentiating the objective function to obtain the maximum one obtains the following system of equations
\begin{align*}
\begin{cases}
& - N \bm{\mu} + \sum_{t= 1}^N \mathbf{y}_t = \mathbf{0} ,\\
& N\mathbf{C}^{-1} - \sum_{t = 1}^N \mathbf{C}^{-1} \big( \mathbf{y}_t - \bm{\mu} \big)^T \big( \mathbf{y}_t - \bm{\mu} \big) \mathbf{C}^{-1} = \mathbf{0},
\end{cases} \Longleftrightarrow
\begin{cases}
& \hat{\bm{\mu}}  = \frac{1}{N} \sum_{t= 1}^N \mathbf{y}_t, \\
& \hat{\mathbf{C}} = \frac{1}{N} \sum_{t = 1}^N \big( \mathbf{y}_t - \hat{\bm{\mu} }\big)^T \big( \mathbf{y}_t - \hat{\bm{\mu}} \big)^T. 
\end{cases}
\end{align*}
The MLE estimators of the Gaussian PPCA model can be obtained in closed form however,  they assume that all samples are of the same importance and hence, an outlying observation contributes to the construction of the estimators equally.  
This outcome results in the information captured by provided eigen vectors losing the intended interpretation (statistical summary) if the data is corrupted even by a single observation. In order to introduce a general notion of observation specific weight, the generalized version of MLE by M-estimators for independently distributed data was proposed in \cite{Maronna1976} and \cite{Huber2009} which defined by the following system of normal equations
\begin{equation*}
\begin{cases}
& \sum_{t = 1}^N \frac{ \psi(D_t)}{D_t}  \mathbf{C}^{-\frac{1}{2}} \left(\mathbf{y}_t - \bm{\mu} \right)   = \mathbf{0}, \\
& \sum_{t = 1}^N \frac{ \psi(D_t)}{D_t}   \mathbf{C}^{-\frac{1}{2}}  \left(\mathbf{y }_t  - \bm{\mu}\right)^T \left(\mathbf{y}_t - \bm{\mu} \right) \mathbf{C}^{-\frac{1}{2}}   = \mathbf{0}.
\end{cases}
\end{equation*}
where $D^2_t = \left( \mathbf{y}_t - \bm{\mu} \right) \mathbf{C}^{-1} \left( \mathbf{y}_t - \bm{\mu} \right)^T$ is a Mahalanobis distance of a sample point $\mathbf{y}_t$ and $\psi:\mathbb{R}_+  \rightarrow \mathbb{R}$.  We may look at estimators of the mean and  covariance obtained by Student-t PPCA of \cite{RidderFranc2003} from the same perspective as the solutions to the system of equations given by
\begin{align*}
\begin{cases}
& \sum_{t=1}^N \frac{\mathbf{C}^{-1} \big( \mathbf{y}_t - \bm{\mu} \big) }{v +  \big( \mathbf{y}_t - \bm{\mu} \big) \mathbf{C}^{-1} \big( \mathbf{y}_t - \bm{\mu} \big)^T } = \mathbf{0}, \\
& N \mathbf{C}^{-1} -(v + d) \sum_{t = 1}^N \frac{\mathbf{C}^{-1} \big( \mathbf{y}_t - \bm{\mu} \big)^T \big( \mathbf{y}_t - \bm{\mu} \big) \mathbf{C}^{-1}}{v +  \big( \mathbf{y}_t - \bm{\mu} \big) \mathbf{C}^{-1} \big( \mathbf{y}_t - \bm{\mu} \big)^T } = \mathbf{0},
\end{cases} \Longleftrightarrow \begin{cases}
& \sum_{t=1}^N \frac{ \mathbf{y}_t - \bm{\mu} }{v +  \big( \mathbf{y}_t - \bm{\mu} \big) \mathbf{C}^{-1} \big( \mathbf{y}_t - \bm{\mu} \big)^T } = \mathbf{0}, \\
& \mathbf{C} = \frac{v + d}{N} \sum_{t = 1}^N \frac{\big( \mathbf{y}_t - \bm{\mu} \big)^T \big( \mathbf{y}_t - \bm{\mu} \big)}{v +  \big( \mathbf{y}_t - \bm{\mu} \big) \mathbf{C}^{-1} \big( \mathbf{y}_t - \bm{\mu} \big)^T }.
\end{cases}
\end{align*}
and remark on the special weighting function that the method introduces to down weight outliers. If $v \rightarrow 0$, the estimator of the covariance matrix under Student-t PPCA corresponds to the Tyler Estimator of \cite{Tyler1987a} .

The above examples of the estimation equation for the PPCA model's parameters explicitly state the functional formulation of the weighting function. It is not a case for the family of GSt models. The analogous system of normal equations for the mean and covariance matrix of random vectors which  follows the GSt PPCA model, is given by the optimisation problem 
\begin{align*}
\hat{\bm{\mu}}, \ \hat{\mathbf{W}}, \hat{\sigma}^2, \hat{\bm{\delta}}_\epsilon, \hat{\bm{\delta}}_x = \argmin_{\bm{\mu}, \ \mathbf{W}, \ \sigma^2, \bm{\delta}_\epsilon, \bm{\delta}_x} \bigg\{ \log \big( \pi_{\mathbf{Y}_{1:N}|\bm{\mu}, \ \mathbf{W}, \ \sigma^2, \bm{\delta}_\epsilon, \bm{\delta}_x}(\mathbf{y}_{1:N}) \big) \bigg\},
\end{align*}
where the loglikelihood of the model is formulated as
\begin{align*}
\log L(\Psi,\mathbf{y}_{1:N} ) = \sum_{t = 1}^N \log \pi_{\mathbf{Y}_t|\Psi} (\mathbf{y}_t) = \sum_{t = 1}^N  \log \bigg(  \int_0^1 \int_0^1    m(\mathbf{y}_{t},s_{\epsilon,t},s_{x,t};\Psi) \ d s_{\epsilon,t} \ d s_{x,t} \bigg),
\end{align*}
since by using Lemma~\ref{lemma:tildePi_Xt_Yt_ind}  and  Lemma~\ref{lemma:tildePi_Ytmissing_Conditional_ind} in the Appendix~\ref{appendix:proofs_GStS}, the $\log \pi_{\mathbf{Y}|\Psi} (\mathbf{y})$ is expressed as 
\begin{align}\label{eq:GSt_log_pdf}
\log \pi_{\mathbf{Y}|\Psi} (\mathbf{y}_t)    =   \log \bigg(  \int_0^1 \int_0^1    m(\mathbf{y}_{t},s_{\epsilon,t},s_{x,t};\Psi) \ d s_{\epsilon,t} \ d s_{x,t} \bigg),
\end{align} 
for the function $m: \mathbb{R}^{d} \times [0,1]^2 \longrightarrow \mathbb{R}$ defined in Theorem \ref{th:Estep_GStS_missing_ind} as
\begin{align*}
& m(\mathbf{y}_{t},s_{\epsilon,t},s_{x,t};\Psi) =  e^{ - \frac{1}{2} \bm{\delta}_x \big( \mathbf{D}_{x,t}^{-1} - \sigma^2  \mathbf{M}_t^{-1}\mathbf{W}^T\mathbf{D}_{\epsilon,t} \mathbf{N}_t^{-1}\mathbf{D}_{\epsilon,t} \mathbf{W}\mathbf{M}_t^{-1}  - \sigma^2  \mathbf{M}_t^{-1} \big)\bm{\delta}_x^T } \pi_{\mathbf{Y}_t |S_{\epsilon,t},S_{x,t},\Psi}(\mathbf{y}_t) \big(\sigma^2\big)^{\frac{k}{2}}  \Big| \mathbf{D}_{x,t} \Big|^{\frac{1}{2}}  \Big| \mathbf{M}_t\Big|^{- \frac{1}{2}} \Big| \mathbf{N}_t\Big|^{-\frac{1}{2}}.
\end{align*}
Therefore, the corresponding log-likelihood function is not explicitly stated as it is defined by two-dimensional integration problems.  Consequently, it is not easy to find the explicit functional form for weight functions which modify the system of estimation equations in the GSt PPCA model. However, we still can indirectly study if this new formulations reduces the sensitivity of the estimators to the effect of outliers. We can achieve that by using an influence function, or more precisely, its approximation, and asses local robustness of an estimator as well as determine its asymptotic variance. 

\subsection{The Influence Function}
We follow the definition of the influence function from \cite{Hampel1986}. Let $\Omega$ be an open convex subset of $\mathbb{R}$.  Let $F$ be a cumulative distribution function parametrized by $\Psi \in \Omega$  of  a $d$-dimensional random vector $\mathbf{Y} \sim F$. We observe $N$ iid (independent and identically distributed) realisations of $\mathbf{Y}$, denoted by $\mathbf{y}_1,\ldots, \mathbf{y}_N$,  that determine its empirical cumulative distribution function, $\hat{F}_N$.  Let $T$ be  a mapping  from the data space defined by $F$ to a parameter space $\Omega$ that is Fisher consistent, that is
$$
T(F) = \Psi \text{ and } T(\hat{F}_N) = \hat{\Psi}_N.
$$
For instance, if we are interested in finding an estimator of the expectation of $\mathbf{Y}$, the parameter of interest is $\Psi =  \bm{\mu} = \mathbb{E}_F \mathbf{Y} $. Then $T(F) = \int_\mathcal{Y} \mathbf{y} dF(\mathbf{y})$.  

The influence function defined in \cite{Hampel1986} is a Gateaux derivative of an estimator $T$ at the measure $F$ in the direction of a local point $\mathbf{y} \in \mathcal{Y}$ expressed by
 \begin{equation*}
 IF(\mathbf{y},T,F) = \lim_{\epsilon \rightarrow 0 } \frac{T \left( (1-\epsilon) F + \epsilon \Delta_\mathbf{y}\right) - T (F)}{\epsilon},
\end{equation*}
where $\Delta_{\mathbf{y}}$ is a probability measure which puts mass $1$ at the point $\mathbf{y}$.  Practically, an influence function can be described as a measure of the effect of an infinitesimal contamination at the point $\mathbf{y}$ on the estimator that is standardized by the mass of the contamination. Furthermore, we assume that for the considered class of estimators the distribution function $F$ satisfies regularity conditions such as differentiability at $\Psi$, $\mathbb{E}_F\big[T(F)^2\big] < \infty$ and
\begin{equation*}
 \int IF(\mathbf{y},T,F) d F (\mathbf{y})  = \mathbb{E}_F \big[ IF(\mathbf{y},T,F) \big] = \mathbf{0}.
\end{equation*}
Under these conditions, we may formulate the expressions for the bias and asymptotic asymptotic variance of an estimator in terms of its influence function. Let us first derive the representation of $T \left( (1-\epsilon) F + \epsilon \Delta_\mathbf{y}\right) $ using its Tylor series expansion around $T (F)$ as shown in \cite{Hampel1974}, \cite{Hampel1986} , that is
\begin{equation*}
T \left( (1-\epsilon) F + \epsilon \Delta_\mathbf{y}\right)   = T (F) + \epsilon  IF(\mathbf{y},T,F)  + o_p(1),
\end{equation*}
where $ o_p(1)$ denotes the convergence of the reminder to zero in probability. This representation can be seen from a more general perspective as the von Mises expansion as in \cite{Fernholz1983}, \cite{vaart1998}. If we study the expansion of the functional $T$ at the empirical distribution $F_N$ around the true distribution $F$, we have
\begin{equation*}
T \left( F_N\right) =   T \left( F \right) + \int IF(\mathbf{y},T,F) d(F_N - F) (\mathbf{y}) + o_p(1) =  \int IF(\mathbf{y},T,F) d F_N(\mathbf{y}) + o_p(1) = \frac{1}{N}\sum_{t = 1}^N IF(\mathbf{y}_t,T,F)  + o_p(1).
\end{equation*}
The Cramer–Rao inequality relates the asymptotic efficiency to the influence function that determines the asymptotic variance of an estimator defined by $T$,
\begin{equation}\label{eq:IF_var}
\sqrt{N} (\hat{\Psi}_N - \Psi ) \rightarrow^{d} \mathcal{N} \Big(0, \mathbb{E} _\mathbf{Y}\big[ IF(\mathbf{Y},T,F) ^2 \big] \Big).
\end{equation}
Influence function can indicate some useful properties of an estimator, ie evaluation of the gross error sensitivity defined as 
\begin{equation}\label{eq:if_gross}
\gamma^*(T,F) = \sup_{\mathbf{y}  \in \mathcal{Y}} \Big| IF(\mathbf{y},T,F) \Big|,
\end{equation}
as a measure of maximum sensitivity of an estimator to local contamination, or a local-shift sensitivity expressed as 
\begin{equation}\label{eq:if_shift}
\lambda^*(T,F) = \sup_{\mathbf{y}_1 ,\mathbf{y}_2  \in \mathcal{Y} } \Big| \frac{IF(\mathbf{y}_1,T,F) - IF(\mathbf{y}_1,T,F)}{\mathbf{y}_1 - \mathbf{y}_2} \Big|,
\end{equation}
that measures the effect of shifting a single observation to a different value in the estimation process.

\subsection{Influence Function for MLE}
The estimation via the maximum likelihood principle can be generalised as defined in \cite{Hampel1986} ,  \cite{Maronna1976} or \cite{Huber2009}. The problem of finding an estimator according to some measure of target $\rho : \mathcal{Y} \times \Omega \rightarrow \mathbb{R} $ that is differentiable and satisfies Leibniz integral rule, can be specified as 
\begin{equation*}
\Psi^* = \argmin_{\Psi \in \Omega} \int_\mathcal{Y} \rho(\mathbf{y},\Psi) dF(\mathbf{y})  ,
\end{equation*}
and is found by solving the normal equations given by
\begin{equation*}
\int_\mathcal{Y}  \varphi(\mathbf{y},\Psi) dF(\mathbf{y})= 0,
\end{equation*}
for $ \varphi(\mathbf{y},\Psi)  =  \frac{\partial \rho(\mathbf{y},\Psi)  }{\partial \Psi} $. As noted in \cite{Hampel1986} , if $\varphi$ is strictly monotone, the problem has a unique solution.  Given this notation and conditions, the influence function of the parameters $\Psi$ can be expressed as
\begin{equation*}
IF(\mathbf{y},\Psi,F) =    - \varphi (\mathbf{y},\Psi) \left(  \frac{\partial  \mathbb{E}_{\mathbf{Y}|\Psi} \big[\varphi(\mathbf{Y},\Psi) \big] }{\partial \Psi} \right)^{-1} .
\end{equation*}
We remark that this formulation of an influence function implies that its expectation is zero, since
\begin{equation*}
\mathbb{E} _{\mathbf{Y}|\Psi}\big[ IF(\mathbf{Y},T,F) \big] = - \mathbb{E} _{\mathbf{Y}|\Psi}\big[ \varphi (\mathbf{y},\Psi) \big] \left(  \frac{\partial  \mathbb{E}_{\mathbf{Y}|\Psi} \big[\varphi(\mathbf{Y},\Psi) \big] }{\partial \Psi} \right)^{-1} = \mathbf{0}.
\end{equation*}
In the maximum log-likelihood estimation for independently distributed data, $ \rho(\mathbf{y},\Psi) =  \log \pi_{\mathbf{Y}|\Psi} (\mathbf{y}) $ and $ \varphi(\mathbf{y},\Psi)  = \frac{\partial \log \pi_{\mathbf{Y}|\Psi} (\mathbf{y}) }{\partial \Psi}  $. The estimators of parameters in GSt PPCA model family are MLE estimators for the logarithm of the probability density function specified defined by the two-dimensional integration over the function $m$ as shown in \eqref{eq:GSt_log_pdf}.  

Given differentiability classes of the function $m$ shown in Lemma~\ref{lemma:m_smoothness_ind}, we can apply products of the Leibniz integral rule to swap the order of differentiating and integrating, both to function $m$ and its first derivative to specify the terms required to calculate the influence function of the PPCA model and the influence function of $\Psi$ is expressed as 
\begin{equation*}
 IF(\mathbf{y},T,F) = - \frac{\partial \log \pi_{\mathbf{Y}|\Psi} (\mathbf{y}) }{\partial \Psi}   \Bigg( \mathbb{E}_{\mathbf{Y}|\Psi} \bigg[  \frac{\partial^2 \log \pi_{\mathbf{Y}|\Psi} (\mathbf{y}) }{\partial \Psi^2}   \bigg] \Bigg)^{-1},
\end{equation*}
where 
\begin{align*}
& \frac{\partial \log \pi_{\mathbf{Y}|\Psi} (\mathbf{y})}{\partial \Psi}   = \frac{ \int_0^1 \int_0^1   \frac{\partial  }{\partial \Psi} m(\mathbf{y}_{t},s_{\epsilon,t},s_{x,t};\Psi)\ d s_{\epsilon,t} \ d s_{x,t}  }{ \int_0^1 \int_0^1    m(\mathbf{y}_{t},s_{\epsilon,t},s_{x,t};\Psi) \ d s_{\epsilon,t} \ d s_{x,t} },
& \frac{\partial^2}{\partial \Psi^2 } \log \pi_{\mathbf{Y}|\Psi} (\mathbf{y})   = \frac{\partial }{\partial \Psi} \bigg( \frac{ \int_0^1 \int_0^1   \frac{\partial  }{\partial \Psi} m(\mathbf{y}_{t},s_{\epsilon,t},s_{x,t};\Psi)\ d s_{\epsilon,t} \ d s_{x,t}  }{ \int_0^1 \int_0^1    m(\mathbf{y}_{t},s_{\epsilon,t},s_{x,t};\Psi) \ d s_{\epsilon,t} \ d s_{x,t} }\bigg).
\end{align*} 
Which consequently results in the formulation of the asymptotic variance of an estimator for GSt PPCA model as follows
\begin{equation*}
\mathbb{E}_{\mathbf{Y}|\Psi}\Big[ IF(\mathbf{y},T,F)^2 \Big] = \Bigg( \mathbb{E}_{\mathbf{Y}|\Psi} \bigg[  \frac{\partial^2 \log \pi_{\mathbf{Y}|\Psi} (\mathbf{y}) }{\partial \Psi^2}   \bigg] \Bigg)^{-1}\mathbb{E}_{\mathbf{Y}|\Psi}\Bigg[ \bigg( \frac{\partial \log \pi_{\mathbf{Y}|\Psi} (\mathbf{y}) }{\partial \Psi}  \bigg)^T \frac{\partial \log \pi_{\mathbf{Y}|\Psi} (\mathbf{y}) }{\partial \Psi}  \Bigg] \Bigg( \mathbb{E}_{\mathbf{Y}|\Psi} \bigg[  \frac{\partial^2 \log \pi_{\mathbf{Y}|\Psi} (\mathbf{y}) }{\partial \Psi^2}   \bigg] \Bigg)^{-1}.
\end{equation*}
Since $\mathbb{E} _{\mathbf{Y}|\Psi}\big[ IF(\mathbf{Y},T,F) \big] = \mathbf{0}$, the expression $ \mathbb{E}  _{\mathbf{Y}|\Psi}\big[ IF(\mathbf{Y},T,F) ^2 \big] $ represents the variance of an influence function under the distribution $F$ of $\mathbf{Y}$. Under the assumptions of the consider family of models, we may use the Chebyshev's inequality to asses the probability of high values of an influence function
\begin{equation*}
\mathbb{P} \bigg( \Big|  IF(\mathbf{Y},T,F) \Big|  \geq a\bigg)   \leq \frac{\mathbb{E}  _{\mathbf{Y}|\Psi}\big[ IF(\mathbf{Y},T,F) ^2 \big]}{a^2},
\end{equation*}
for some $a >0$. Therefore, influence functions with lower variance are less likely to exceed level $a$ and consequently, the corresponding estimators are less likely to be sensitive to the effect of a local contamination.

\subsection{Assessing Numerically Properties of Influence Function  for GSt PPCA family}\label{ssec:if_sim}
In the classical theory of statistical inference, we assume knowledge of a model that characterises the data generating process. Based on this information, we derive estimators of parameters of interest and quantify their properties. However, the robust theory argues that the perfect model is often not known and if it is known, it will be an approximation of reality that is given by limited sample size. As argued in \cite{Hampel1986}, robust theory considers the distribution of estimators not only under the true model but also under other probability distributions. This short study can be seen from the perspective that we have a set of realisations of a random process. We assume its distribution to follow some statistical model and derive estimators of its parameters.  The estimators are just a function of the observed sample set, the introduced functional on its empirical distribution.  

Given this interpretation, we illustrate an exercise that shows the behaviour of estimators given two scenarios:  
\begin{description}
\item[S1] the true model that generates the data is consistent with the assumption on the distribution that is used to derive the estimator;
\item[S2] the true model that generates the data is not consistent with the assumption on the distribution that is used to derive the estimator;
\end{description}
Therefore, the scenario (S1) and (S2) assume no misspecification and misspecification of the model that characterizes an observation vector, respectively. 

We derive the influence function for the parameters $\sigma^2$ and $\mathbf{W}$ for $3$ PPCA frameworks: the standard Gaussian PPCA of \cite{TippingBishop1999a}, the standard Student-t PPCA of \cite{RidderFranc2003} and the Grouped-t GSt PPCA as Special Case 1 of the GSt PPCA family.  In scenario (S1), we calculate influence functions on the datasets that are consistent with the distributions of the PPCA models. 

In scenario (S2), we assume that the observation vector follows the model of Grouped-t GSt PPCA and study what happens if the data that we observe reflected more complex structure than assumed in the estimation. We show that fitting naive models such as Gaussian PPCA and Student-t PPCA, that do not have the structure of Grouped-t GSt PPCA or even Student-t GSt PPCA,  would lose efficiency and robustness in terms of the asymptotic variance and bias as well as the measures of sensitivity, the gross errors and local-shift sensitivities. We show the impact of the separation of the tail effect or grouped multiple-degree-of-freedom structures that define patterns of marginal heavy-tail distributions on the loss of robustness across three PPCA models.

\subsubsection{Set-up of Simulation Study}
To examine numerically the robustness of the estimators, we conduct the following simulation study. We generate $M = 1000$ realisations of the $d$-dimensional random vector $\mathbf{Y}  \sim F$, $\mathbf{y}_1, \ \ldots, \ \mathbf{y}_M$. Depending on the scenario, the distribution $F$ refers to the model of the Gaussian PPCA,  Student-t PPCA or Grouped-t PPCA.  We use the stochastic approximations of the expectations $\mathbb{E}_{\mathbf{Y}|\Psi}  \Big[ \log \pi_{\mathbf{Y}|\Psi} (\mathbf{Y})  \Big]$, $\mathbb{E}_{\mathbf{Y}|\Psi}  \Big[ \frac{\partial}{\partial \Psi} \log \pi_{\mathbf{Y}|\Psi} (\mathbf{Y})  \Big]$ and $\mathbb{E}_{\mathbf{Y}|\Psi}  \Big[ \frac{\partial^2}{\partial \Psi \partial \Psi} \log \pi_{\mathbf{Y}|\Psi} (\mathbf{Y})  \Big]$ by the $M$ realisations of $\mathbf{Y}$ formulated as combined Monte Carlo-Quadrature approximations (2-d quadrature for the inner integrals).
\begin{align*}\notag
& \mathbb{E}_{\mathbf{Y}|\Psi}  \Big[ \log \pi_{\mathbf{Y}|\Psi} (\mathbf{Y})  \Big] \approx \sum_{t = 1}^M  \log \Bigg(  \int_0^1 \int_0^1    m(\mathbf{y}_{t},s_{\epsilon,t},s_{x,t};\Psi) \ d s_{\epsilon,t} \ d s_{x,t} \Bigg),\\\notag
& \mathbb{E}_{\mathbf{Y}|\Psi}  \Big[ \frac{\partial}{\partial \Psi }\log \pi_{\mathbf{Y}|\Psi} (\mathbf{Y})  \Big] \approx \sum_{t = 1}^M \Bigg\{  \frac{ \int_0^1 \int_0^1   \frac{\partial  }{\partial \Psi} m(\mathbf{y}_{t},s_{\epsilon,t},s_{x,t};\Psi)\ d s_{\epsilon,t} \ d s_{x,t}  }{ \int_0^1 \int_0^1    m(\mathbf{y}_{t},s_{\epsilon,t},s_{x,t};\Psi) \ d s_{\epsilon,t} \ d s_{x,t} } \Bigg\},\\
& \mathbb{E}_{\mathbf{Y}|\Psi}  \Big[ \frac{\partial^2}{\partial \Psi^2 }\log \pi_{\mathbf{Y}|\Psi} (\mathbf{Y})  \Big] \approx \sum_{t = 1}^M \Bigg\{\frac{\partial }{\partial \Psi} \bigg( \frac{ \int_0^1 \int_0^1   \frac{\partial  }{\partial \Psi} m(\mathbf{y}_{t},s_{\epsilon,t},s_{x,t};\Psi)\ d s_{\epsilon,t} \ d s_{x,t}  }{ \int_0^1 \int_0^1    m(\mathbf{y}_{t},s_{\epsilon,t},s_{x,t};\Psi) \ d s_{\epsilon,t} \ d s_{x,t} }\bigg)\Bigg\} .
\end{align*}
We assume as well that we know the true parameters of the PPCA models and we calculate their influence functions under this assumption.  The random vector $\mathbf{Y}$ has dimensionality $d=3$, the dimentionality of the latent vector $\mathbf{X}$ is $k=2$, $\bm{\mu} = \mathbf{0}$, with $\sigma^2 = 0.1$ and the $d\times k$ projection matrix
\begin{equation*}
\mathbf{W}_{3 \times 2} = \begin{bmatrix}
w_1 & w_2 \\
w_3 & w_4 \\
w_5 & w_6
\end{bmatrix} =  \begin{bmatrix}
0.3& 1 \\
1.23& 0.8 \\
0.021 & 0.98
\end{bmatrix}.
\end{equation*}
For the standard Student-t PPCA we examine its influence function across the grid of degrees of freedom $\nu \in \Big\{4, 10,20,100 \Big\}$.  On the other hand, the random vectors $\bm{\epsilon}_t$ and $\mathbf{X}_t$ under Grouped-t GSt PPCA model have vectors of degrees of freedom $\bm{\nu}_\epsilon \in  \big\{4,100 \big\}^3$ and  $\bm{\nu}_\epsilon \in  \big\{4,100 \big\}^2$.  This assumption allows us to have different assumptions on heavy tails per marginal and, due to independence of $\bm{\epsilon}_t$ and $\mathbf{X}_t$, to separate the effect of heavy-tails between the new representation and perturbation. We remark that if all marginals of $\bm{\epsilon}_t$ and $\mathbf{X}_t$ have the same profiles of heavy tails, the Grouped-t GSt PPCA collapses to Student-t GSt PPCA.  Also, the data that follows distributions implied by the Student-t PPCA or Grouped-t GSt PPCA models with the degrees of freedom around $100$ may be seen analogous to the cases of normally distributed.

\subsubsection{Asymptotic Variance of Estimators}
First, we examine the robustness of the PPCA models in terms of the asymptotic variance of their estimators. The study reveals how misspecification of a true distribution of an observation process impacts on an asymptotic variance of the obtained estimators, that its, the precision of the estimation. We show that according to this criterion of robustness, the GSt PPCA model is the most efficient either under correctly specified or misspecified scenarios, (S1) and (S2) respectively. Therefore, the efficiency aspect of the estimation is enhanced in the new class of techniques as well its sensitivity to the perturbation in comparison to the other studied baseline methods of Gaussian PPCA and Student-t PPCA.

We show that the efficiency of the estimator of $\sigma^2$, $\hat{\sigma}^2$ is mostly dependent on the assumptions on the distribution of perturbation term. Next,  we explain that the loss of robustness defined by the asymptotic variance of the estimator of $\mathbf{W}$, $\hat{\mathbf{W}}$ is not uniform across all elements of the projection matrix and is impacted by both, the marginal distribution assumptions of both, $\bm{\epsilon}_t$ and $\mathbf{X}_t$. 
When misspecification of the model is present in the estimation and tail dependence and skewness structure is present in the data generating mechanism, this will impact the efficiency.

The asymptotic variances corresponding to the scenario (S1) for Gaussian PPCA, Student-t PPCA and Grouped-t GSt PPCA are listed in Table \ref{tab:if_var_s1}. The models used for the estimators of parameters $\mathbf{W}$ and $\sigma^2$ are consistent with assumptions about distributions of the observation vectors. The table presents the asymptotic variances for Gaussian PPCA or median values of the asymptotic variances with corresponding interquartile range across all considered degrees of freedom for Student-t PPCA and Grouped-t GSt PPCA. The interquartile range informs us about the dispersion of values across the combinations of the degrees of freedom. 

The Grouped-t GSt PPCA obtains lowest asymptotic variances, regardless of the degrees of freedom that characterizes $\mathbf{Y}$. It is seemingly uniform across degrees of freedom and parameters. This outcome stems from the fact that both, the median values and the interquartile ranges are the lowest.   

Figure \ref{fig:if_s2} shows the logarithm of the asymptotic variances for the scenario (\textbf{S2}), when the observation data follows a Grouped-t GSt PPCA model with different cases for degrees of freedom (x-axis and y-axis of the plots). The estimators of parameters are derived upon Gaussian PPCA (a), Student-t PPCA (b) or Grouped-t GSt PPCA (c) models. The asymptotic variances are scaled per parameter to unify the colour scale.   

Labels of the y-axis and x-axis on all plots correspond to the multidimensional vectors of degrees of freedom $\bm{\nu}_\epsilon$ and $\bm{\nu}_x$ for $\bm{\epsilon}_t$ and $\mathbf{X}_t$, respectively. For instance if  $\bm{\nu}_\epsilon = [4,100,4]$ the corresponding label on the y-axis is '4\_100\_4' and if  $\bm{\nu}_x = [100,4]$ - the label on the x-axis is '100\_4'.  

Similarly to the scenario (S1),  the asymptotic variances under scenario (S2) are the lowest for the estimators derived under the Grouped-t GSt PPCA model. If elements of $\bm{\nu}_\epsilon$ and $\bm{\nu}_x$ are the same per vector, the framework collapses to the Student-t GSt PPCA case. Hence, the robustness of the GSt PPCA class of methods is the highest according to this criterion. We remark on a few observations that might be of interest to the reader.

\begin{figure}[H]
  \centering   
       \begin{subfigure}[b]{0.95\textwidth}
         \centering
 		\includegraphics[scale = 0.15]{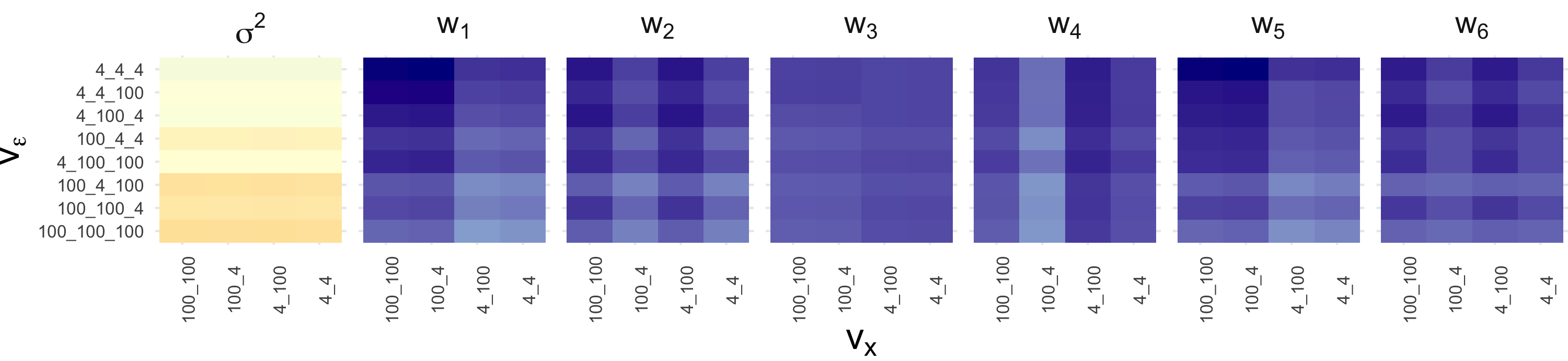}
 		 		\vspace{-0.3cm}
         \caption{Gaussian PPCA}
     \end{subfigure}  
     
      \begin{subfigure}[b]{0.95\textwidth}
         \centering
 		\includegraphics[scale = 0.15]{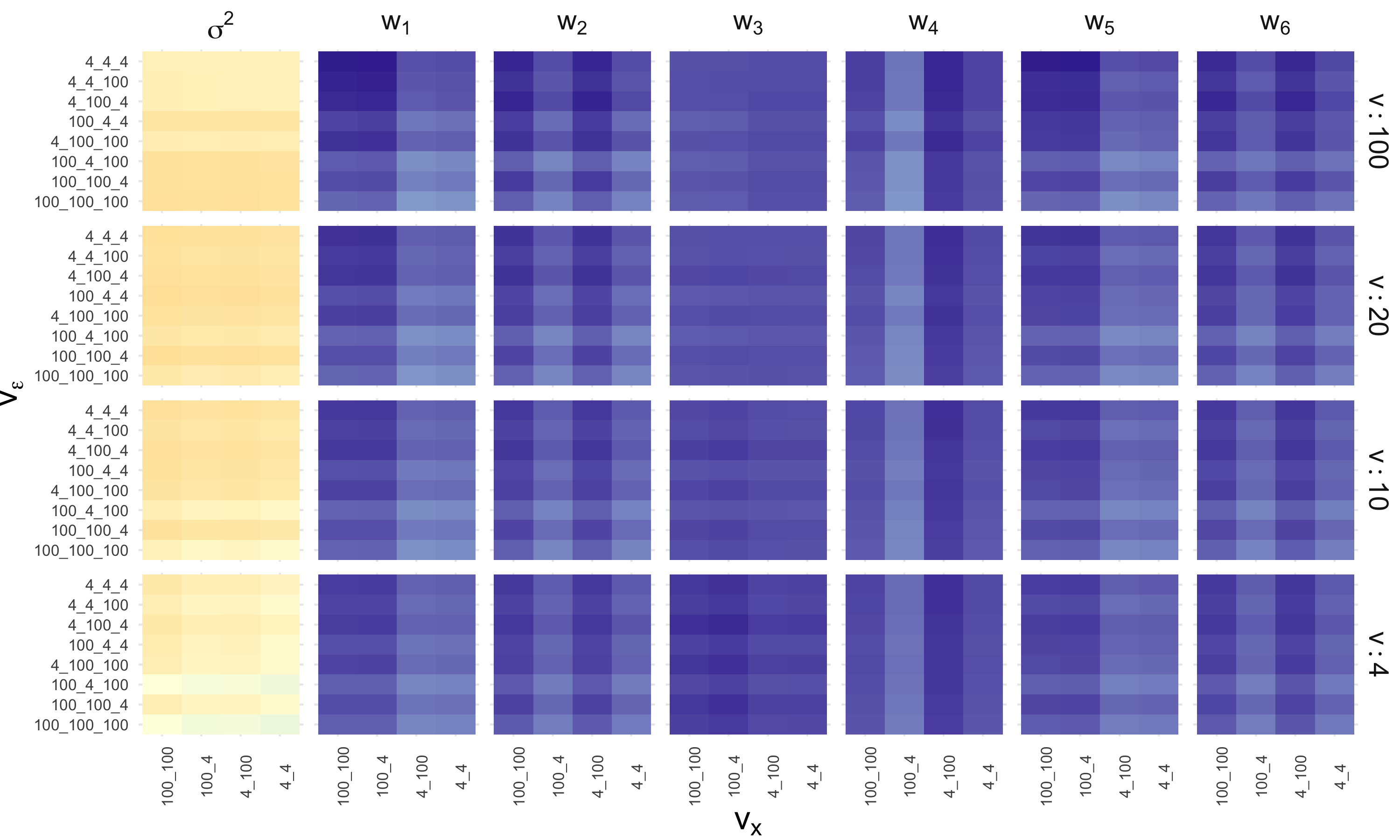}
 		 		\vspace{-0.3cm}
         \caption{Student-t PPCA}
     \end{subfigure}            
    
          \begin{subfigure}[b]{0.95\textwidth}
         \centering
 		\includegraphics[scale = 0.15]{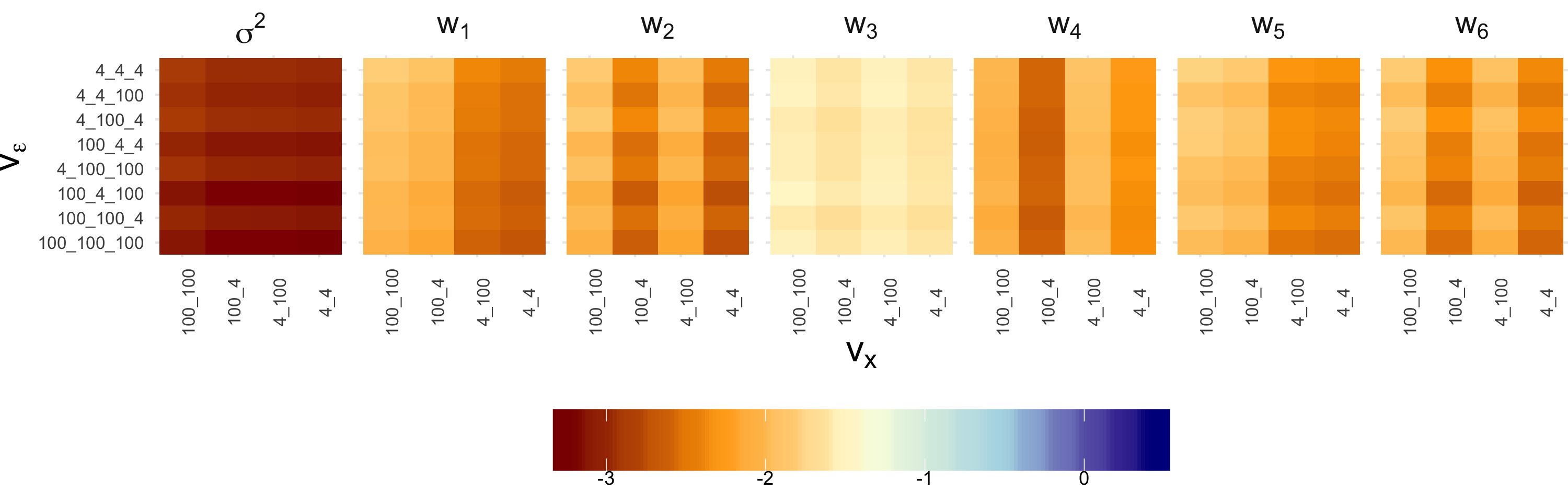}
         \caption{Grouped-t GSt PPCA}
     \end{subfigure}

        \caption{The logarithm of asymptotic variances defined in \eqref{eq:IF_var} of estimators for $\sigma^2$ and $\mathbf{W}$ for Gaussian PPCA (a), Student-t PPCA (b) and Grouped-t GSt PPCA (c) under the scenario ( \textbf{S2}). The observation data follows Student-t GSt PPCA  model under different assumptions on degrees of freedom $\bm{\nu}_\epsilon$ (y-axis) and $\bm{\nu}_x$ (x-axis). The columns-wise order of the panels corresponds to the variances for different parameters. The row-wise order of the panels in panel (b) corresponds to the assumptions on the degree of freedom $\nu$ that the estimators are derived upon. The panel (c) presents results for Grouped-t GSt PPCA when distribution assumptions that are used to derive the estimators and characterise the observation data are consistent.}\label{fig:if_s2}
\end{figure}

The asymptotic variance of $\hat{\sigma}^2$ decreases with an increasing number of marginals that are light-tail in $\bm{\epsilon}_t$. We remark that a profile of heavy-tails of the first marginal of $\bm{\epsilon}_t$ is important as well - if it is heavy-tailed then the asymptotic variance of  $\hat{\sigma}^2$  increases.  

The efficiency of $\hat{\mathbf{W}}$ exhibits interesting behaviours when the sample data has marginal-specific assumptions on distributions of $\bm{\epsilon}_t$ and $\mathbf{X}_t$. The asymptotic variance is element-specific for $\hat{w}_{i}$, for $i = 1,\ldots,6$, and depends mostly on the profile of heavy-tails that corresponds to the marginal of $\mathbf{X}_t$ which is a projection of $\mathbf{Y}_t$. However, its efficiency is also slightly impacted by the heavy-tail distribution of element-specific marginal of $\bm{\epsilon}_t$, that corresponds to the element of $\mathbf{Y}_t$ which is projected by  $\hat{w}_{i}$.  For instance, the estimation of $w_1,\ w_3, \ w_5$  has the highest precision when the first component of $\bm{\nu}_x$ is $4$, that is for the labels on the x-axis '4\_100' and '4\_4' as these components of $\mathbf{W}$ project onto the first marginal of $\mathbf{X}_t$. 

The same patterns of asymptotic variances given by different combinations of values in $\bm{\nu}_\epsilon$ and $\bm{\nu}_x$ that characterize sample data are observed for the estimators defined under Gaussian PPCA and Student-t PPCA models. However, their asymptotic variances are significantly higher and more sensitive to the distribution of the error terms, especially if it is not captured by the model's assumptions used for the estimation. Cases of Gaussian PPCA and Student-t PPCA with high degrees of freedom are the most illustrative for this example. For instance, let us consider asymptotic variances of $w_1$ when the first marginal of $\bm{\epsilon}_t$ is light-tail - the precision of the parameter's estimation under these frameworks declines.  In addition, when the t-Student PPCA model is assumed to be more heavy-tailed, $\nu \in  \big\{ 4,10\big\}$,  the corresponding estimator of $\sigma^2$ loses efficiency, when a combination of light- and heavy-tail marginal profiles of $\bm{\epsilon}_t$ and $\mathbf{X}_t$ are assumed.  Also, the framework becomes less robust to the data that was corrupted by a light-tail perturbation.   

Lastly, we observe that the discrepancy between the asymptotic variances of $\hat{w}_3$ and the other components in $\mathbf{W}$ is the highest under Grouped-t PPCA. 
\vspace{-10pt}
\begin{table}[htb]\caption{The asymptotic variances defined in \eqref{eq:IF_var} of estimators for $\sigma^2$ and $\mathbf{W}$ for PPCA, Student-t PPCA and Student-t GSt PPCA under the scenario ( \textbf{S1}) given their stochastic approximations. The measure is standardized per parameter.  The values corresponding to Student-t PPCA and Grouped-t GSt PPCA reflect the median across all true models for the observation vector that are dependent on the combinations of degrees of freedom. The values in the brackets represent an interquartile range. }\label{tab:if_var_s1}
\centering \small
\begin{tabular}{l|lll}
  \hline
  &  \multicolumn{3}{c}{PPCA} \\ 
 & Gaussian  & Student-t  & Grouped-t GSt  \\ 
  \hline
$\sigma^2$ & 1.9e-02 & 2.3e-02 (4.62e-03) & 9.3e-04 (2.81e-04) \\ 
  $w_1$ & 2.3e-01 & 2.6e-01 (3.13e-02) & 1.4e-03 (1.76e-03) \\ 
  $w_2$ & 1.8e-01 & 2.1e-01 (2.6e-02) &  1.1e-03 (1.37e-03) \\ 
  $w_3$ & 8.2e-01 & 9.9e-01 (1.65e-01) & 1.9e-02 (3.54e-03) \\ 
  $w_4$ & 1.5 & 1.7 (1.79e-01) & 7.3e-03 (6.47e-03) \\ 
  $w_5$ & 1.6e-01 & 1.8e-01 (1.85e-02) & 2e-03 (2.4e-03) \\ 
  $w_6$ & 1.5e-01 & 1.7e-01 (2.23e-02) & 1.6e-03 (1.82e-03) \\ 
   \hline
\end{tabular}
\end{table} 
\FloatBarrier
\vspace{-10pt}
\subsubsection{Sensitivity of Estimators }
The sensitivity of estimators is another property that we discuss to compare the robustness of the three PPCA frameworks using the notion of their influence functions.  The sensitivity measures a maximum impact of outlying points on the estimators.  We focus on the gross error sensitivity, that reflects a maximum impact of single contamination on an estimator, defined in \eqref{eq:if_gross}. Therefore it reflects the information of how a perturbation of a single point on the data set decreases the information conveyed by the dataset about a true parameter. Secondly, we study the local-shift sensitivity that measures the effect of removing a probability mass from one point from the domain of the random variable to another and its maximum effect on the estimator, defined in \eqref{eq:if_shift}.  Therefore, it reflects the effect of maximum deviation of points in the data and their impact on the estimation, standardized by a range of the deviation.

Given $M$ realisations of the observation random vectors, we evaluate influence functions point-wise.  The gross-error sensitivity is calculated as a maximum absolute value of evaluated influence functions for a generated dataset.  The local-shift sensitivity is calculated by measuring the pairwise $\mathbb{L}_1$ distances between realisations and separately per corresponding influence functions. Then, the maximum of the fractions defined in \eqref{eq:if_shift} is obtained for different cases of the data. 

The numerical approximations of $\gamma^*(T,F)$ and $\lambda^*(T,F)$  under the scenario (S1) for Gaussian PPCA, Student-t PPCA and Grouped-t GSt PPCA are listed in Table \ref{tab:if_errors_s1}.  The table presents the values for Gaussian PPCA or the median values with corresponding interquartile range across all considered degrees of freedom for Student-t PPCA and Grouped-t GSt PPCA.   Both sensitivity measures confirm that the estimators under the Grouped-t GSt PPCA are the most robust according to the sensitivity measures. 
\begin{figure}[H]
  \centering
       \begin{subfigure}[b]{0.95\textwidth}
         \centering
 		\includegraphics[scale = 0.15]{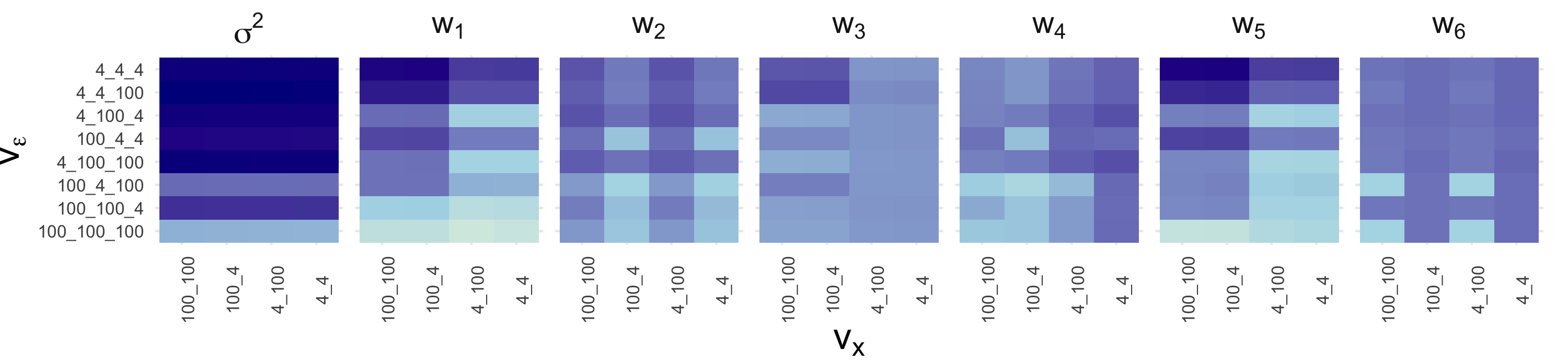}
 		\vspace{-0.3cm}
         \caption{Gaussian PPCA}
     \end{subfigure}  
     
      \begin{subfigure}[b]{0.95\textwidth}
         \centering
 		\includegraphics[scale = 0.15]{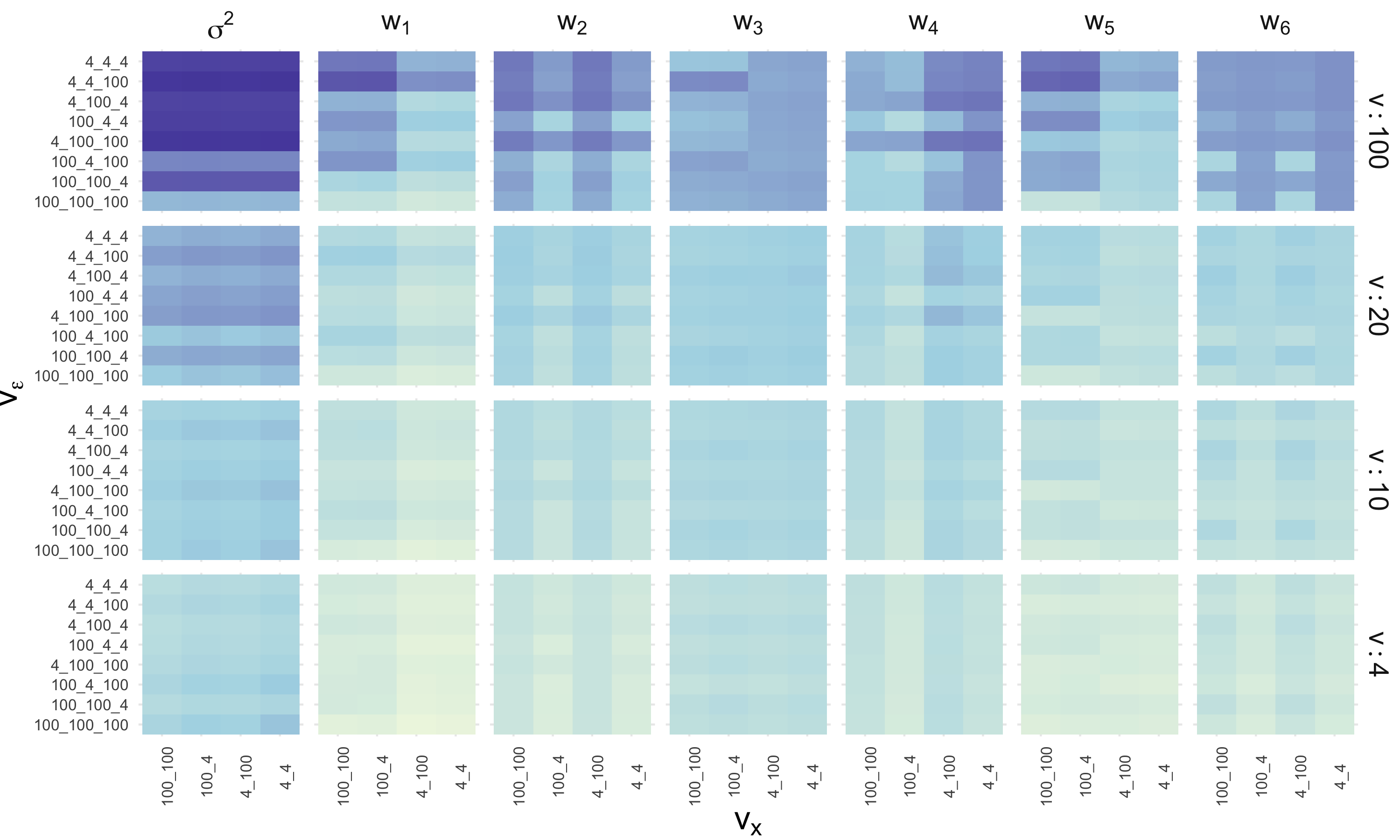}
 		 		\vspace{-0.3cm}
         \caption{Student-t PPCA}
     \end{subfigure}             
    
          \begin{subfigure}[b]{0.95\textwidth}
         \centering
 		\includegraphics[scale = 0.15]{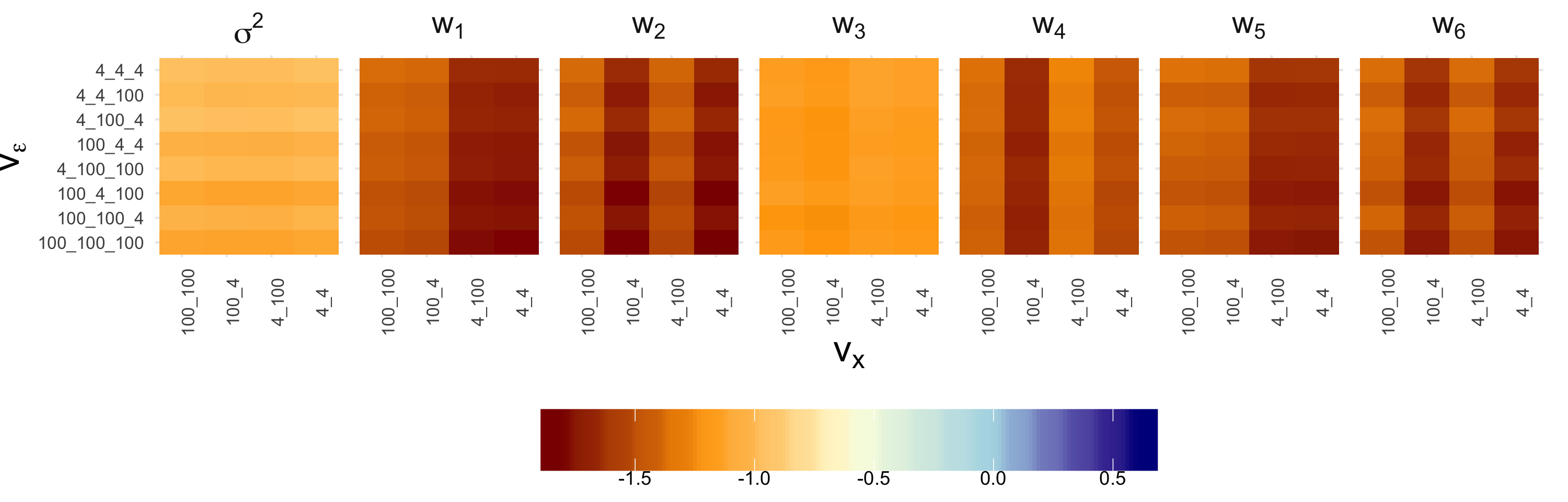}
 		 		\vspace{-0.3cm}
         \caption{Grouped-t GSt PPCA}
     \end{subfigure}      
        \caption{The logarithm of the gross error sensitivity from  \eqref{eq:if_gross} of estimators for $\sigma^2$ and $\mathbf{W}$ for Gaussian PPCA (a), Student-t PPCA (b) and Grouped-t GSt PPCA (c) under the scenario ( \textbf{S2}) given $1000$ realisation of $\mathbf{Y}$. The observation data follows Student-t GSt PPCA  model under different assumptions on degrees of freedom $\nu_\epsilon$ (y-axis) and $\nu_x$ (x-axis). The columns-wise order of the panels corresponds to the variances for different parameters. The row-wise order of the panels in panel (b) corresponds to the assumptions on the degree of freedom $\nu$ that the estimators of  $\sigma^2$ and $\mathbf{W}$ are derived upon. The panel (c) presents results for Grouped-t GSt PPCA when distribution assumptions that are used to derive the estimators and characterise the observation data are consistent.  }\label{fig:if_gross_s2}
\end{figure}

\begin{figure}[H]
  \centering  
       \begin{subfigure}[b]{0.95\textwidth}
         \centering
 		\includegraphics[scale = 0.15]{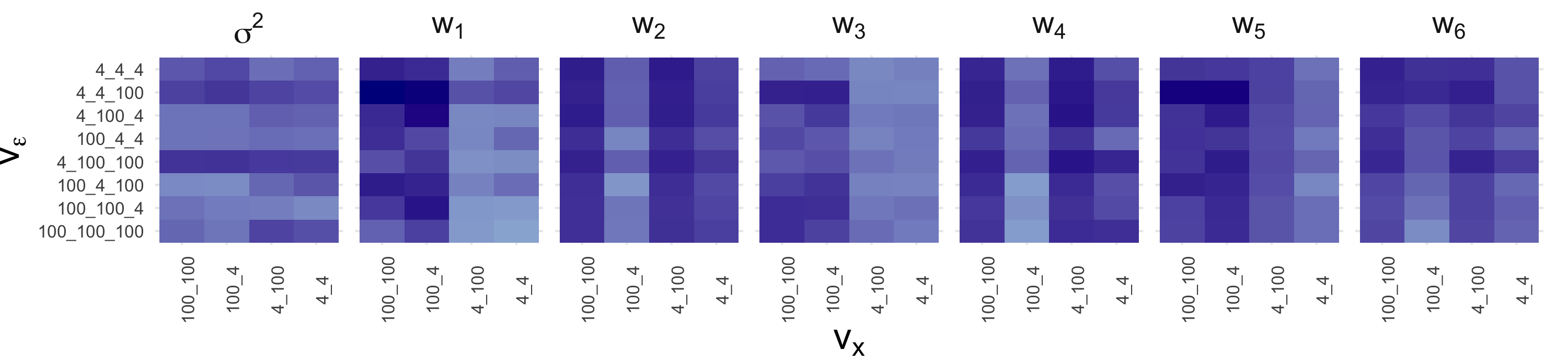}
 		 		\vspace{-0.3cm}
         \caption{Gaussian PPCA}
     \end{subfigure}  
      
      \begin{subfigure}[b]{0.95\textwidth}
         \centering
 		\includegraphics[scale = 0.15]{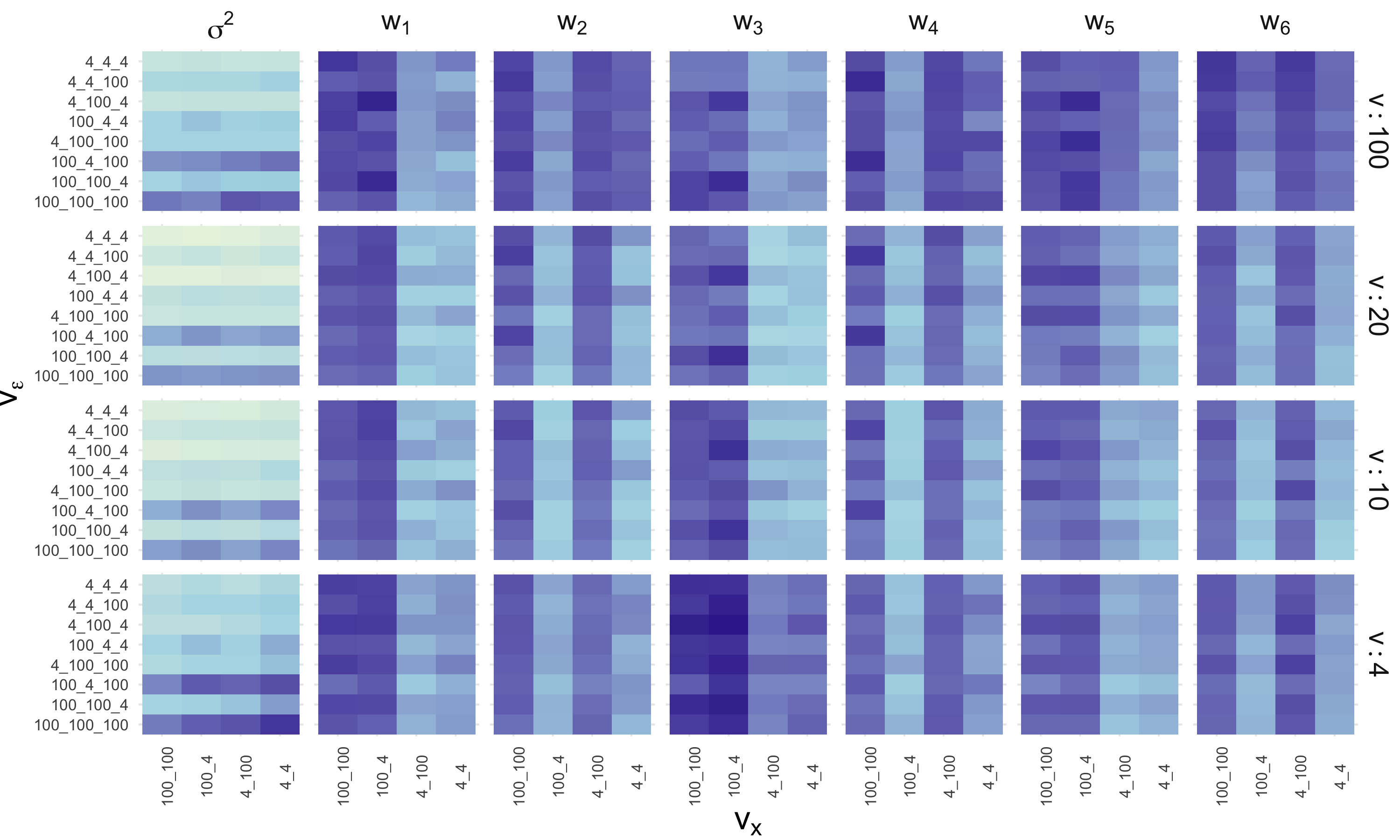}
 		 		\vspace{-0.3cm}
         \caption{Student-t PPCA}
     \end{subfigure}

          \begin{subfigure}[b]{0.95\textwidth}
         \centering
 		\includegraphics[scale = 0.15]{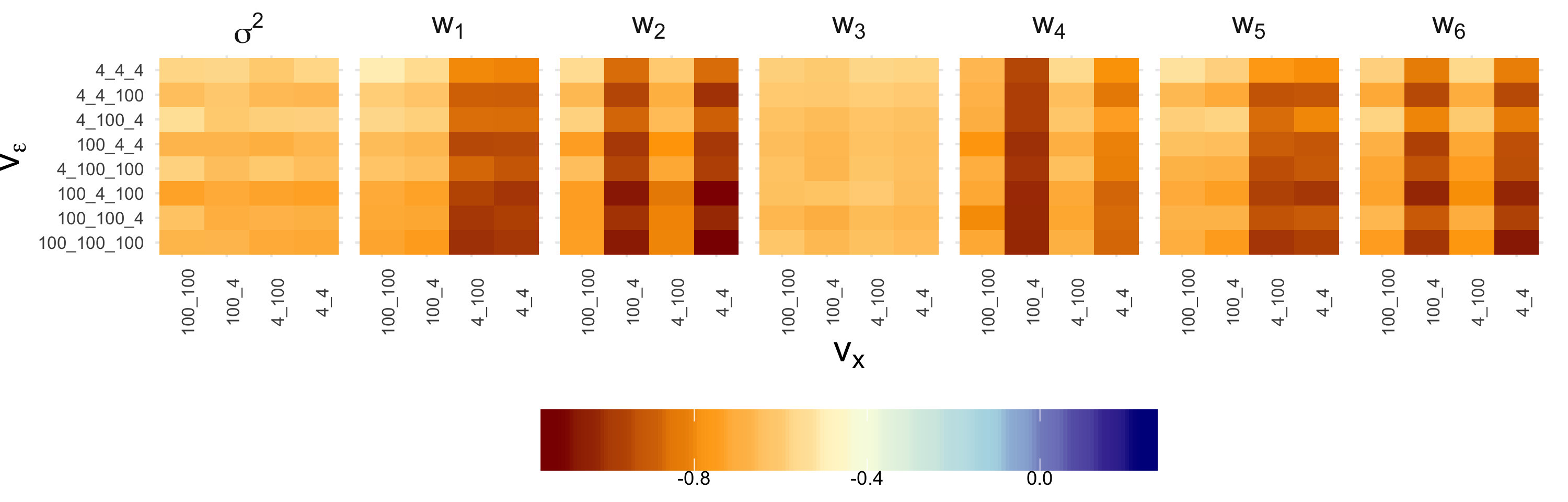}
 		 		\vspace{-0.3cm}
         \caption{Grouped-t GSt PPCA}
     \end{subfigure}      
        \caption{The logarithm of local-shift sensitivity from  \eqref{eq:if_shift} of estimators for $\sigma^2$ and $\mathbf{W}$ for Gaussian PPCA (a), Student-t PPCA (b) and Grouped-t GSt PPCA (c) under the scenario ( \textbf{S2}) given $1000$ realisation of $\mathbf{Y}$. The observation data follows Student-t GSt PPCA  model under different assumptions on degrees of freedom $\nu_\epsilon$ (y-axis) and $\nu_x$ (x-axis). The columns-wise order of the panels corresponds to the variances for different parameters. The row-wise order of the panels in panel (b) corresponds to the assumptions on the degree of freedom $\nu$ that the estimators of  $\sigma^2$ and $\mathbf{W}$ are derived upon. The panel (c) presents results for Grouped-t GSt PPCA when distribution assumptions that are used to derive the estimators and characterise the observation data are consistent. }\label{fig:if_shift_s2}
\end{figure}
\begin{table}[ht]\caption{The gross error sensitivity defined in \eqref{eq:if_gross} and the local-shift sensitivity defined in \eqref{eq:if_shift} of estimators for $\sigma^2$ and $\mathbf{W}$ for PPCA, Student-t PPCA and Student-t GSt PPCA under the scenario ( \textbf{S1}) given given $M$ realisations of $\mathbf{Y}$. The measures are standardised by parameter.  The values corresponding to Student-t PPCA and Grouped-t GSt PPCA reflect the median across all true models for the observation vector that are dependent on the combinations of degrees of freedom. The values in the brackets represent an interquartile range. }\label{tab:if_errors_s1}
\centering \small
\begin{tabular}{l|lll|lll}
&  \multicolumn{6}{c}{PPCA}  \\
 & Gaussian & Student-t & Grouped-t GSt & Gaussian  & Student-t  & Grouped-t GSt  \\ 
\hline
&  \multicolumn{3}{c}{$\gamma^*(T,F)$} &  \multicolumn{3}{c}{$\lambda^*(T,F)$}\\
\hline  
$\sigma^2$ & 8.6e-01 & 8.6e-01 (6.51e-02) & 9.4e-02 (1.79e-02) & 9.8e-01 & 8.4e-01 (1.08e-01) & 2.2e-01 (4.03e-02) \\ 
  $w_1$ & 2.5 & 2.5(1.87e-01) & 9.3e-02 (5.93e-02) &  2.7  & 2.3 (2.57e-01) & 3.4e-01 (2.15e-01  \\ 
  $w_2$ & 2.3 & 2.9 (2.99e-01) & 8.2e-02 (5.05e-02) & 2.3  & 2.2 (2.86e-01) & 2.8e-01 (1.6e-01) \\ 
  $w_3$ & 6.3  & 6.0 (1.34) & 3.9e-01 (3.04e-02) & 2.8  & 2.6 (2.04e-01) & 5.3e-01 (5.07e-02)\\ 
  $w_4$ & 6.6 & 5.8 (6.82e-01) & 2.3e-01 (1.06e-01)  & 3.6  & 3.2 (1.37e-01) & 5e-01 (2.31e-01) \\ 
  $w_5$ & 2.4 & 2.3 (4.28e-01) &  1.1e-01 (6.06e-02) & 2.2  & 2.1 (2.66e-01) & 3.7e-01 (1.89e-01) \\ 
  $w_6$ & 2.6 & 1.8 (4.44e-01) & 9.5e-02 (5.22e-02)& 2.3 & 2.0 (2.15e-01) & 3.2e-01 (1.63e-01) \\ 
   \hline 
\end{tabular}
\end{table}
The sensitivity of the estimators for the two measures under the misspecified case, the scenario (S2), are illustrated in Figure \ref{fig:if_gross_s2} and Figure \ref{fig:if_shift_s2} , for $\gamma^*(T,F)$ and $\lambda^*(T,F)$, respectively. 

The gross-error sensitivity exhibits similar patterns that are observed for the asymptotic variance. Given Grouped-t GSt PPCA, we observe a decline of the sensitivity of $\hat{\sigma}^2$ to the contamination of a single point when the distribution of the perturbation term is more light-tailed. However, when the estimator is derived under Gaussian PPCA or Student-t PPCA models, the single point contamination has the highest impact on the estimation $\sigma^2$ when models do not capture the heavy-tail distribution of $\bm{\epsilon}_t$. The sensitivity measure of the estimator under the Student-t PPCA model decreases when the assumed distribution is more heavy-tailed. This increase of robustness for estimation of  $\sigma^2$, that depends on $\nu$, is more rapid than when we considered the asymptotic variance.  The gross-error sensitivity of $\hat{\mathbf{W}}$ is element-specific and depends on marginal-specific assumptions on distributions of $\bm{\epsilon}_t$ and $\mathbf{X}_t$. We observe similar patterns as with the asymptotic variance. Again, the estimation of components of $\mathbf{W}$ under Gaussian PPCA and Student-t PPCA is more impacted by the distribution of $\bm{\epsilon}_t$, especially when the estimators have no flexibility to handle heavy-tail data.  Also, the gross-error sensitivity of the estimators under the Student-t PPCA decreases when $\nu$ decreases and becomes more uniform across different assumptions on the sample data. 

The robustness of the estimators that is defined by the local shift sensitivity, $\lambda^*(T,F)$, is illustrated in Figure \ref{fig:if_shift_s2}.  The sensitivity analysis by this measure results in similar patterns, especially for $\mathbf{W}$, that are observed for the asymptotic variance. However, it is less impacted by the different assumptions on the distribution of $\mathbf{\epsilon}_t$.  The Grouped-t GSt PPCA is significantly more robust according to this measure than the other PPCA frameworks. The main difference between the outcomes for the local-shift sensitivity and the other studied measures of robustness is the sensitivity of $\hat{\sigma}^2$.  We observe almost uniform sensitivity of this estimator under Grouped-t GSt PPCA regardless of the data assumptions. When the estimator is derived under the Gaussian PPCA or Student-t PPCA models, it is least robust for the data that has light-tail perturbation.  Also, we observe little change for the Student-t PPCA performance when the degrees of freedom $\nu$ change.

\FloatBarrier

\subsubsection{Bias of Estimators}
In the final part, we focus on the analysis of bias of the estimation for $\sigma^2$ and $\mathbf{W}$ given the increasing sample size.  We narrow our study and examine accuracy of the estimation for Student-t GSt PPCA as a special case of Grouped-t GSt PPCA, Student-t PPCA and Gaussian PPCA frameworks under the misspecified data case, the scenario (S2),  and different sample sizes $N = 100, \ 500, \ 1000, \ 5000, \ 10000$. We generate $M = 50$ replications of $\mathbf{Y}_{1:N}$ for each distribution and sample size assumptions to numerically calculate the mean square errors of the estimation, having the true values of the parameters specified as before.

Figure \ref{fig:if_mse}  shows the change of estimation accuracy for PPCA frameworks over increasing sample size measured by the mean squared error of $\sigma^2$ and its estimate, and the mean squared error of $\mathbf{W}$ across all its elements. We observe that the increasing sample size increasingly improves the estimation via Student-t GSt PPCA (red line), especially of $\sigma^2$. The estimation via Student-t GSt PPCA is consistently highly accurate regardless of the degrees of freedom.  When the distribution of $\bm{\epsilon}_t$ is light-tail, the accuracy of the estimators of Student-t PPCA is close to the one provided by the more complex framework. For this data case, Gaussian PPCA estimates well only $\sigma^2$, and its discrepancy of the accuracy of the estimation via Student-t GSt PPCA becomes higher with bigger sample size. On the other hand, when the distribution of $\mathbf{X}_t$ is light-tail, the estimation via Student-t PPCA and Gaussian PPCA results in good estimates of $\mathbf{W}$, however, at the prices of the estimators of $\sigma^2$.  We remark that  when the data follows more complex distribution such as Grouped-t GSt PPCA model, the discrepancies in the estimation accuracy increases significantly and the standard methods struggle to provide accurate estimation.

\begin{figure}[H]
  \centering   
\includegraphics[width = \textwidth]{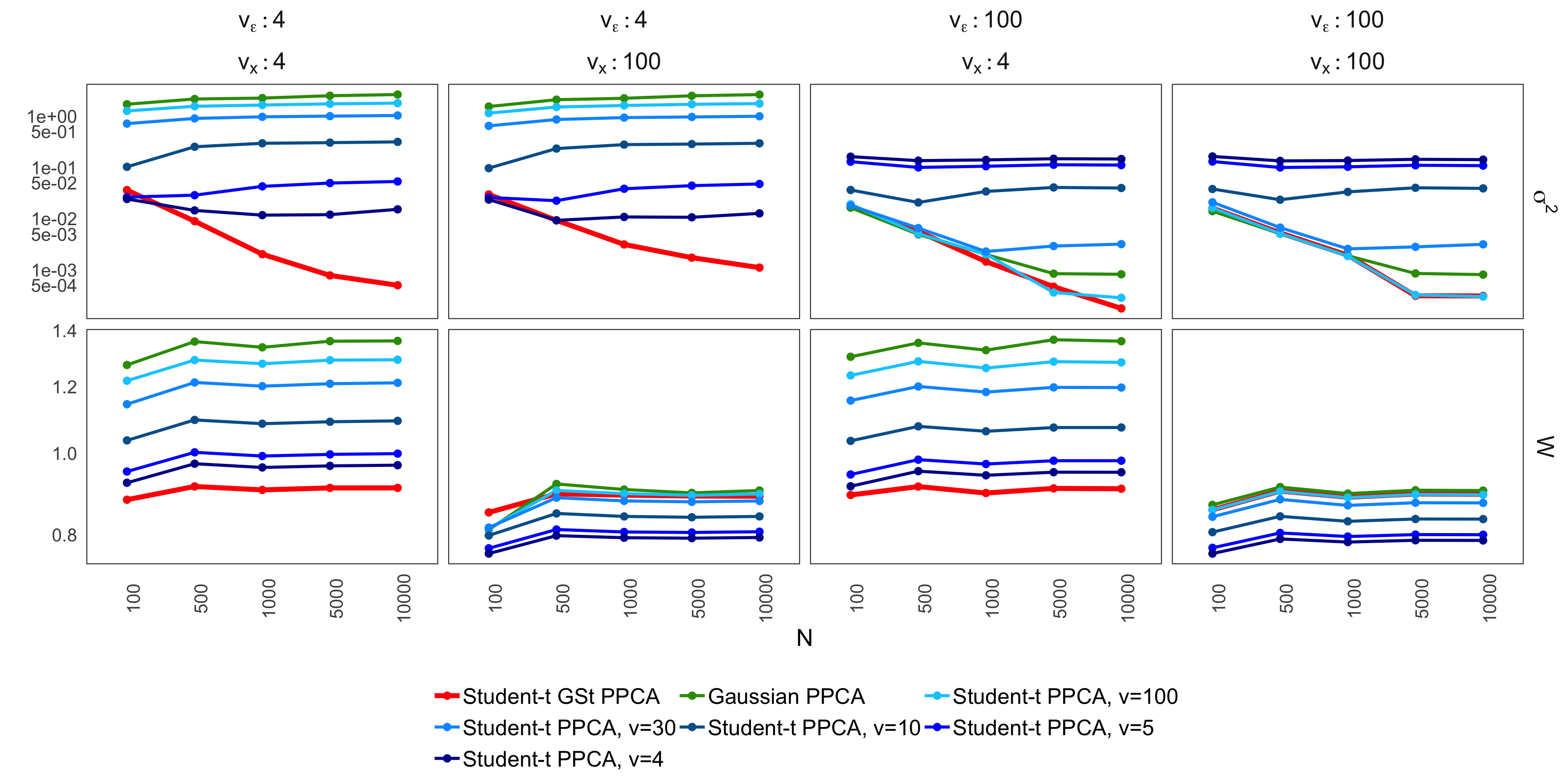}
         
  \caption{The logarithm of the mean squared errors of the estimates for $\sigma^2$ and $\mathbf{W}$ (across all its elements)  for Gaussian PPCA, Student-t PPCA and Student-t GSt under the scenario (\textbf{S2}) versus sample size $N$.  The column-wise order of the panels corresponds to the distribution assumption of $\mathbf{Y}_{1:N}$, the degrees of freedom $\nu_\epsilon$ of $\bm{\epsilon}_t$ and $\nu_x$ of $\mathbf{X}_t$, that defines the Student-t GSt PPCA model of $\mathbf{Y}_{1:N}$. The row-wise order corresponds to the logarithms of the means square errors per parameter. The colours of lines correspond to accuracy under different PPCA models assumptions. }\label{fig:if_mse}
\end{figure}

\section{Real Data Study on a Set of Crypto Assets}
PCA or PPCA methods can be used to measure market concentration and the potential for diversification. They are often employed to identify highly concentrated assets or to reduce the complexity of large sets of financial instruments by transforming them into a new set of uncorrelated components. One example of an application of these components is a strategy of diversified risk parity on the new representation that allocates portfolio weights to the original set.  

These feature extraction frameworks can be seen as a set of techniques to reveal common factors, called principal components, in a way that best explains the variability in the original data. The transformation of the observation data into principal components is defined in such a way that principal components have a decreasing variance. The methods suggest how to lower the dimensionality of our original data set by excluding elements which are in majority described by components with the least significant contribution to the overall variance and therefore reduce the size of investment portfolios universe of possible constituent assets to perform risk-based asset selection and weighting. 

In the following part we study the linear interactions between Bitcoin and 19 other altcoin crypto assets that are ranked highest on the list of top virtual currencies by market cap given two separate periods, 2018, so-called Initial Coin Offering (ICO) period where most of the cryptocurrencies projects were born and 2019, that is a start of lending markets when the other altcoins coins started to be less frequently traded. The details of the considered assets are given in Table \ref{tab:crypto_list}. The data for our study was collected from the Coin Metrices website (https://coinmetrics.io). We follow the categorisation of the assets from the Cryptoslate website (https://cryptoslate.com).  

We do not intend to present an optimal model that describes the dynamics of the studied data sets. Our motivation is to emphasise the effects of the robustness of the family of GSt PPCA methods that have been discussed in Section \ref{sec:IF}.  We present the eigen decomposition of the covariance matrix of the set of $20$ crypto assets given its estimators defined by $5$ different PPCA models: the standard Gaussian PPCA by \cite{TippingBishop1999a}, the standard Student-t PPCA of \cite{RidderFranc2003} and the three special cases of the family of GSt PPCA models introduced in Section \ref{sec:GStSPPCA}: Grouped-t GSt PPCA as Special Case 1 from Section \ref{sssec:PPCA_GStS_special}, Student-t GSt PPCA as Special Case 3 and Skew-t GSt PPCA as simplified Special Case 2 discussed in Appendix~\ref{appendix:skewX_syn_case_studies}. We show that the analysis of the eigendecomposition of the covariance matrix of the observation sets can provide useful insights into the distinguishable components of the variance that have economic interpretations. In the next section, we show that the flexibility of the PPCA frameworks from GSt PPCA family results in the separation of the covariance matrix into the components that have clearer economic interpretation than the other considered PPCA methods. 

\begin{table}[ht]
\caption{The list of $20$ altcoin crypto assets from 01-01-2018 to 31-12-2019 with corresponding categories. }\label{tab:crypto_list}
\centering \small
\begin{tabular}{lll|| lll}
  \hline
Ticker & Currency Name & Category & Ticker & Currency Name & Category \\ 
  \hline
  bat & Basic Attention Token & Advertising & etc & Ethereum Classic & Smart Contracts \\ 
  powr & Power Ledge & Energy  & eth & Ethereum & Smart Contracts\\ 
  bnb & Binance Coin & Exchange &  lsk & Lisk & Smart Contracts \\ 
   omg & OmiseGO & Financial Service &  neo & NEO & Smart Contracts\\ 
  xrp & XRP & Financial Service &  xlm & Stellar & Smart Contracts  \\ 
  dash & Dash & Governance &    dai & Dai & Stablecoin \\ 
  link & Chainlink &  Interoperability  & usdt\_eth & Tether (Ethereum) & Stablecoin  \\
  xmr & Monero & Privacy & bch & Bitcoin Cash & Technology\\  
  ada & Cardano &Smart Contracts &  btc & Bitcoin & Technology \\   
  eos & EOS & Smart Contracts&  ltc & Litecoin & Technology  \\ 

   \hline
\end{tabular}
\end{table}

\begin{figure}[H]
  \centering   
 	\includegraphics[width = \textwidth]{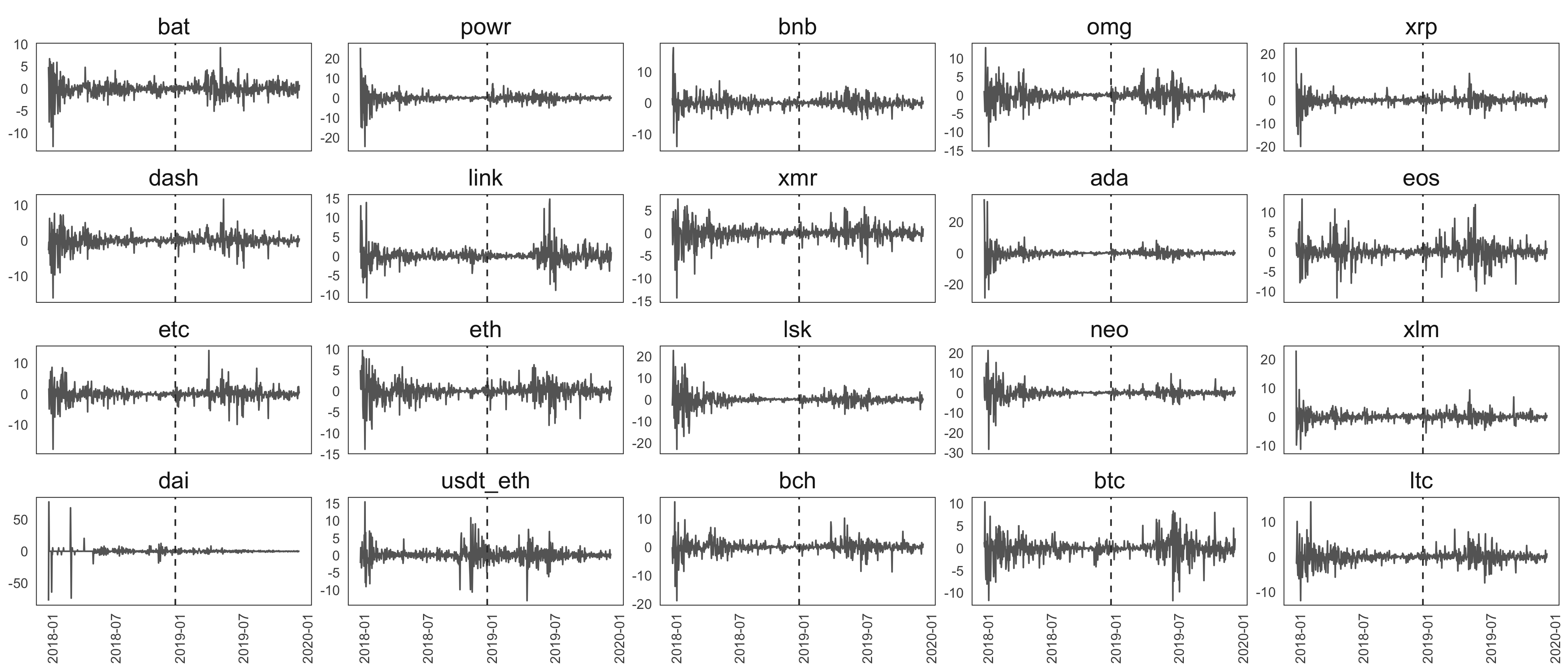}
        \caption{The standardized daily returns of $20$ crypto currencies listed in Table \ref{tab:crypto_list}. }
\label{fig:crypto_rets_std_cat}
\end{figure}

\subsection{Data Preparation \& Discussion on Collinearity}
To begin our study, we unify the magnitude of assets' values over time by considering their standardized returns. We calculate standard daily returns which are defined as a daily nominal change in price over time references to the US Dollar stable coin USD Tether (USDT) as numeraire. This was selected as it is the stable coin with highest market capitalization and utilization in all key exchanges. We divide the set of returns in subsets related to 2018 and 2019 and standardize them robustly per currency exchange rate by the Huber M-estimators of \cite{Huber1964}. 

Figure \ref{fig:crypto_rets_std_cat} illustrates the daily returns of the currencies. The dashed vertical lines divide the set of returns into two separate periods, one corresponding to 2018 and another one to 2019.   Figure \ref{fig:params_skew_sample_cat} shows the marginal estimates of the skewness by daily returns of assets. It is our motivation to include the model with skewness to our analysis. 

Figure~\ref{fig:5crypto_pairplots_cat_2019} illustrates the interactions between pairs of examined assets (off-diagonal panels) as well and histograms of the returns (diagonal panels) for year 2018 (red color) and 2019 (black color), respectively.  We observe weaker dependence between stable coins, USD Tether and Dai, and the other assets across the years.  The panels in Figure~\ref{fig:5crypto_pairplots_cat_2019} show possible increasing collinearity between Dai and Bitcoin, Bitcoin Cash and Litecoin, the three assets from the same category 'Technology' in 2019. On the other hand, the dependence between USD Tether and the remaining assets stays weak. The weak collinearity between stable coins and the rest of the assets stems from their design. They have been created to bridge between the highly volatile crypto currencies and stability of fiat currencies but still providing anonymity to its buyers. 

It is also worth to point out weakening dependence between Chainlink (link) and the other assets with 2019.  In 2018, we can still observe a moderate correlation between Chainlink and coins such as Basic Attention Token (bat), Cardano (ada), Dash (dash) and assets from Technology. It becomes significantly weaker in 2019 as shown in Figure~\ref{fig:5crypto_pairplots_cat_2019}.  

\begin{figure}[H]
\centering
\includegraphics[width = 0.75\textwidth]{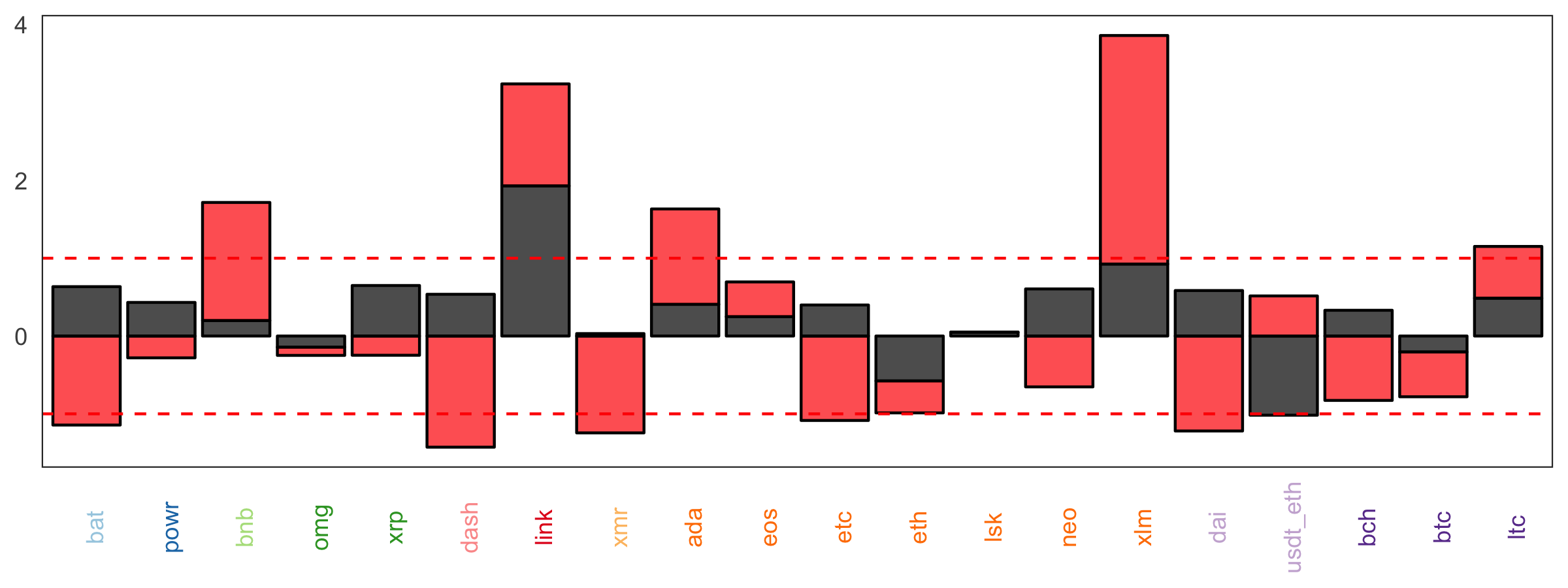}  
  \caption{The yearly sample estimates of univariate skewness for  returns of $20$ crypto currencies listed in Table \ref{tab:crypto_list}. The columns indicate the year of the sample.  The estimates for the standardized returns are the same as for non-standardized. The colors corresponds to the different sample periods, 2018 (red) and 2019 (black). }\label{fig:params_skew_sample_cat}
\end{figure}

\FloatBarrier

\subsection{Covariance Decomposition for 20 Crypto Currencies with Category-specific Heavy Tails Assumptions}
We seek to compare the interpretation of the observed collinearity with the eigen decomposition of the covariance that would select uncorrelated directions which explain the majority of the variance in the analysed data set of $20$ crypto assets. The optimal PPCA model choices and resulting log-likelihoods for $5$ PPCA frameworks are shown in Table~\ref{tab:20ccy_res}.

Overall results of the covariance decompositions have implications for the benefits of diversification as they indicate that the majority of the altcoin crypto assets are driven by a common factor that is highly correlated to Bitcoin. This co-moving group of assets is characterized by the highest contribution to the overall variance. The identification of this collinearity can aid the portfolio selection and management as holding only one of these assets provide most of the benefits for the diversification and allows to invested funds in other assets. On the other had, the remaining principal direction indicate uncorrelated assets that can be used for the risk hedging purposes. 

The proportion of the market variance explained by the first component increases from 2018 to 2019. It suggests that Bitcoin returns in 2019 to be the main driver of altcoins and, consequently, of the overall variance on the market.  We remark that the decomposition by the GSt PPCA family of methods is always characterized by the higher proportion of the first principal components to the overall variance.

The observed estimates of sample skewness in Figure \ref{fig:params_skew_sample_cat} motivate us to use a grid of $7$ elements per margin to select the skewness parameter, $\bm{\delta}_x$ for Skew-t GSt PPCA.  Therefore, we consider that the elements of the $3$-dimensional vector $\bm{\delta}_x$ can take values in the set $\{-1,-0.5,-0.2,0,0.2,0.5,1\}$.  Only in 2018 for non-robustly standardized returns, the first component of $\mathbf{X}_t$ was characterized by skewness and $\bm{\delta}_x = [0.5,0,0]$.  Otherwise,  the best model for Skew-t GSt PPCA assumes zero and the model simplifies to Student-t GSt PPCA model.  

The estimates of the three eigenvectors corresponding to the highest eigenvalues given by each of PPCA approaches are illustrated in Figure \ref{fig:params_evecl_cat} (a). The row-wise-wise order of panels corresponds to the different periods and standardisation methodologies. The labels of the y-axis correspond to the crypto assets with category-specific font colours. 

The estimates of the principal directions are consistent with the interpretation of the linear interactions between the assets.  In 2018, especially for the frameworks from the GSt PPCA family, the variance of the data set is decomposed into the principal direction that reflects the dependence between all crypto assets, except stable coins. The remaining principal components are dominated by each of the stable coins separately, USD Tether and Dai, respectively.  Given the weak interaction between USD Tether and Dai in 2018 as in Figure~\ref{fig:5crypto_pairplots_cat_2019}, it is an expected finding. Therefore, we observe little collinearity between USD Tether and Dai what is an observation that can be utilized in improvements to portfolio diversification. Only Gaussian PPCA indicated negative collinearity between the group of crypto assets correlated with Bitcoin and the stable coin Dai that is indicated by the first component. 

The contribution of the principal directions to the variance of the dataset measured by eigenvalues of the covariance matrix is illustrated in Figure \ref{fig:params_evecl_cat} (b). The panels show the proportion between the first eigenvalue and the remaining ones.  In 2018, we observe higher contribution to the overall variance of the second principal direction that represents the dynamic of Dai, especially for the Gaussian and Student-t PPCA. The variance explained by the third components, that characterizes Theter, is significantly lower.  The methods from the GSt PPCA family suggest higher proportion of the first principal components to the overall variance. 

In 2019, all frameworks suggest higher loadings on stable coins in the first component than in 2018; the second principal direction concentrates solely on the dynamic of Chain Link and the third component is dominated by stable coins. The frameworks from GSt PPCA family suggest higher loadings on one of the stable coins, Dai or USD Tether. Therefore, the stable coins are recognized as collinear by all frameworks in 2019. Also, their correlation with the rest of the assets increases. The proportions of eigenvalues illustrated in Figure \ref{fig:params_evecl_cat} (b) suggest that the first component dominates the variance of the dataset and the remaining ones have a significantly lower contributions. Therefore, the variance of the market explained by the stable coins is reduced in 2019.

\begin{figure}[H]
  \centering  
         \begin{subfigure}[b]{0.69\textwidth}
         \centering
\includegraphics[width = \textwidth]{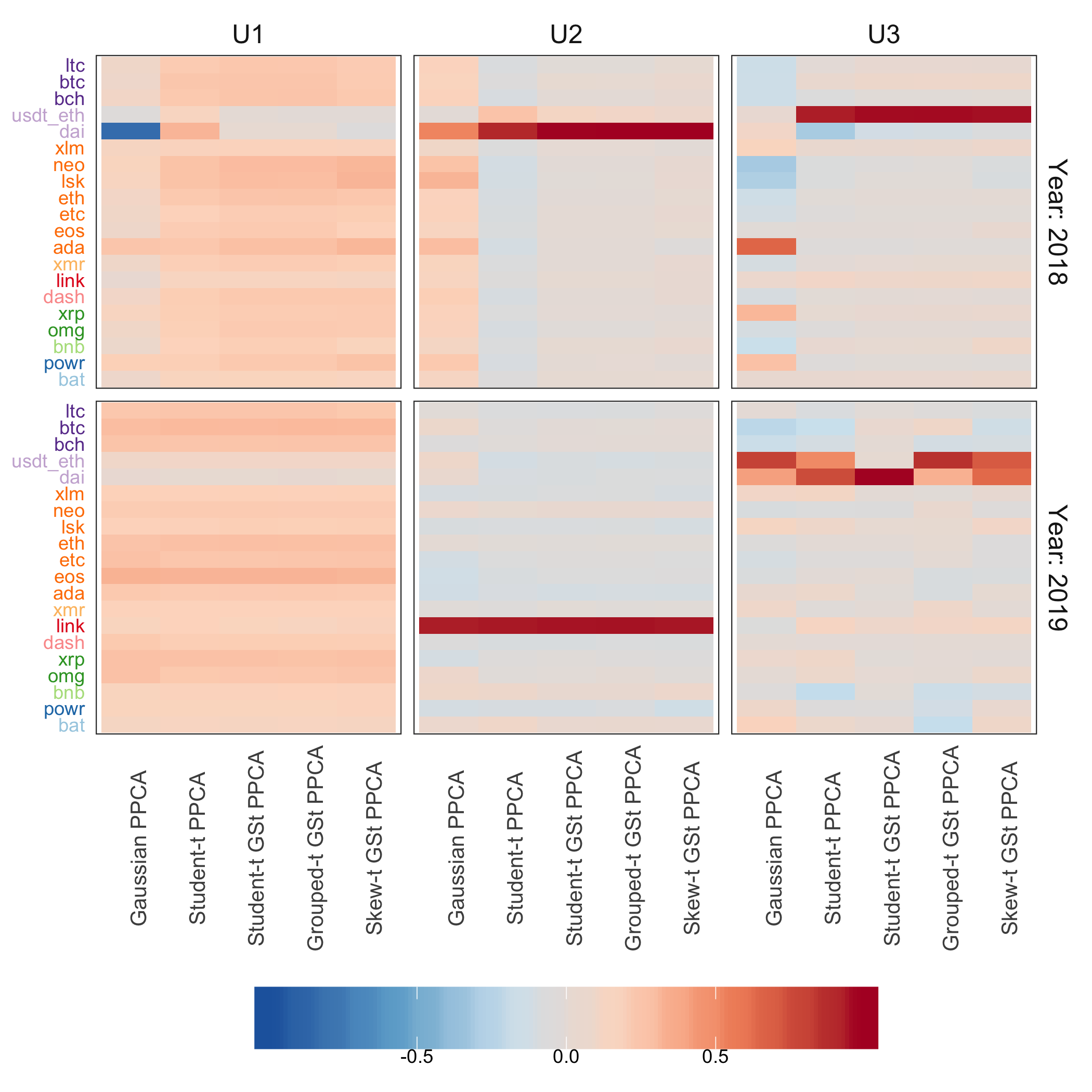}     
\caption{}
     \end{subfigure}  
            \begin{subfigure}[b]{0.295\textwidth}
         \centering
	\includegraphics[width = \textwidth]{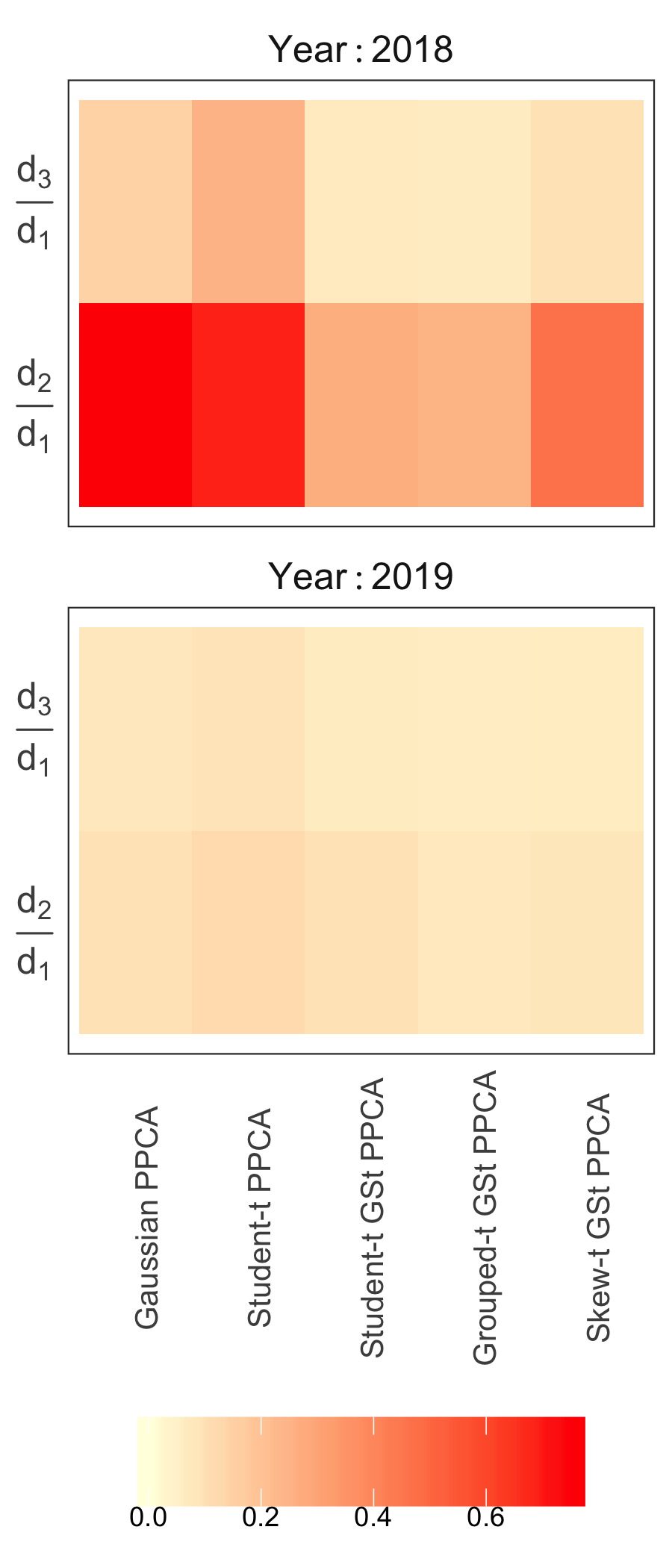}
	\caption{}
     \end{subfigure}   
       
 \caption{The yearly estimates of eigenvectors of the covariance matrix of $\mathbf{Y}_t$ (a) and the proportions between the estimates of the first eigenvalue ($d_1$) and the remaining ones($d_2, \ d_3$) obtained via $5$ PPCA algorithms (x-axis): Gaussian PPCA, Student-t PPCA, Student-t GSt PPCA and Grouped-t GSt PPCA and Skew-t GSt PPCA. The column-wise order corresponds to the year of the sample set and the standardisation method applied to returns.The observation process  consists of $20$ cryptocurrencies listed in Table \ref{tab:crypto_list} with category-specific assumption on heavy tails if possible. }\label{fig:params_evecl_cat}
\end{figure}
\FloatBarrier
\section{Conclusions}\label{sec:conclusions}
The research presented in this work constitutes important and novel contributions towards probabilistic feature extraction methods and their application to statistical modelling.  We focused on developing a dimensionality reduction methodology, which addresses a difficult but not uncommon situation when the underlying observation data is not fully observed; that is, it both contains missing information and is corrupted by noise. We develop a framework of dimensionality reduction which adapts PPCA and extends the standard assumptions on the distribution of the observation data to heavy-tailed and skewed distributions.  The method addresses a common situation, in which the elements of an observation vector have non-trivial dependency structures. It is especially relevant to large data problems when subsets of the observation vector represent different, complex information, often combined from various data sets. Therefore, we introduce a method which assumes many sources of corruption which are no longer identical across the dimensionality of the observation vector. Crucially, the model improves the ability to capture tail dependence patterns in the multivariate data analysis. The adaptation of the skew-t distribution in the PPCA framework adds flexibility to account for asymmetric distributions if they are relevant to an observation set. We derive efficient Expectation-Maximisation algorithm for the estimation of the parameters in the new model.

We assessed the robustness of the new class of PPCA methods by examining their estimation accuracy and various measures of sensitivity to corruption defined by influence function on simulation studies. The developed family of methods is characterized by the highest robustness in comparison to standard PPCA methods such as Gaussian PPCA or Student-t PPCA.  If the sample data reveals characteristics captured by the GSt PPCA family, such as separation of the tail effect or grouped multiple-degree-of-freedom structures that define patterns of marginal heavy-tail distributions, we show the significant loss of accuracy and robustness for simpler approaches and high robustness of GSt PPCA family. The class of GSt PPCA methods is also most robust when the data follow simpler distributions.

We illustrated the applicability and performance of the new class of methods on a real study where we examine linear interactions and covariance decompositions on the dataset of $20$ cryptocurrencies.  We commented on the practical aspect of the exercise such as benefits of diversification as the study identifies that the majority of the assets are driven by a common factor that is highly correlated to Bitcoin and has the highest contribution to the overall variance.   This outcome can aid portfolio selection and management.  Also, the remaining components of the decomposition reveal uncorrelated assets that can be used for risk hedging purposes.

\bibliographystyle{unsrt}  
\bibliography{bib}  
\newpage
\appendix

\section{Proofs of the EM Algorithm for Generalized Skew-t Probabilistic Principal Component Analysis} \label{appendix:proofs_GStS}

\begin{lemma}\label{lemma:Gaussian_Conditional}
Let $\mathbf{X}$ be a $d$-dimensional random vectors such that $\mathbf{X} = \big[ \mathbf{X}_1, \mathbf{X}_2\big]$ for $d_1$-dimensional subvector $\mathbf{X}_1$ and $d_2$-dimensional subvector $\mathbf{X}_2$ , $d_1 + d_2 = d$. If  $\mathbf{X}$ follows multivariate Gaussian distribution, that is
\begin{equation*}
\mathbf{X} = \Big[ \mathbf{X}_1, \mathbf{X}_2\Big] \sim \mathcal{N} \Bigg( \big[\bm{\mu}_1, \bm{\mu}_2 \big], \begin{bmatrix}
\bm{\Sigma}_{11} & \bm{\Sigma}_{12} \\
\bm{\Sigma}_{21} & \bm{\Sigma}_{22}
\end{bmatrix}\Bigg),
\end{equation*}
then for $i,j = 1,2$ and $i \neq j$, we have the following 
\begin{align*}
& \mathbf{X}_i | \mathbf{X}_j \sim \mathcal{N} \Big( \bm{\mu}_i + \big[ \mathbf{X}_j - \bm{\mu}_j \big] \bm{\Sigma}_{jj}^{-1} \bm{\Sigma}_{ji},  \bm{\Sigma}_{ii} -  \bm{\Sigma}_{ij}  \bm{\Sigma}_{jj}^{-1}  \bm{\Sigma}_{ji} \Big) \text{ for } \mathbf{X}_i \sim \mathcal{N} \big( \bm{\mu}_i,  \bm{\Sigma}_{ii} \big) , 
\end{align*}
where $\bm{\mu}_1, \bm{\mu}_2$ are $d_1$- and $d_2$-dimensional real valued vectors, respectively, $\bm{\Sigma}_{11}$ and $ \bm{\Sigma}_{11}$ are $d_1 \times d_1$ and $d_2 \times d_2$ symmetric, positive-definite real valued matrices, and $ \bm{\Sigma}_{12} =  \bm{\Sigma}_{21}^T$ is a $d_1 \times d_2$ real valued matrix. 
\end{lemma}
\begin{proof}
Please refer to Theorem 2.3.12 in \cite{Gupta2003a}.
\end{proof}

\begin{lemma}\label{lemma:tildePi_Xt_Yt_ind}
Let a $d$ dimensional observation vector $\mathbf{Y}_t$ be modelled as in \eqref{eq:ppca_model} under the Generalized Skew-t PPCA model defined in Section~\ref{sec:GStSPPCA} with the latent processes $\mathbf{X}_t, \ \bm{\epsilon}_t, \ \mathbf{U}_t$ and $\mathbf{V}_t$ following the assumptions defined in (5) and (6). The join probability function of the variables $\mathbf{Y}_t, \mathbf{X}_t, \mathbf{U}_t$ and $\mathbf{V}_t$ can be decomposed into the product of three functions
\begin{equation*}
\pi_{\mathbf{Y}_t, \mathbf{X}_t, \mathbf{U}_t, \mathbf{V}_t| \Psi}(\mathbf{y}_t, \mathbf{x}_t, \mathbf{u}_t,\mathbf{v}_t)   = \pi_{\mathbf{X}_t|\mathbf{Y}_t,\mathbf{U}_t,\mathbf{V}_t,\Psi}(\mathbf{x}_t)   \pi_{\mathbf{Y}_t|\mathbf{U}_t,\mathbf{V}_t,\Psi}(\mathbf{y}_t)  C( \mathbf{u}_t,\mathbf{v}_t;\Psi \big) ,
\end{equation*}
where $\pi_{\mathbf{X}_t|\mathbf{Y}_t,\mathbf{U}_t,\mathbf{V}_t,\Psi}(\mathbf{x}_t)$ and $\pi_{\mathbf{Y}_t|\mathbf{U}_t\mathbf{V}_t,\Psi}(\mathbf{x}_t)$ are the conditional probability function of the $k$ dimensional random vector $\mathbf{X}_t|\mathbf{Y}_t,\mathbf{U}_t,\mathbf{V}_t \sim \mathcal{N} \big(\bm{\mu}_{x,t}, \bm{\Sigma}_{x,t} \big)$ and $d$ dimensional random vector $\mathbf{Y}_t|\mathbf{U}_t,\mathbf{V}_t \sim \mathcal{N} \big(\bm{\mu}_{y,t}, \bm{\Sigma}_{y,t} \big)$, respectively, for the first and second central moments of the distributions given by
\begin{align*}
& \bm{\mu}_{x,t} = \Big( \big(\mathbf{y}_t - \bm{\mu} - \bm{\delta}_\epsilon \mathbf{D}_{\epsilon,t}^{-1}\big)\mathbf{D}_{\epsilon,t} \mathbf{W} + \sigma^2 \bm{\delta}_x \Big) \mathbf{M}_t^{-1} \text{ and } \bm{\Sigma}_{x,t} = \sigma^2 \mathbf{M}_t^{-1}, \\
&\bm{\mu}_{y,t} =   \bm{\mu} + \bm{\delta}_\epsilon \mathbf{D}_{\epsilon,t}^{-1} + \sigma^2 \bm{\delta}_x \mathbf{M}_t^{-1}\mathbf{W}^T  \mathbf{D}_{\epsilon,t}  \mathbf{N}_t^{-1} \text{ and } \bm{\Sigma}_{y,t} = \sigma^2  \mathbf{N}_t^{-1} 
\end{align*}
for $\mathbf{M}_t = \sigma^2 \mathbf{D}_{x,t} + \mathbf{W}^T \mathbf{D}_{\epsilon,t} \mathbf{W}$ and $ \mathbf{N}_t =  \mathbf{D}_{\epsilon,t}  -  \mathbf{D}_{\epsilon,t} \mathbf{W}\mathbf{M}_t^{-1} \mathbf{W}^T \mathbf{D}_{\epsilon,t}$. The function $C:\mathbb{R}_+^d \times \mathbb{R}_+^k \longrightarrow \mathbb{R}$ is given by
\begin{align*}
&C(\mathbf{u}_t, \mathbf{v}_t; \Psi) =  \pi_{\mathbf{U}_t | \Psi} ( \mathbf{u}_t)   \pi_{\mathbf{V}_t | \Psi} ( \mathbf{v}_t) \big(\sigma^2\big)^{\frac{k}{2}} \Big| \mathbf{D}_{\epsilon,t}\Big|^{ \frac{1}{2}}  \Big| \mathbf{D}_{x,t}\Big|^{\frac{1}{2}} \Big| \mathbf{M}_t\Big|^{- \frac{1}{2}} \Big|\mathbf{N}_t\Big|^{-\frac{1}{2}}  \\\notag
&\hspace{0.3cm} \times \exp \bigg\{ - \frac{1}{2} \bm{\delta}_x \Big( \mathbf{D}_{x,t}^{-1} - \sigma^2  \mathbf{M}_t^{-1}\mathbf{W}^T \mathbf{D}_{\epsilon,t} \mathbf{N}_t^{-1}\mathbf{D}_{\epsilon,t} \mathbf{W}\mathbf{M}_t^{-1}  - \sigma^2  \mathbf{M}_t^{-1} \Big)\bm{\delta}_x^T \bigg\}.
\end{align*}
\end{lemma}
\begin{proof}
Using the Chain Rule of probabilities we obtain the following decomposition of the likelihood function
\begin{align}\label{eq:tildePi_Xt_deriv_ind}\notag
& \pi_{\mathbf{Y}_t, \mathbf{X}_t, \mathbf{U}_t, \mathbf{V}_t| \Psi}(\mathbf{y}_t, \mathbf{x}_t, \mathbf{u}_t,\mathbf{v}_t)  = \pi_{\mathbf{Y}_t|\mathbf{X}_t, \mathbf{U}_t, \mathbf{V}_t,\Psi} ( \mathbf{y}_t) \cdot \pi_{\mathbf{X}_t| \mathbf{U}_t, \mathbf{V}_t,\Psi} (\mathbf{x}_t) \cdot \pi_{\mathbf{U}_t | \Psi} ( \mathbf{u}_t) \cdot \pi_{\mathbf{V}_t | \Psi} ( \mathbf{v}_t)  \\\notag
& = \pi_{\mathbf{U}_t | \Psi} ( \mathbf{u}_t) \cdot \pi_{\mathbf{V}_t | \Psi} ( \mathbf{v}_t) \  \big(2\pi\big)^{-\frac{d+k}{2}}\big(\sigma^2\big)^{-\frac{d}{2}} \Big| \mathbf{D}_{\epsilon,t} \Big|^{ \frac{1}{2}} \Big| \mathbf{D}_{x,t} \Big|^{\frac{1}{2}} \exp \Bigg\{ - \frac{1}{2} \big( \mathbf{x}_t - \bm{\mu}_{x,t} \big)\bm{\Sigma}_{x,t}^{-1} \big(\mathbf{x}_t - \bm{\mu}_{x,t}  \big)^T \Bigg\} \\
& \hspace{0.3cm} \times \exp \Bigg\{ - \frac{1}{2} \bigg( \bm{\delta}_x \mathbf{D}_{x,t}^{-1}\bm{\delta}_x^T  + \frac{1}{\sigma^2}\big(\mathbf{y}_t - \bm{\mu} - \bm{\delta}_\epsilon \mathbf{D}_{\epsilon,t}^{-1}\big)\mathbf{D}_{\epsilon,t} \big(\mathbf{y}_t - \bm{\mu} - \bm{\delta}_\epsilon \mathbf{D}_{\epsilon,t}^{-1} \big)^T -\bm{\mu}_{x,t} \bm{\Sigma}_{x,t}^{-1} \bm{\mu}_{x,t}^T \bigg)  \Bigg\},
\end{align}
where $\bm{\Sigma}_{x,t}  = \sigma^2 \mathbf{M}_t^{-1}, \ \bm{\mu}_{x,t}  = \Big( \big(\mathbf{y}_t - \bm{\mu} - \bm{\delta}_\epsilon \mathbf{D}_{\epsilon,t}^{-1}\big)\mathbf{D}_{\epsilon,t} \mathbf{W} + \sigma^2 \bm{\delta}_x \Big) \mathbf{M}_t^{-1}$ and $\mathbf{M}_t = \sigma^2 \mathbf{D}_{x,t} + \mathbf{W}^T \mathbf{D}_{\epsilon,t} \mathbf{W}$. Let us denote 
\begin{equation*}
\pi_{\mathbf{X}_t|\mathbf{Y}_t,\mathbf{U}_t,\mathbf{V}_t,\Psi}(\mathbf{x}_t)  = \big( 2 \pi\big)^{- \frac{k}{2}} \Big|\bm{\Sigma}_{x,t} \Big|^{-\frac{1}{2}} \exp \bigg\{ - \frac{1}{2} \big( \mathbf{x}_t - \bm{\mu}_{x,t} \big)\bm{\Sigma}_{x,t}^{-1} \big( \mathbf{x}_t - \bm{\mu}_{x,t} \big)^T \bigg\} ,
\end{equation*}
and remark that it is a probability function of $k$ dimensional Gaussian random variable with the mean vector $\bm{\mu}_{x,t}$ and the covariance matrix $\bm{\Sigma}_{x,t}$. The probability function $\pi_{\mathbf{X}_t|\mathbf{Y}_t,\mathbf{U}_t,\mathbf{V}_t,\Psi}(\mathbf{x}_t)$ contains all expressions with the vector $\mathbf{x}_t$ present in the probability function $\pi_{\mathbf{Y}_t, \mathbf{X}_t, \mathbf{U}_t, \mathbf{V}_t| \Psi}(\mathbf{y}_t, \mathbf{x}_t, \mathbf{u}_t,\mathbf{v}_t)$. In the next step we show how to obtain the formula for the probability function $\pi_{\mathbf{Y}_t|\mathbf{U}_t,\mathbf{V}_t,\Psi}(\mathbf{x}_t)$, which jointly with the function $\pi_{\mathbf{X}_t|\mathbf{Y}_t,\mathbf{U}_t,\mathbf{V}_t,\Psi}(\mathbf{x}_t)$,  contains all expressions with the vector $\mathbf{y}_t$. Using  \eqref{eq:tildePi_Xt_deriv_ind}, we obtain the following formulations
\begin{align*}
&\pi_{\mathbf{Y}_t, \mathbf{X}_t, \mathbf{U}_t, \mathbf{V}_t| \Psi}(\mathbf{y}_t, \mathbf{x}_t, \mathbf{u}_t,\mathbf{v}_t) =   \pi_{\mathbf{X}_t|\mathbf{Y}_t,\mathbf{U}_t,\mathbf{V}_t,\Psi}(\mathbf{x}_t) \ \pi_{\mathbf{U}_t | \Psi} ( \mathbf{u}_t) \  \pi_{\mathbf{V}_t | \Psi} ( \mathbf{v}_t) \  \big(2\pi\big)^{-\frac{d}{2}} \big(\sigma^2\big)^{-\frac{d-k}{2}}  \Big|\mathbf{D}_{\epsilon,t} \Big|^{\frac{1}{2}} \Big| \mathbf{D}_{x,t} \Big|^{ \frac{1}{2}} \Big| \mathbf{M}_t\Big|^{- \frac{1}{2}}  \\\notag
&\hspace{0.3cm} \times   \exp \bigg\{ - \frac{1}{2} \bm{\delta}_x \mathbf{D}_{x,t}^{-1}\bm{\delta}_x^T -\frac{1}{2}\big( \underbrace{\mathbf{y}_t  - \bm{\mu} - \bm{\delta}_\epsilon \mathbf{D}_{\epsilon,t}^{-1} - \sigma^2 \bm{\delta}_x \mathbf{M}_t^{-1}\mathbf{W}^T \mathbf{D}_{\epsilon,t} \mathbf{N}_t^{-1}}_{\mathbf{y}_t  - \bm{\mu}_{y,t}}\big)\bm{\Sigma}_{y,t}^{-1}\big( \underbrace{\mathbf{y}_t - \bm{\mu} - \bm{\delta}_\epsilon \mathbf{D}_{\epsilon,t}^{-1} - \sigma^2 \bm{\delta}_x \mathbf{M}_t^{-1}\mathbf{W}^T\mathbf{D}_{\epsilon,t} \mathbf{N}_t^{-1}}_{\mathbf{y}_t  - \bm{\mu}_{y,t}}\big)^T\Bigg\} \\\notag
&\hspace{0.3cm}  \times  \exp \Bigg\{ -\frac{1}{2} \bigg( -\sigma^2 \bm{\delta}_x \mathbf{M}_t^{-1}\mathbf{W}^T \mathbf{D}_{\epsilon,t} \mathbf{N}_t^{-1}\mathbf{D}_{\epsilon,t} \mathbf{W}\mathbf{M}_t^{-1}\bm{\delta}_x^T  -  \sigma^2 \bm{\delta}_x \mathbf{M}_t^{-1} \bm{\delta}_x^T\bigg)\Bigg\},
\end{align*}
where $\bm{\Sigma}_{y,t} = \sigma^2 \mathbf{N}_t^{-1} \text{ for } \mathbf{N}_t = \mathbf{D}_{\epsilon,t} - \mathbf{D}_{\epsilon,t}\mathbf{W}\mathbf{M}_t^{-1} \mathbf{W}^T \mathbf{D}_{\epsilon,t}$. Let us denote $\bm{\mu}_{y,t} =   \bm{\mu} + \bm{\delta}_\epsilon \mathbf{D}_{\epsilon,t}^{-1} + \sigma^2 \bm{\delta}_x \mathbf{M}_t^{-1}\mathbf{W}^T \mathbf{D}_{\epsilon,t} \mathbf{N}_t^{-1}$ and define 
\begin{equation*}
\pi_{\mathbf{Y}_t|\mathbf{U}_t,\mathbf{V}_t,\Psi}(\mathbf{y}_t)  = \big( 2 \pi\big)^{- \frac{d}{2}} \Big|\bm{\Sigma}_{y,t} \Big|^{-\frac{1}{2}} \exp \Bigg\{ - \frac{1}{2} \big( \mathbf{y}_t - \bm{\mu}_{y,t} \big)\bm{\Sigma}_{y,t}^{-1} \big( \mathbf{y}_t - \bm{\mu}_{y,t} \big)^T \Bigg\}.
\end{equation*}
The probability function $\pi_{\mathbf{Y}_t|\mathbf{U}_t,\mathbf{V}_t,\Psi}(\mathbf{y}_t) $ is a density of the $d$-dimensional random variable which follows the Gaussian distribution with the mean vector $\bm{\mu}_{y,t}$ and the covariance matrix $\bm{\Sigma}_{y,t}$. Using the definition of the density function we obtain that
\begin{align*}
&\pi_{\mathbf{Y}_t, \mathbf{X}_t, \mathbf{U}_t, \mathbf{V}_t| \Psi}(\mathbf{y}_t, \mathbf{x}_t, \mathbf{u}_t,\mathbf{v}_t) = 
 \pi_{\mathbf{X}_t|\mathbf{Y}_t,\mathbf{U}_t,\mathbf{V}_t,\Psi}(\mathbf{x}_t) \ \pi_{\mathbf{Y}_t |\mathbf{U}_t,\mathbf{V}_t,\Psi}(\mathbf{y}_t) \ \pi_{\mathbf{U}_t | \Psi} ( \mathbf{u}_t) \  \pi_{\mathbf{V}_t | \Psi} ( \mathbf{v}_t)  \\\notag
&\hspace{0.3cm}  \times \big(\sigma^2\big)^{-\frac{d-k}{2}}   \Big| \mathbf{D}_{\epsilon,t} \Big|^{ \frac{1}{2}} \mathbf{D}_{x,t} \Big|^{ \frac{1}{2}} \Big| \mathbf{M}_t\Big|^{- \frac{1}{2}} \Big| \sigma^2 \mathbf{N}_t^{-1}\Big|^{\frac{1}{2}}  \exp \Bigg\{ - \frac{1}{2} \bm{\delta}_x \Big( \mathbf{D}_{x,t}^{-1} - \sigma^2  \mathbf{M}_t^{-1}\mathbf{W}^T \mathbf{D}_{\epsilon,t} \mathbf{N}_t^{-1}\mathbf{D}_{\epsilon,t} \mathbf{W}\mathbf{M}_t^{-1}  - \sigma^2  \mathbf{M}_t^{-1} \Big)\bm{\delta}_x^T \Bigg\}.
\end{align*}
Therefore, denoting
\begin{align*}
&C(\mathbf{u}_t, \mathbf{v}_t; \Psi) =  \pi_{\mathbf{U}_t | \Psi} ( \mathbf{u}_t) \cdot  \pi_{\mathbf{V}_t | \Psi} ( \mathbf{v}_t) \big(\sigma^2\big)^{\frac{k}{2}} \Big| \mathbf{D}_{\epsilon,t}\Big|^{ \frac{1}{2}}  \Big| \mathbf{D}_{x,t}\Big|^{\frac{1}{2}} \Big| \mathbf{M}_t\Big|^{- \frac{1}{2}} \Big|\mathbf{N}_t\Big|^{-\frac{1}{2}} \\
& \hspace{1cm} \times \exp \bigg\{ - \frac{1}{2} \bm{\delta}_x \Big( \mathbf{D}_{x,t}^{-1} - \sigma^2  \mathbf{M}_t^{-1}\mathbf{W}^T \mathbf{D}_{\epsilon,t} \mathbf{N}_t^{-1}\mathbf{D}_{\epsilon,t} \mathbf{W}\mathbf{M}_t^{-1}  - \sigma^2  \mathbf{M}_t^{-1} \Big)\bm{\delta}_x^T \bigg\},
\end{align*}
we obtain the following decomposition of the joint probability function 
\begin{align*}
\pi_{\mathbf{Y}_t, \mathbf{X}_t, \mathbf{U}_t, \mathbf{V}_t| \Psi}(\mathbf{y}_t, \mathbf{x}_t, \mathbf{u}_t,\mathbf{v}_t) = \pi_{\mathbf{X}_t|\mathbf{Y}_t,\mathbf{U}_t,\mathbf{V}_t,\Psi}(\mathbf{x}_t) \pi_{\mathbf{Y}_t |\mathbf{U}_t,\mathbf{V}_t,\Psi}(\mathbf{y}_t) C(\mathbf{u}_t, \mathbf{v}_t; \Psi).
\end{align*}
\end{proof}

\begin{lemma} \label{lemma:w_ind}
Let us recall the probability function  $\pi_{\mathbf{X}_t|\mathbf{Y}_t, \mathbf{U}_t, \mathbf{V}_t, \Psi}(\mathbf{x}_t)$ of a $k$-dimensional Gaussian random variable $\mathbf{X}_t$ with the covariance matrix
$\bm{\Sigma}_{x,t}$ and the mean vector $\bm{\mu}_{x,t}$ defined in Lemma \ref{lemma:tildePi_Xt_Yt_ind}. Under the assumptions of the Generalized Skew-t PPCA specified in Section~\ref{sec:GStSPPCA}, the solution to the following integration problem
\begin{equation*}
\int_{\mathbb{R}^k} \log \Big(\pi_{\mathbf{Y}_t, \mathbf{X}_t, \mathbf{U}_t, \mathbf{V}_t| \Psi^*}(\mathbf{y}_t, \mathbf{x}_t, \mathbf{u}_t,\mathbf{v}_t) \Big) \pi_{\mathbf{X}_t|\mathbf{Y}_t, \mathbf{U}_t, \mathbf{V}_t, \Psi}(\mathbf{x}_t)\ d \mathbf{x}_t,
\end{equation*}
is equal to the function $w: \mathbb{R}^d \times \mathbb{R}^d_+ \times \mathbb{R}^k_+ \longrightarrow \mathbb{R}$ defined as following
\footnotesize
\begin{align*}
& w(\mathbf{y}_t,\mathbf{u}_t, \mathbf{v}_t;\Psi,\Psi^*)= \log \pi_{\mathbf{U}_t| \Psi^*}( \mathbf{u}_t)  + \log \pi_{\mathbf{V}_t| \Psi^*}( \mathbf{v}_t) - \frac{k+d}{2}\log 2\pi - \frac{d}{2} \log \sigma^{*2} + \frac{1}{2} \sum_{i = 1}^d \log u_t^i + \frac{1}{2} \sum_{j = 1}^k \log v_t^j \\\notag
& \hspace{0.3cm} - \frac{1}{2} \bm{\delta}_x^* \mathbf{D}_{x,t}^{-1} \bm{\delta}_x^{*} - \frac{1}{2\sigma^{*2}} \Big( \mathbf{y}_t - \bm{\mu}^* - \bm{\delta}_\epsilon^*\mathbf{D}_{\epsilon,t}^{-1} \Big) \mathbf{D}_{\epsilon,t} \Big( \mathbf{y}_t - \bm{\mu}^* - \bm{\delta}_\epsilon^*\mathbf{D}_{\epsilon,t}^{-1} \Big)^T   - \frac{1}{2} \tr \bigg\{ \sigma^2 \mathbf{M}_t^{-1} \Big( \frac{1}{\sigma^{*2}}\mathbf{W}^{*T} \mathbf{D}_{\epsilon,t} \mathbf{W}^{*} + \mathbf{D}_{x,t} \Big) \bigg\}   \\\notag
&\hspace{0.3cm}+  \bigg( \Big(\mathbf{y}_t - \bm{\mu} - \bm{\delta}_\epsilon \mathbf{D}_{\epsilon,t}^{-1}\Big)\mathbf{D}_{\epsilon,t} \mathbf{W} + \sigma^2 \bm{\delta}_x \bigg) \mathbf{M}_t^{-1} \bigg( \frac{1}{\sigma^{*2}}\mathbf{W}^{*T} \mathbf{D}_{\epsilon,t} \Big( \mathbf{y}_t - \bm{\mu}^* - \bm{\delta}_\epsilon^*\mathbf{D}_{\epsilon,t}^{-1} \Big)^T + \bm{\delta}_x ^{*T}\bigg) \\\notag
& \hspace{0.3cm} - \frac{1}{2} \tr \bigg\{ \mathbf{M}_t^{-1}\Big( \big(\mathbf{y}_t - \bm{\mu} - \bm{\delta}_\epsilon \mathbf{D}_{\epsilon,t}^{-1}\big)\mathbf{D}_{\epsilon,t} \mathbf{W} + \sigma^2 \bm{\delta}_x \Big)^T \Big( \big(\mathbf{y}_t - \bm{\mu} - \bm{\delta}_\epsilon \mathbf{D}_{\epsilon,t}^{-1}\big)\mathbf{D}_{\epsilon,t} \mathbf{W} + \sigma^2 \bm{\delta}_x \Big) \mathbf{M}_t^{-1} \Big( \frac{1}{\sigma^{*2}}\mathbf{W}^{*T} \mathbf{D}_{\epsilon,t} \mathbf{W}^{*} + \mathbf{D}_{x,t} \Big) \bigg\} 
\end{align*}\normalsize
where $\mathbf{M}_t = \sigma^2 \mathbf{D}_{x,t} + \mathbf{W}^T \mathbf{D}_{\epsilon,t} \mathbf{W}$.
\end{lemma}
\begin{proof}
Let us recall the probability function $\pi_{\mathbf{X}_t|\mathbf{Y}_t, \mathbf{U}_t, \mathbf{V}_t, \Psi}(\mathbf{x}_t)$ defined in Lemma \ref{lemma:tildePi_Xt_Yt_ind} which is a density of a $k$-dimensional Gaussian random variable $\mathbf{X}_t$ with the covariance matrix $\bm{\Sigma}_{x,t} = \sigma^2 \mathbf{M}_t^{-1}$ and the mean vector $\bm{\mu}_{x,t} = \Big( \big(\mathbf{y}_t - \bm{\mu} - \bm{\delta}_\epsilon \mathbf{D}_{\epsilon,t}^{-1}\big)\mathbf{D}_{\epsilon,t} \mathbf{W} + \sigma^2 \bm{\delta}_x \Big) \mathbf{M}_t^{-1}$, where $\mathbf{M}_t = \sigma^2 \mathbf{D}_{x,t} + \mathbf{W}^T \mathbf{D}_{\epsilon,t} \mathbf{W}$. Hence, the solution to the integration problem is given by
\begin{align*}
& w(\mathbf{y}_t,\mathbf{u}_t, \mathbf{v}_t;\Psi,\Psi^*) = \int_{\mathbb{R}^k} \log \Big(\pi_{\mathbf{Y}_t, \mathbf{X}_t, \mathbf{U}_t, \mathbf{V}_t| \Psi^*}(\mathbf{y}_t, \mathbf{x}_t, \mathbf{u}_t,\mathbf{v}_t) \Big) \pi_{\mathbf{X}_t|\mathbf{Y}_t, \mathbf{U}_t, \mathbf{V}_t, \Psi}(\mathbf{x}_t) \ d \mathbf{x}_t \\\notag
& \hspace{0.2cm} = \int_{\mathbb{R}^k} \bigg( \log \pi_{\mathbf{Y}_t| \mathbf{X}_t, \mathbf{U}_t, \mathbf{V}_t| \Psi^*}(\mathbf{y}_t) + \log \pi_{ \mathbf{X}_t| \mathbf{U}_t, \mathbf{V}_t, \Psi^*}(\mathbf{x}_t) + \log \pi_{\mathbf{U}_t| \Psi^*}( \mathbf{u}_t)  + \log \pi_{\mathbf{V}_t| \Psi^*}( \mathbf{v}_t)   \bigg)  \pi_{\mathbf{X}_t|\mathbf{Y}_t, \mathbf{U}_t, \mathbf{V}_t, \Psi}(\mathbf{x}_t)  \ d \mathbf{x}_t \\\notag
& \hspace{0.2cm} = \log \pi_{\mathbf{U}_t| \Psi^*}( \mathbf{u}_t)  + \log \pi_{\mathbf{V}_t| \Psi^*}( \mathbf{v}_t) - \frac{k+d}{2}\log 2\pi - \frac{d}{2} \log \sigma^{*2} + \frac{1}{2} \sum_{i = 1}^d \log u_t^i + \frac{1}{2} \sum_{j = 1}^k \log v_t^j- \frac{1}{2} \bm{\delta}_x^* \mathbf{D}_{x,t}^{-1} \bm{\delta}_x^{*T} \\\notag
& \hspace{0.3cm} - \frac{1}{2\sigma^{*2}} \Big( \mathbf{y}_t - \bm{\mu}^* - \bm{\delta}_\epsilon^*\mathbf{D}_{\epsilon,t}^{-1} \Big) \mathbf{D}_{\epsilon,t} \Big( \mathbf{y}_t - \bm{\mu}^* - \bm{\delta}_\epsilon^*\mathbf{D}_{\epsilon,t}^{-1} \Big)^T   - \frac{1}{2} \tr \bigg\{ \sigma^2 \mathbf{M}_t^{-1} \Big( \frac{1}{\sigma^{*2}}\mathbf{W}^{*T} \mathbf{D}_{\epsilon,t} \mathbf{W}^{*} + \mathbf{D}_{x,t} \Big) \bigg\}   \\\notag
&\hspace{0.3cm}+  \bigg( \Big(\mathbf{y}_t - \bm{\mu} - \bm{\delta}_\epsilon \mathbf{D}_{\epsilon,t}^{-1}\Big)\mathbf{D}_{\epsilon,t} \mathbf{W} + \sigma^2 \bm{\delta}_x \bigg) \mathbf{M}_t^{-1} \bigg( \frac{1}{\sigma^{*2}}\mathbf{W}^{*T} \mathbf{D}_{\epsilon,t} \Big( \mathbf{y}_t - \bm{\mu}^* - \bm{\delta}_\epsilon^*\mathbf{D}_{\epsilon,t}^{-1} \Big)^T + \bm{\delta}_x ^{*T}\bigg) \\\notag
& \hspace{0.3cm} - \frac{1}{2} \tr \bigg\{ \mathbf{M}_t^{-1}\Big( \big(\mathbf{y}_t - \bm{\mu} - \bm{\delta}_\epsilon \mathbf{D}_{\epsilon,t}^{-1}\big)\mathbf{D}_{\epsilon,t} \mathbf{W} + \sigma^2 \bm{\delta}_x \Big)^T \Big( \big(\mathbf{y}_t - \bm{\mu} - \bm{\delta}_\epsilon \mathbf{D}_{\epsilon,t}^{-1}\big)\mathbf{D}_{\epsilon,t} \mathbf{W} + \sigma^2 \bm{\delta}_x \Big) \mathbf{M}_t^{-1} \Big( \frac{1}{\sigma^{*2}}\mathbf{W}^{*T} \mathbf{D}_{\epsilon,t} \mathbf{W}^{*} + \mathbf{D}_{x,t} \Big) \bigg\} 
\end{align*}
\end{proof}

\begin{lemma}\label{lemma:tildePi_Ytmissing_Conditional_ind}
Let us consider the partition of the observation vector $\mathbf{Y}_t = [\mathbf{Y}_t^o, \mathbf{Y}_t^m]$ into the subvector with observed and missing entries, $\mathbf{Y}_t^o$ and $\mathbf{Y}_t^m$, respectively. The observation vector $\mathbf{Y}_t$ is modelled as in \eqref{eq:ppca_model} under the assumptions stated in Subsection~\ref{ssec:PPCA_GStS_ind}. The conditional distribution of the random vector $\mathbf{Y}_t|\mathbf{U}_t, \mathbf{V}_t, \Psi$, derived in Lemma \ref{lemma:tildePi_Xt_Yt_ind}, is Gaussian with the mean vector $\bm{\mu}_{y,t} =   \bm{\mu} + \bm{\delta}_\epsilon \mathbf{D}_{\epsilon,t}^{-1} + \sigma^2 \bm{\delta}_x \mathbf{M}_t^{-1}\mathbf{W}^T \mathbf{N}_t^{-1}$ and the covariance matrix $\bm{\Sigma}_{y,t} = \sigma^2 \mathbf{D}_{\epsilon,t}^{-1} \mathbf{N}_t^{-1}$, where  
$\mathbf{M}_t = \sigma^2 \mathbf{D}_{x,t} + \mathbf{W}^T \mathbf{D}_{\epsilon,t} \mathbf{W}$ and $ \mathbf{N}_t = \Big( \mathbb{I}_d - \mathbf{W}\mathbf{M}_t^{-1} \mathbf{W}^T \mathbf{D}_{\epsilon,t}\Big)$. The conditional distribution 
$\mathbf{Y}_t^m|\mathbf{Y}_t^o,\mathbf{U}_t, \mathbf{V}_t, \Psi$ is also Gaussian with the mean vector $\tilde{\bm{\mu}}_{Y^m,t}$ and the covariance matrix $\tilde{\bm{\Sigma}}_{Y^m,t}$ such that 
\begin{align*}
& \tilde{\bm{\mu}}_{Y^m,t}= \bm{\mu}_{y,t}^m + \big( \mathbf{Y}_t^o - \bm{\mu}_{y,t}^m \big)\bm{\Sigma}_{y,t}^{oo \ -1} \bm{\Sigma}_{y,t}^{om} \\
& \tilde{\bm{\Sigma}}_{Y^m,t}= \bm{\Sigma}_{y,t}^{mm} - \bm{\Sigma}_{y,t}^{mo} \bm{\Sigma}_{y,t}^{oo \ -1} \bm{\Sigma}_{y,t}^{om}. 
\end{align*}
In addition, the marginal distribution of $\mathbf{Y}_t^o|\mathbf{U}_t, \mathbf{V}_t, \Psi$ is Gaussian with the mean vector $\bm{\mu}_{y,t}^o$ and the covariance matrix $\bm{\Sigma}_{y,t}^{oo}$. The subvectors $ \bm{\mu}_{y,t}^o$ and  $\bm{\mu}_{y,t}^m$ contain elements of the vector $ \bm{\mu}_{y,t}$ which correspond to the observed or missing entries of the observation vector $\mathbf{Y}_t$, respectively, and are $d_o$ and $d_m$-dimensional. The square matrices $\bm{\Sigma}_{y,t}^{oo}$ and $\bm{\Sigma}_{y,t}^{mm}$
contain elements of the matrix $ \bm{\Sigma}_{y,t}$ which correspond by rows and columns to the observed or missing entries of the observation vector $\mathbf{Y}_t$, respectively, and are $d_o \times d_o$ and $d_m \times d_m$-dimensional. The non-square matrix $\bm{\Sigma}_{y,t}^{mo} = \bm{\Sigma}_{y,t}^{om \ T}$ contains elements of the matrix $\bm{\Sigma}_{y,t}$ which correspond by rows to the missing and by columns to the observed entries of the observation vector $\mathbf{Y}_t$, and is $d_m \times d_o$-dimensional.
\end{lemma}
\begin{proof}
Please refer to Lemma \ref{lemma:tildePi_Xt_Yt_ind} for the derivation of the conditional distribution $\mathbf{Y}_t|\mathbf{U}_t, \mathbf{V}_t, \Psi$ and apply the formulas for the standard conditional Gaussian distribution from Lemma \ref{lemma:Gaussian_Conditional}. 
\end{proof}

\begin{lemma}\label{lemma:integral_w_ind}
Let us consider the partition of the observation vector $\mathbf{Y}_t = [\mathbf{Y}_t^o, \mathbf{Y}_t^m]$ into the subvector with observed and missing entries, $\mathbf{Y}_t^o$ and $\mathbf{Y}_t^m$, respectively. The observation vector $\mathbf{Y}_t$ is modelled as in \eqref{eq:ppca_model} under the assumptions of Independent Generalized Skew-t PPCA stated in Subsection~\ref{ssec:PPCA_GStS_ind}. Let us recall the function $ w(\mathbf{y}_t, \mathbf{u}_t, \mathbf{v}_t; \Psi,\Psi^*)$ defined as in Lemma \ref{lemma:w_ind} and the conditional distribution of the random vector $\mathbf{Y}_t^m|\mathbf{Y}_t^o,\mathbf{U}_t, \mathbf{V}_t, \Psi$ specified in Lemma \ref{lemma:tildePi_Ytmissing_Conditional_ind}. Given a realisation of the vector $\mathbf{Y}_t$ at time points $t = 1, \ldots, N$ with observable components  $\mathbf{y}_t^o$, the function $v$ defined by the following integration problem
\begin{equation}\label{eq:proof_lemma_integral_w} 
v(\mathbf{y}_t^o, \mathbf{u}_t, \mathbf{v}_t; \Psi,\Psi^*)  =   \int_{\mathbb{R}^{d_m}}  w(\mathbf{y}_t, \mathbf{u}_t, \mathbf{v}_t; \Psi,\Psi^*)  \pi_{\mathbf{Y}_t^m|\mathbf{Y}_t^o,\mathbf{U}_t,\mathbf{V}_t,\Psi}(\mathbf{y}_t^m) \ d \mathbf{y}^m_t ,
\end{equation}
can be expressed as \footnotesize
\begin{align*}
& v(\mathbf{y}_t^o, \mathbf{u}_t, \mathbf{v}_t; \Psi,\Psi^*)  = \log \pi_{\mathbf{U}_t| \Psi^*}( \mathbf{u}_t)  + \log \pi_{\mathbf{V}_t| \Psi^*}( \mathbf{v}_t) - \frac{k+d}{2}\log 2\pi - \frac{d}{2} \log \sigma^{*2} + \frac{1}{2} \sum_{i = 1}^d \log u_t^i + \frac{1}{2} \sum_{j = 1}^k \log v_t^j   \\\notag
& \hspace{0.3cm}  - \frac{1}{2} \bm{\delta}_x^* \mathbf{D}_{x,t}^{-1} \bm{\delta}_x^{*T}- \frac{1}{2 \sigma^{*2}} \tr \Big\{  \mathbb{E}_{\mathbf{Y}_t^m|\mathbf{Y}_t^o,\mathbf{U}_t,\mathbf{V}_t, \Psi}\big[ \mathbf{Y}_t^T \mathbf{Y}_t\big]    \mathbf{D}_{\epsilon,t} \Big\} +  \frac{1}{ \sigma^{*2}}  \mathbb{E}_{\mathbf{Y}_t^m|\mathbf{Y}_t^o,\mathbf{U}_t,\mathbf{V}_t, \Psi}\big[ \mathbf{Y}_t\big]  \Big(\mathbf{D}_{\epsilon,t}\bm{\mu}^{*T}+ \bm{\delta}_\epsilon^{*T} \Big) \\\notag
&\hspace{0.3cm}-  \frac{1}{ 2\sigma^{*2}} \bm{\mu}^{*} \mathbf{D}_{\epsilon,t} \bm{\mu}^{*T} - \frac{1}{\sigma^{*2}} \bm{\mu}^* \bm{\delta}_\epsilon^{*T} - \frac{1}{2 \sigma^{*2}} \bm{\delta}_\epsilon^* \mathbf{D}_{\epsilon,t}^{-1}  \bm{\delta}_\epsilon^{*T} + \frac{1}{\sigma^{*2}} \tr \Big\{ \mathbb{E}_{\mathbf{Y}_t^m|\mathbf{Y}_t^o,\mathbf{U}_t,\mathbf{V}_t, \Psi}\big[ \mathbf{Y}_t^T \mathbf{Y}_t\big]  \mathbf{D}_{\epsilon,t}  \mathbf{W} \mathbf{M}_t^{-1}\mathbf{W}^{*T} \mathbf{D}_{\epsilon,t} \Big\}  \\\notag
&\hspace{0.3cm} +  \mathbb{E}_{\mathbf{Y}_t^m|\mathbf{Y}_t^o,\mathbf{U}_t,\mathbf{V}_t, \Psi}\big[ \mathbf{Y}_t\big] \mathbf{D}_{\epsilon,t}  \mathbf{W} \mathbf{M}_t^{-1} \bigg(  - \frac{1}{\sigma^{*2}}\mathbf{W}^{*T} \mathbf{D}_{\epsilon,t}\bm{\mu}^{*T} - \frac{1}{\sigma^{*2}}\mathbf{W}^{*T} \bm{\delta}_\epsilon^{*T}+ \bm{\delta}_x ^{*T} \bigg) \\\notag
&\hspace{0.3cm} + \frac{1}{\sigma^{*2}} \bigg( - \bm{\mu}\mathbf{D}_{\epsilon,t} \mathbf{W} - \bm{\delta}_\epsilon \mathbf{W} + \sigma^2 \bm{\delta}_x \bigg)\mathbf{M}_t^{-1} \mathbf{W}^{*T} \mathbf{D}_{\epsilon,t}  \mathbb{E}_{\mathbf{Y}_t^m|\mathbf{Y}_t^o,\mathbf{U}_t,\mathbf{V}_t, \Psi}\big[ \mathbf{Y}_t\big]^T  \\\notag
&\hspace{0.3cm}+  \bigg( - \bm{\mu}\mathbf{D}_{\epsilon,t} \mathbf{W} - \bm{\delta}_\epsilon \mathbf{W} + \sigma^2 \bm{\delta}_x \bigg) \mathbf{M}_t^{-1} \bigg(- \frac{1}{\sigma^{*2}}\mathbf{W}^{*T} \mathbf{D}_{\epsilon,t}\bm{\mu}^{*T} - \frac{1}{\sigma^{*2}}\mathbf{W}^{*T} \bm{\delta}_\epsilon^{*T}+ \bm{\delta}_x ^{*T} \bigg) \\\notag
& \hspace{0.3cm} - \frac{1}{2} \tr \bigg\{  \mathbf{M}_t^{-1} \mathbf{W} ^T\mathbf{D}_{\epsilon,t}  \mathbb{E}_{\mathbf{Y}_t^m|\mathbf{Y}_t^o,\mathbf{U}_t,\mathbf{V}_t, \Psi}\big[ \mathbf{Y}_t^T \mathbf{Y}_t\big]  \mathbf{D}_{\epsilon,t} \mathbf{W}\mathbf{M}_t^{-1}  \Big( \frac{1}{\sigma^{*2}}\mathbf{W}^{*T} \mathbf{D}_{\epsilon,t} \mathbf{W}^{*} + \mathbf{D}_{x,t} \Big) \bigg\} \\\notag
& \hspace{0.3cm} -\tr \bigg\{\mathbf{M}_t^{-1} \mathbf{W} ^T\mathbf{D}_{\epsilon,t}   \mathbb{E}_{\mathbf{Y}_t^m|\mathbf{Y}_t^o,\mathbf{U}_t,\mathbf{V}_t, \Psi}\big[ \mathbf{Y}_t\big]^T  \Big( - \bm{\mu}\mathbf{D}_{\epsilon,t} \mathbf{W} - \bm{\delta}_\epsilon \mathbf{W} + \sigma^2 \bm{\delta}_x \Big) \mathbf{M}_t^{-1}  \Big( \frac{1}{\sigma^{*2}}\mathbf{W}^{*T} \mathbf{D}_{\epsilon,t} \mathbf{W}^{*} + \mathbf{D}_{x,t} \Big) \bigg\} \\\notag
& \hspace{0.3cm} - \frac{1}{2} \tr \bigg\{  \mathbf{M}_t^{-1}\Big( - \bm{\mu}\mathbf{D}_{\epsilon,t} \mathbf{W} - \bm{\delta}_\epsilon \mathbf{W} + \sigma^2 \bm{\delta}_x \Big)^T\Big( - \bm{\mu}\mathbf{D}_{\epsilon,t} \mathbf{W} - \bm{\delta}_\epsilon \mathbf{W} + \sigma^2 \bm{\delta}_x \Big) \mathbf{M}_t^{-1}  \Big( \frac{1}{\sigma^{*2}}\mathbf{W}^{*T} \mathbf{D}_{\epsilon,t} \mathbf{W}^{*} + \mathbf{D}_{x,t} \Big) \bigg\} \\\notag
&\hspace{0.3cm}  - \frac{\sigma^2}{2} \tr \bigg\{ \mathbf{M}_t^{-1} \Big( \frac{1}{\sigma^{*2}}\mathbf{W}^{*T} \mathbf{D}_{\epsilon,t} \mathbf{W}^{*} + \mathbf{D}_{x,t} \Big) \bigg\}.
\end{align*}\normalsize
Using the mean $\tilde{\bm{\mu}}_{Y^m,t}$ and the covariance matrix $ \tilde{\bm{\Sigma}}_{Y^m,t}$ of the random vector $\mathbf{Y}_t^m|\mathbf{Y}_t^o,\mathbf{U}_t,\mathbf{V}_t, \Psi$  from Lemma \ref{lemma:tildePi_Ytmissing_Conditional_ind}, the first and second non-central moments of the vector $\mathbf{Y}_t$ are the following
\begin{align*}
& \mathbb{E}_{\mathbf{Y}_t^m|\mathbf{Y}_t^o,\mathbf{U}_t,\mathbf{V}_t, \Psi}\big[ \mathbf{Y}_t\big] = \big[\mathbf{y}_t^o , \ \tilde{\bm{\mu}}_{Y^m,t} \big], \\
&  \mathbb{E}_{\mathbf{Y}_t^m|\mathbf{Y}_t^o,\mathbf{U}_t,\mathbf{V}_t, \Psi}\big[ \mathbf{Y}_t^T \mathbf{Y}_t\big] = \begin{bmatrix}
\mathbf{0} & \mathbf{0} \\
\mathbf{0} & \tilde{\bm{\Sigma}}_{Y^m,t}
\end{bmatrix} + \mathbb{E}_{\mathbf{Y}_t^m|\mathbf{Y}_t^o,\mathbf{U}_t,\mathbf{V}_t, \Psi}\big[ \mathbf{Y}_t\big]^T \mathbb{E}_{\mathbf{Y}_t^m|\mathbf{Y}_t^o,\mathbf{U}_t,\mathbf{V}_t, \Psi}\big[ \mathbf{Y}_t\big].
\end{align*}
\end{lemma}
\begin{proof}
Let us recall, that the observation vector $\mathbf{Y}_t=[\mathbf{Y}_t^o,\mathbf{Y}_t^m]$ is partitioned into the $d_o$-dimensional subvector with observed entries, $\mathbf{Y}_t^o$ and the $d_m$-dimensional subvector with unobserved entries, $\mathbf{Y}_t^m$, such that $d_m = d - d_m$. Given the function $w(\mathbf{y}_t, \mathbf{u}_t, \mathbf{v}_t; \Psi,\Psi^*) $ defined in Lemma \ref{lemma:w_ind}, the solution to the integration problem from  \eqref{eq:proof_lemma_integral_w} is derived as follows
\begin{align*}
& v(\mathbf{y}_t^o, \mathbf{u}_t, \mathbf{v}_t; \Psi,\Psi^*)  = \int_{\mathbb{R}^{d_m}}  w(\mathbf{y}_t, \mathbf{u}_t, \mathbf{v}_t; \Psi,\Psi^*) \ \pi_{\mathbf{Y}_t^m|\mathbf{Y}_t^o,\mathbf{U}_t,\mathbf{V}_t,\Psi}(\mathbf{y}_t^m) \ d \mathbf{y}^m_t .
\end{align*}
The probability function $ \pi_{\mathbf{Y}_t^m|\mathbf{Y}_t^o,\mathbf{U}_t,\mathbf{V}_t,\Psi}(\mathbf{y}_t^m)$ is specified in Lemma \ref{lemma:tildePi_Ytmissing_Conditional_ind} as well as the first and second non-central moments of the distribution $\mathbf{Y}_t^m|\mathbf{Y}_t^o, \mathbf{U}_t, \mathbf{V}_t, \Psi$. Consequently we have the following solutions to the integrations problems 
\begin{align*}
&  \int_{\mathbb{R}^{d_m}}  \mathbf{y}_t  \ \pi_{\mathbf{Y}_t^m|\mathbf{Y}_t^o,\mathbf{U}_t,\mathbf{V}_t,\Psi}(\mathbf{y}_t^m) \ d \mathbf{y}^m_t =  \mathbb{E}_{\mathbf{Y}_t^m|\mathbf{Y}_t^o,\mathbf{U}_t,\mathbf{V}_t, \Psi}\big[ \mathbf{Y}_t\big] = \big[\mathbf{Y}_t^o , \ \tilde{\bm{\mu}}_{Y^m,t} \big], \\
&  \int_{\mathbb{R}^{d_m}}  \mathbf{y}_t^T \mathbf{y}_t  \ \pi_{\mathbf{Y}_t^m|\mathbf{Y}_t^o,\mathbf{U}_t,\mathbf{V}_t,\Psi}(\mathbf{y}_t^m) \ d \mathbf{y}^m_t=  \mathbb{E}_{\mathbf{Y}_t^m|\mathbf{Y}_t^o,\mathbf{U}_t,\mathbf{V}_t, \Psi}\big[ \mathbf{Y}_t^T \mathbf{Y}_t\big] = \begin{bmatrix}
\mathbf{0} & \mathbf{0} \\
\mathbf{0} & \tilde{\bm{\Sigma}}_{Y^m,t}
\end{bmatrix} + \mathbb{E}_{\mathbf{Y}_t^m|\mathbf{Y}_t^o,\mathbf{U}_t,\mathbf{V}_t, \Psi}\big[ \mathbf{Y}_t\big]^T \mathbb{E}_{\mathbf{Y}_t^m|\mathbf{Y}_t^o,\mathbf{U}_t,\mathbf{V}_t, \Psi}\big[ \mathbf{Y}_t\big],
\end{align*}
for $\tilde{\bm{\mu}}_{Y^m,t}= \bm{\mu}_{y,t}^m + \big( \mathbf{Y}_t^o - \bm{\mu}_{y,t}^m \big)\bm{\Sigma}_{y,t}^{oo \ -1} \bm{\Sigma}_{y,t}^{om} $ and 
$\tilde{\bm{\Sigma}}_{Y^m,t}= \bm{\Sigma}_{y,t}^{mm} - \bm{\Sigma}_{y,t}^{mo} \bm{\Sigma}_{y,t}^{oo \ -1} \bm{\Sigma}_{y,t}^{om}$, where $\bm{\mu}_{y,t} =   \bm{\mu} + \bm{\delta}_\epsilon \mathbf{D}_{\epsilon,t}^{-1} + \sigma^2 \bm{\delta}_x \mathbf{M}_t^{-1}\mathbf{W}^T \mathbf{N}_t^{-1}$ and  $\bm{\Sigma}_{y,t} = \sigma^2 \mathbf{D}_{\epsilon,t}^{-1} \mathbf{N}_t^{-1}$, given  
$\mathbf{M}_t = \sigma^2 \mathbf{D}_{x,t} + \mathbf{W}^T \mathbf{D}_{\epsilon,t} \mathbf{W}$ and $ \mathbf{N}_t = \Big( \mathbb{I}_d - \mathbf{W}\mathbf{M}_t^{-1} \mathbf{W}^T \mathbf{D}_{\epsilon,t}\Big)$, which are the central moments of the conditional distribution $\mathbf{Y}_t|\mathbf{U}_t, \mathbf{V}_t, \Psi$ derived in Lemma \ref{lemma:tildePi_Xt_Yt_ind}. Hence, we can express the integral \eqref{eq:proof_lemma_integral_w} using the moments of the random vector  $\mathbf{Y}_t^m|\mathbf{Y}_t^o, \mathbf{U}_t, \mathbf{V}_t, \Psi$, that is
\begin{align*}
& v(\mathbf{y}_t^o, \mathbf{u}_t, \mathbf{v}_t; \Psi,\Psi^*) = \log \pi_{\mathbf{U}_t| \Psi^*}( \mathbf{u}_t)  + \log \pi_{\mathbf{V}_t| \Psi^*}( \mathbf{v}_t) - \frac{k+d}{2}\log 2\pi - \frac{d}{2} \log \sigma^{*2} + \frac{1}{2} \sum_{i = 1}^d \log u_t^i + \frac{1}{2} \sum_{j = 1}^k \log v_t^j   \\\notag
& \hspace{0.3cm}  - \frac{1}{2} \bm{\delta}_x^* \mathbf{D}_{x,t}^{-1} \bm{\delta}_x^{*T}- \frac{1}{2 \sigma^{*2}} \tr \Big\{  \mathbb{E}_{\mathbf{Y}_t^m|\mathbf{Y}_t^o,\mathbf{U}_t,\mathbf{V}_t, \Psi}\big[ \mathbf{Y}_t^T \mathbf{Y}_t\big]    \mathbf{D}_{\epsilon,t} \Big\} +  \frac{1}{ \sigma^{*2}}  \mathbb{E}_{\mathbf{Y}_t^m|\mathbf{Y}_t^o,\mathbf{U}_t,\mathbf{V}_t, \Psi}\big[ \mathbf{Y}_t\big]  \Big(\mathbf{D}_{\epsilon,t}\bm{\mu}^{*T}+ \bm{\delta}_\epsilon^{*T} \Big) \\\notag
&\hspace{0.3cm}-  \frac{1}{ 2\sigma^{*2}} \bm{\mu}^{*} \mathbf{D}_{\epsilon,t} \bm{\mu}^{*T} - \frac{1}{\sigma^{*2}} \bm{\mu}^* \bm{\delta}_\epsilon^{*T} - \frac{1}{2 \sigma^{*2}} \bm{\delta}_\epsilon^* \mathbf{D}_{\epsilon,t}^{-1}  \bm{\delta}_\epsilon^{*T} + \frac{1}{\sigma^{*2}} \tr \Big\{ \mathbb{E}_{\mathbf{Y}_t^m|\mathbf{Y}_t^o,\mathbf{U}_t,\mathbf{V}_t, \Psi}\big[ \mathbf{Y}_t^T \mathbf{Y}_t\big]  \mathbf{D}_{\epsilon,t}  \mathbf{W} \mathbf{M}_t^{-1}\mathbf{W}^{*T} \mathbf{D}_{\epsilon,t} \Big\}  \\\notag
&\hspace{0.3cm} +  \mathbb{E}_{\mathbf{Y}_t^m|\mathbf{Y}_t^o,\mathbf{U}_t,\mathbf{V}_t, \Psi}\big[ \mathbf{Y}_t\big] \mathbf{D}_{\epsilon,t}  \mathbf{W} \mathbf{M}_t^{-1} \bigg(  - \frac{1}{\sigma^{*2}}\mathbf{W}^{*T} \mathbf{D}_{\epsilon,t}\bm{\mu}^{*T} - \frac{1}{\sigma^{*2}}\mathbf{W}^{*T} \bm{\delta}_\epsilon^{*T}+ \bm{\delta}_x ^{*T} \bigg) \\\notag
&\hspace{0.3cm} + \frac{1}{\sigma^{*2}}  \mathbb{E}_{\mathbf{Y}_t^m|\mathbf{Y}_t^o,\mathbf{U}_t,\mathbf{V}_t, \Psi}\big[ \mathbf{Y}_t\big]  \mathbf{D}_{\epsilon,t}  \mathbf{W}^* \mathbf{M}_t^{-1} \bigg( - \bm{\mu}\mathbf{D}_{\epsilon,t} \mathbf{W} - \bm{\delta}_\epsilon \mathbf{W} + \sigma^2 \bm{\delta}_x \bigg)^T  \\\notag
&\hspace{0.3cm}+  \bigg( - \bm{\mu}\mathbf{D}_{\epsilon,t} \mathbf{W} - \bm{\delta}_\epsilon \mathbf{W} + \sigma^2 \bm{\delta}_x \bigg) \mathbf{M}_t^{-1} \bigg(- \frac{1}{\sigma^{*2}}\mathbf{W}^{*T} \mathbf{D}_{\epsilon,t}\bm{\mu}^{*T} - \frac{1}{\sigma^{*2}}\mathbf{W}^{*T} \bm{\delta}_\epsilon^{*T}+ \bm{\delta}_x ^{*T} \bigg) \\\notag
& \hspace{0.3cm} - \frac{1}{2} \tr \bigg\{  \mathbf{M}_t^{-1} \mathbf{W} ^T\mathbf{D}_{\epsilon,t}  \mathbb{E}_{\mathbf{Y}_t^m|\mathbf{Y}_t^o,\mathbf{U}_t,\mathbf{V}_t, \Psi}\big[ \mathbf{Y}_t^T \mathbf{Y}_t\big]  \mathbf{D}_{\epsilon,t} \mathbf{W}\mathbf{M}_t^{-1}  \Big( \frac{1}{\sigma^{*2}}\mathbf{W}^{*T} \mathbf{D}_{\epsilon,t} \mathbf{W}^{*} + \mathbf{D}_{x,t} \Big) \bigg\} \\\notag
& \hspace{0.3cm} -\tr \bigg\{\mathbf{M}_t^{-1} \mathbf{W} ^T\mathbf{D}_{\epsilon,t}   \mathbb{E}_{\mathbf{Y}_t^m|\mathbf{Y}_t^o,\mathbf{U}_t,\mathbf{V}_t, \Psi}\big[ \mathbf{Y}_t\big]^T  \Big( - \bm{\mu}\mathbf{D}_{\epsilon,t} \mathbf{W} - \bm{\delta}_\epsilon \mathbf{W} + \sigma^2 \bm{\delta}_x \Big) \mathbf{M}_t^{-1}  \Big( \frac{1}{\sigma^{*2}}\mathbf{W}^{*T} \mathbf{D}_{\epsilon,t} \mathbf{W}^{*} + \mathbf{D}_{x,t} \Big) \bigg\} \\\notag
& \hspace{0.3cm} - \frac{1}{2} \tr \bigg\{  \mathbf{M}_t^{-1}\Big( - \bm{\mu}\mathbf{D}_{\epsilon,t} \mathbf{W} - \bm{\delta}_\epsilon \mathbf{W} + \sigma^2 \bm{\delta}_x \Big)^T\Big( - \bm{\mu}\mathbf{D}_{\epsilon,t} \mathbf{W} - \bm{\delta}_\epsilon \mathbf{W} + \sigma^2 \bm{\delta}_x \Big) \mathbf{M}_t^{-1}  \Big( \frac{1}{\sigma^{*2}}\mathbf{W}^{*T} \mathbf{D}_{\epsilon,t} \mathbf{W}^{*} + \mathbf{D}_{x,t} \Big) \bigg\} \\\notag
&\hspace{0.3cm}  - \frac{\sigma^2}{2} \tr \bigg\{ \mathbf{M}_t^{-1} \Big( \frac{1}{\sigma^{*2}}\mathbf{W}^{*T} \mathbf{D}_{\epsilon,t} \mathbf{W}^{*} + \mathbf{D}_{x,t} \Big) \bigg\}.
\end{align*}
\end{proof}

\begin{lemma} \label{remark:chi_square_smoothness}
The quantile distribution function of the Chi-square variable with $\nu$ degrees of freedom, $\chi^{-1}_{\nu}(s):[0,1] \rightarrow [0,\infty)$, belongs to the differentiable class  $\mathbb{C}^{ [\frac{ \nu}{2}]}$  as has derivatives of all orders not greater than $[\frac{\nu}{2}]$ given by the following differential relation
$$
\frac{\partial^m \chi^{-1}_{\nu} (s) }{\partial s^m}  = 2^{-m} \sum_{j = 1}^m \Big( \frac{m}{j} \Big) (-1)^{m+j} \partial^m \chi^{-1}_{\nu-2j},
$$
where the operator $[c]$ for any number $c$ returns the closest integer number to $c$, which is not greater than $s$. Hence, the quantile distribution function $\chi^{-1}_{\nu}(s)$ is continuous on the closed interval $[0,1]$ and has continuous derivatives up to $[\frac{ \nu}{2}]$ order. 
\end{lemma}
\begin{proof}
Please refer to \cite{AbramowitzStegun1972} and recall that the degrees of freedom belongs to the class of positive numbers and consequently $\nu-2m$ is greater than $0$.
\end{proof}

\begin{lemma}\label{lemma:m_smoothness_ind}
The function $m: \mathbb{R}^{d} \times [0,1]^2 \longrightarrow \mathbb{R}$ defined in Theorem~\ref{th:Estep_GStS_missing_ind} is of class $\mathbb{C}^{r_\epsilon,r_x}\big([0,1]^2\big)$ with respect to $s_{\epsilon,t}$ and $s_{x,t}$ for  $r_\epsilon = \min  \big[\frac{\bm{\nu}_\epsilon}{2}\big]$ and $r_x = \min  \big[\frac{\bm{\nu}_x}{2}\big]$.
\end{lemma}

\begin{proof}
We show that the function $m: \mathbb{R}^{d} \times [0,1]^2 \rightarrow \mathbb{R}$ 
\begin{align*}
& m(\mathbf{y}_t,s_{\epsilon,t},s_{x,t};\Psi) :=  \pi_{\mathbf{Y} |S_{\epsilon,t},S_{x,t},\Psi}(\mathbf{y}_t) \big(\sigma^2\big)^{\frac{k}{2}}  \Big| \mathbf{D}_{\epsilon,t} \Big|^{ \frac{1}{2}} \Big| \mathbf{D}_{x,t} \Big|^{\frac{1}{2}} \Big| \mathbf{M}_t\Big|^{- \frac{1}{2}} \Big| \mathbf{N}_t\Big|^{-\frac{1}{2}} e^{ - \frac{1}{2} \bm{\delta}_x \Big( \mathbf{D}_{x,t}^{-1} - \sigma^2  \mathbf{M}_t^{-1}\mathbf{W}^T \mathbf{D}_{\epsilon,t} \mathbf{N}_t^{-1}\mathbf{D}_{\epsilon,t} \mathbf{W}\mathbf{M}_t^{-1}  - \sigma^2  \mathbf{M}_t^{-1} \Big)\bm{\delta}_x^T},
\end{align*}
has continuous partial derivatives
\begin{equation}\label{eq:m_partial}
\frac{\partial^{i+j} m(\mathbf{y}_t,s_{\epsilon,t},s_{x,t};\Psi) }{\partial s_{\epsilon,t}^i \ \partial s_{x,t}^j},
\end{equation}
for  $0 \leq i \leq r_\epsilon $ and $0 \leq j \leq r_x$.  Let $m(\mathbf{y}_t,s_\epsilon, s_x) := \tilde{m}(s_\epsilon,s_x)$.  The function $\tilde{m}$ can be formulated as a composite of polynomial $p$ and the function $g$ such that $\tilde{m}(s_\epsilon,s_x) = (p \circ g) (s_\epsilon,s_x)$ for $i = 1, \ldots, d+k$ where
\begin{itemize}
\item  $g : [0,1]^2 \rightarrow \mathbb{R}^d_+ \times \mathbb{R}_+^k$, such that $g(s_\epsilon,s_x) := \big( T_\epsilon(s_\epsilon),T_x(s_x)\big) = \big( g_1(s_\epsilon),\ldots,g_d(s_\epsilon),g_{d+1}(s_x) \ldots, g_{d+k}(s_x)\big)$ for $g_i:[0,1] \rightarrow \mathbb{R}_+$ being the following
\begin{align*}
& \big( g_1(s), \ldots, g_{d}(s)\big) = \bigg( \frac{\chi_{\nu_{\epsilon}^1}^{-1} (s)}{\nu_{\epsilon}^1}, \ldots, \frac{\chi_{\nu_{\epsilon}^d}^{-1} (s)}{\nu_{\epsilon}^d}  \bigg)
\textrm{ and } \big( g_{d+1}(s), \ldots, g_{d+k}(s)\big) = \bigg( \frac{\chi_{\nu_x^1}^{-1} (s)}{\nu_x^1}, \ldots, \frac{\chi_{\nu_x^k}^{-1} (s)}{\nu_x^k} \bigg),
\end{align*}
given the function $T_\epsilon:[0,1] \rightarrow \mathbb{R}^d_+$ and $T_x:[0,1] \rightarrow \mathbb{R}^k_+$ defined in Theorem~\ref{th:Estep_GStS_missing_ind}.
\vspace{0.3cm}
\item $p : \mathbb{R}^d_+ \times \mathbb{R}_+^k \rightarrow \mathbb{R}$, such that $p(\mathbf{t}_\epsilon,\mathbf{t}_x ) =  \prod_{k = 1}^5 p_k (\mathbf{t}_\epsilon,\mathbf{t}_x )$ for $(\mathbf{t}_\epsilon,\mathbf{t}_x )\in \mathbb{R}^d_+ \times \mathbb{R}_+^k$ and $p_k:\mathbb{R}^d_+ \times \mathbb{R}_+^k \rightarrow \mathbb{R}$ being the following
\begin{align*}
& p_1(\mathbf{t}_\epsilon,\mathbf{t}_x ) = \bigg( \prod_{i = 1}^d t_{\epsilon,i} \prod_{j = 1}^k t_{x,j} \bigg)^{\frac{1}{2}}, \ p_2(\mathbf{t}_\epsilon,\mathbf{t}_x ) = \det \Big( M(\mathbf{t}_\epsilon,\mathbf{t}_x )  \Big)^{-\frac{1}{2}} , \ p_3(\mathbf{t}_\epsilon,\mathbf{t}_x ) = \det \Big( N(\mathbf{t}_\epsilon,\mathbf{t}_x ) \Big)^{-\frac{1}{2}} ,\\
&p_4(\mathbf{t}_\epsilon,\mathbf{t}_x ) = \exp \bigg\{ - \frac{1}{2} \bm{\delta}_x \Big( \diag( \mathbf{t}_x)^{-1} - \sigma^2  \mathbf{M}(\mathbf{t}_\epsilon,\mathbf{t}_x )^{-1}\mathbf{W}^T \diag( \mathbf{t}_\epsilon) \mathbf{N}(\mathbf{t}_\epsilon,\mathbf{t}_x )^{-1}\diag( \mathbf{t}_\epsilon) \mathbf{W}\mathbf{M}(\mathbf{t}_\epsilon,\mathbf{t}_x )^{-1}  - \sigma^2  \mathbf{M}(\mathbf{t}_\epsilon,\mathbf{t}_x )^{-1} \Big)\bm{\delta}_x^T \bigg\} ,\\
& p_5(\mathbf{t}_\epsilon,\mathbf{t}_x ) = \Big( 2\pi \sigma^2 \big)^{-\frac{d}{2}} \det \Big( \mathbf{N}(\mathbf{t}_\epsilon,\mathbf{t}_x ) \Big)^{\frac{1}{2}} \exp \bigg\{ - \frac{1}{2\sigma^2} \big(\mathbf{y} - u(\mathbf{t}_\epsilon,\mathbf{t}_x ) \big)\mathbf{N}(\mathbf{t}_\epsilon,\mathbf{t}_x )^{-1}\big(\mathbf{y} - u(\mathbf{t}_\epsilon,\mathbf{t}_x ) \big)^T, \bigg\} 
\end{align*}
where 
\begin{align*} 
&M(\mathbf{t}_\epsilon,\mathbf{t}_x )_{d \times d} = \sigma^2 \diag(\mathbf{t}_x)  + \mathbf{W}^T \diag( \mathbf{t}_\epsilon)  \mathbf{W}, \\
& N(\mathbf{t}_\epsilon,\mathbf{t}_x )_{k \times k}  = \diag( \mathbf{t}_\epsilon)  - \diag( \mathbf{t}_\epsilon) \mathbf{W} M(\mathbf{t}_\epsilon,\mathbf{t}_x ) \mathbf{W}^T \diag( \mathbf{t}_\epsilon), \\
&u(\mathbf{t}_\epsilon,\mathbf{t}_x ) = \bm{\mu} + \bm{\delta}_\epsilon \diag( \mathbf{t}_\epsilon)^{-1} + \sigma^2 \bm{\delta}_x \mathbf{M}(\mathbf{t}_\epsilon,\mathbf{t}_x ) ^{-1}\mathbf{W}^T  \diag( \mathbf{t}_\epsilon) \mathbf{N}(\mathbf{t}_\epsilon,\mathbf{t}_x ) ^{-1}.
\end{align*}
\end{itemize}
Given the above representation of $\tilde{m}$, the differentiability class of the function can be determined by specifying the minimum differentiability class of $g$ and $p$. It can be seen by applying the chain rule formula for multivariate partial derivatives from \cite{Bruno1857} to (4) which results in
\begin{equation*}
\frac{\partial^{i+j} }{\partial s_{\epsilon,t}^i \ \partial s_{x,t}^j}  \tilde{m}(s_{\epsilon,t},s_{x,t};\Psi) = \frac{\partial}{\partial \mathbf{t}_\epsilon \partial \mathbf{t}_x^T} p(\mathbf{t}_\epsilon,\mathbf{t}_x )  \   \frac{\partial^{i+j} }{\partial s_{\epsilon,t}^i \ \partial s_{x,t}^j}  g(s_{\epsilon,t},s_{x,t}) .
\end{equation*}
We start by specifying the class of the functions $T_\epsilon$ and $T_x$ defined as 
\begin{align*}\notag
&T_\epsilon(s) := \bigg( \frac{\chi_{\nu_{\epsilon}^1}^{-1} (s)}{\nu_{\epsilon}^1}, \ldots, \frac{\chi_{\nu_{\epsilon}^d}^{-1} (s)}{\nu_{\epsilon}^d}  \bigg)_{1 \times d} \mbox{ and } T_x(s) := \bigg( \frac{\chi_{\nu_x^1}^{-1} (s)}{\nu_x^1}, \ldots, \frac{\chi_{\nu_x^k}^{-1} (s)}{\nu_x^k} \bigg)_{1 \times k}.
\end{align*}
Using Lemma \ref{remark:chi_square_smoothness}, it is straightforward to show that $T_\epsilon(s)$ and $T_x(s)$ have continuous derivatives up to the order $\min_i [\frac{v_\epsilon^i}{2}]$ for $T_\epsilon(s)$ and  $\min_j [\frac{v_x^j}{2}]$ for $T_x(s)$. It stems from the fact that the differentiability class of the multidimensional functions equals to the minimum differentiability class of its marginals.  Consequently, the function $g$ belongs to the class $\mathbb{C}^{r_\epsilon,r_x}([0,1]^2)$ for $r_\epsilon = \min  \big[\frac{\bm{\nu}_\epsilon}{2}\big]$ and $r_x = \min  \big[\frac{\bm{\nu}_x}{2}\big]$ with respect to $s_\epsilon$ and  $s_x$, respectively. 

Secondly, we show that the functions $p$ is infinitely differentiable under the model assumptions from Subsection~\ref{ssec:PPCA_GStS_ind}. Since the function $p_1$ is polynomial, it is infinitely differentiable. The function $p_k$ for $k = 2,\ldots, 5$ are rational polynomials and, consequently, are infinitely differentiable except for regions where $\det \Big( M(\mathbf{t}_\epsilon,\mathbf{t}_x )  \Big) = 0$ or $\det \Big( N(\mathbf{t}_\epsilon,\mathbf{t}_x )  \Big) = 0$.  However, given the distribution assumptions of the GSt PPCA model, the determinants of $M(\mathbf{t}_\epsilon,\mathbf{t}_x )  \Big)$ and $N(\mathbf{t}_\epsilon,\mathbf{t}_x )$ are always greater than zero. It is due to the fact the matrices represent covariances of random vectors as shown in proven in Lemma \ref{lemma:tildePi_Xt_Yt_ind} and hence, are positive-definite from definition.  It implicates that the function $p = \prod_{k = 1}^5 p_k$ is smooth. The above reasoning shows that function $\tilde{m}$ belongs to the class $\mathbb{C}^{r_\epsilon,r_x}\big([0,1]^2\big)$ for $r_\epsilon = \min  \big[\frac{\bm{\nu}_\epsilon}{2}\big]$ and $r_x = \min  \big[\frac{\bm{\nu}_x}{2}\big]$.
\end{proof}

\subsection{Proof of Theorem~\ref{th:Estep_GStS_missing_ind}}\label{proof:th_Estep_GStS_missing_ind}
\begin{proof}
Let the observation vector $\mathbf{Y}_t$ is  partitioned into two subvectors, $\mathbf{Y}_t = [\mathbf{Y}_t^o, \mathbf{Y}_t^m]$ with observed and unobserved entries of $\mathbf{Y}_t$, respectively.  
The joint probability function of the variables $\mathbf{Y}_t, \ \mathbf{X}_t, \mathbf{U}_t$ and $\mathbf{V}_t$ for Generalized Skew-t PPCA model is specified in (7). Given the independence of the random vectors $\mathbf{Y}_t, \ \mathbf{X}_t, \mathbf{U}_t$ and $\mathbf{V}_t$ over their realisations, the E-step of the corresponding EM algorithm can be written as 
\begin{align*}\notag
& Q(\Psi, \Psi^* )  = \mathbb{E}_{\mathbf{Y}^m_{1:N},\mathbf{X}_{1:N},\mathbf{U}_{1:N},\mathbf{V}_{1:N}| \mathbf{Y}^o_{1:N}, \Psi} \Big[\log \pi_{\mathbf{Y}_{1:N}, \mathbf{X}_{1:N}, \mathbf{U}_{1:N}, \mathbf{V}_{1:N}|\Psi^*}\left(\mathbf{Y}_{1:N}, \mathbf{X}_{1:N}, \mathbf{U}_{1:N},\mathbf{V}_{1:N}  \right) \Big]\\\notag
& \hspace{0.2cm}= c \int_{\mathbb{R}^{N \times d}_+} \int_{\mathbb{R}^{N \times k}_+} \int_{\mathbb{R}^{N \times d_m}} \sum_{t=1}^N \Bigg\{ \int_{\mathbb{R}^{N \times k}} \bigg[ \log \bigg(\pi_{\mathbf{Y}_t, \mathbf{X}_t, \mathbf{U}_t, \mathbf{V}_t| \Psi^*}(\mathbf{y}_t, \mathbf{x}_t, \mathbf{u}_t,\mathbf{v}_t) \bigg) \\
& \hspace{1cmcm} \times \prod_{s=1}^N \pi_{\mathbf{Y}_s, \mathbf{X}_s, \mathbf{U}_s, \mathbf{V}_s| \Psi}(\mathbf{y}_s, \mathbf{x}_s, \mathbf{u}_s,\mathbf{v}_s)\bigg] \ d \mathbf{x}_{1:N} \Bigg\} d \mathbf{y}^m_{1:N} \ d\mathbf{u}_{1:N} \ d\mathbf{v}_{1:N},
\end{align*}
for $c = \frac{1}{\pi_{\mathbf{Y}_{1:N}^o| \Psi}(\mathbf{y}_{1:N}^o)} \geq 0$. Using Lemma \ref{lemma:tildePi_Xt_Yt_ind}, we can decompose the probability function $ \pi_{\mathbf{Y}_t, \mathbf{X}_t, \mathbf{U}_t, \mathbf{V}_t| \Psi}(\mathbf{y}_t, \mathbf{x}_t, \mathbf{u}_t,\mathbf{v}_t)$ into the product of three functions
\begin{equation*}
\pi_{\mathbf{Y}_t, \mathbf{X}_t, \mathbf{U}_t, \mathbf{V}_t| \Psi}(\mathbf{y}_t, \mathbf{x}_t, \mathbf{u}_t,\mathbf{v}_t)   = \pi_{\mathbf{X}_t|\mathbf{Y}_t,\mathbf{U}_t,\mathbf{V}_t,\Psi}(\mathbf{x}_t) \  \pi_{\mathbf{Y}_t|\mathbf{U}_t,\mathbf{V}_t,\Psi}(\mathbf{y}_t) \ C( \mathbf{u}_t,\mathbf{v}_t;\Psi \big) . 
\end{equation*}
Please refer to the result derived in Lemma \ref{lemma:tildePi_Xt_Yt_ind} to specify probability functions and the function $C( \mathbf{u}_t,\mathbf{v}_t;\Psi \big):\mathbb{R}^d_+ \times  \mathbb{R}^k_+\longrightarrow \mathbb{R} $. 

\noindent Therefore, we can further simplify the formulation of the $Q$ function
\begin{align*}\notag
& Q(\Psi, \Psi^* )  =  c\int_{\mathbb{R}^{N \times d}_+} \int_{\mathbb{R}^{N \times k}_+} \int_{\mathbb{R}^{N \times d_m}} \Bigg( \sum_{t=1}^N w(\mathbf{y}_t, \mathbf{u}_t, \mathbf{v}_t; \Psi,\Psi^*) \Bigg) \Bigg( \prod_{s=1}^N  \pi_{\mathbf{Y}_s|\mathbf{U}_s,\mathbf{V}_s,\Psi}(\mathbf{y}_s) \Bigg) \ d \mathbf{y}^m_{1:N} \Bigg( \prod_{s=1}^N C( \mathbf{u}_t,\mathbf{v}_t;\Psi \big) \Bigg) \ d\mathbf{u}_{1:N} \ d\mathbf{v}_{1:N} ,
\end{align*}
for the function $w:\mathbb{R}^d \times \mathbb{R}^d_+ \times  \mathbb{R}^k_+\longrightarrow \mathbb{R} $ defined in Lemma \ref{lemma:w_ind}. We can follow the similar steps to simplify the integration over the vector $\mathbf{y}_{1:N}^m$ using the fact that $\mathbf{y}_t$ are mutually independent . The next step is to specify the solutions to the integration problems 
\begin{align*}
&\int_{\mathbb{R}^{d_m}} \pi_{\mathbf{Y}_t|\mathbf{U}_t,\mathbf{V}_t,\Psi}(\mathbf{y}_t)  \ d \mathbf{y}^m_t 
 \textrm{ and } 
\int_{\mathbb{R}^{d_m}}  w(\mathbf{y}_t, \mathbf{u}_t, \mathbf{v}_t; \Psi,\Psi^*) \pi_{\mathbf{Y}_t|\mathbf{U}_t,\mathbf{V}_t,\Psi}(\mathbf{y}_t) \ d \mathbf{y}^m_t, 
\end{align*}
which rely on $\pi_{\mathbf{Y}_t^m|\mathbf{Y}_t^o,\mathbf{U}_t,\mathbf{V}_t,\Psi}(\mathbf{y}_t^m)$. The probability function $\pi_{\mathbf{Y}_t|\mathbf{U}_t,\mathbf{V}_t,\Psi}(\mathbf{y}_t)$ is a  Gaussian and conditional distributions of Gaussian random vectors are broadly known. We derive the probability functions in Lemma \ref{lemma:tildePi_Ytmissing_Conditional_ind} such that
\begin{equation*}
\pi_{\mathbf{Y}_t|\mathbf{U}_t,\mathbf{V}_t,\Psi}(\mathbf{y}_t) = \pi_{\mathbf{Y}_t^m|\mathbf{Y}_t^o,\mathbf{U}_t,\mathbf{V}_t,\Psi}(\mathbf{y}_t^m) \ \pi_{\mathbf{Y}_t^o |\mathbf{U}_t,\mathbf{V}_t,\Psi}(\mathbf{y}_t^o),
\end{equation*}
and use them to obtain the formulation of the $Q$ function as following
\begin{align*} 
& Q(\Psi, \Psi^* ) = c\int_{\mathbb{R}^{N \times d}_+} \int_{\mathbb{R}^{N \times k}_+} \underbrace{ \sum_{t=1}^N \Bigg\{ \int_{\mathbb{R}^{d_m}}  w(\mathbf{y}_t, \mathbf{u}_t, \mathbf{v}_t; \Psi,\Psi^*)  \pi_{\mathbf{Y}_t^m|\mathbf{Y}_t^o,\mathbf{U}_t,\mathbf{V}_t,\Psi}(\mathbf{y}_t^m) \ d \mathbf{y}^m_t    \Bigg\}}_{v(\mathbf{y}_{1:N}^o,\mathbf{u}_{1:N},\mathbf{v}_{1:N};\Psi,\Psi^*)} \\
& \hspace{1cm} \times \underbrace{ \bigg( \prod_{s=1}^N  \pi_{\mathbf{Y}_t^o |\mathbf{U}_t,\mathbf{V}_t,\Psi}(\mathbf{y}_t^o)  C( \mathbf{u}_t,\mathbf{v}_t;\Psi \big) \bigg)}_{h(\mathbf{y}_{1:N}^o,\mathbf{u}_{1:N},\mathbf{v}_{1:N};\Psi)} \ d\mathbf{u}_{1:N} \ d\mathbf{v}_{1:N}.
\end{align*}
Given the solution to the integration problem $v(\mathbf{y}_{1:N}^o,\mathbf{u}_{1:N},\mathbf{v}_{1:N};\Psi,\Psi^*)$ in Lemma \ref{lemma:integral_w_ind} and denoting f
\begin{align*}
h(\mathbf{y}_{t}^o,\mathbf{u}_{t},\mathbf{v}_{t};\Psi) =&  \pi_{\mathbf{Y}_t^o |\mathbf{U}_t,\mathbf{V}_t,\Psi}(\mathbf{y}_t^o)  \pi_{\mathbf{U}_t | \Psi} ( \mathbf{u}_t)  \pi_{\mathbf{V}_t | \Psi} ( \mathbf{v}_t)  \big(\sigma^2\big)^{\frac{k}{2}}  \Big| \mathbf{D}_{x,t}^{-1} \Big|^{- \frac{1}{2}} \Big| \mathbf{M}_t\Big|^{- \frac{1}{2}} \Big| \mathbf{N}_t^{-1}\Big|^{\frac{1}{2}}  \\
& \times \exp \Bigg\{ - \frac{1}{2} \bm{\delta}_x \Big( \mathbf{D}_{x,t}^{-1} - \sigma^2  \mathbf{M}_t^{-1}\mathbf{W}^T \mathbf{N}_t^{-1}\mathbf{D}_{\epsilon,t} \mathbf{W}\mathbf{M}_t^{-1}  - \sigma^2  \mathbf{M}_t^{-1} \Big)\bm{\delta}_x^T \Bigg\},
\end{align*}
the E-step of the Generalized Skew-t PPCA is the following
\begin{align}\label{eq:proof_Estep_Q_GStS_missing_ind}
& Q(\Psi, \Psi^* )= \sum_{t=1}^N \Bigg\{  \int_{\mathbb{R}^{d}_+} \int_{\mathbb{R}^{k}_+} v(\mathbf{y}_{t}^o,\mathbf{u}_{t},\mathbf{v}_{t};\Psi,\Psi^*)  h(\mathbf{y}_{t}^o,\mathbf{u}_{t},\mathbf{v}_{t};\Psi) \ d\mathbf{u}_t \ d \mathbf{v}_t \ \prod_{s=1, s \neq t}^N \int_{\mathbb{R}^{d}_+} \int_{\mathbb{R}^{k}_+} h(\mathbf{y}_{s}^o,\mathbf{u}_{s},\mathbf{v}_{s};\Psi) \ d\mathbf{u}_{s} \ d\mathbf{v}_{s}\Bigg\}.
\end{align}
Let us recall the definition of the mutually independent multivariate Gamma random vectors $\mathbf{U}_t$ and $\mathbf{V}_t$ from (5) in Subsection~\ref{ssec:PPCA_GStS_ind} as transformations of the mutually independent univariate uniform random variables $S_{\epsilon,t}, \ S_{x,t} \in \mathcal{U}\left(0,1\right)$ by the quantile functions $T_\epsilon:[0,1] \rightarrow \mathbb{R}_+^d$ and $T_x:[0,1] \rightarrow \mathbb{R}_+^k$ such that
\begin{align*}
&T_\epsilon(s) := \bigg( \frac{\chi_{\nu_{\epsilon}^1}^{-1} (s)}{\nu_{\epsilon}^1}, \ldots, \frac{\chi_{\nu_{\epsilon}^d}^{-1} (s)}{\nu_{\epsilon}^d}  \bigg)_{1 \times d} \ \mbox{ and } \ T_x(s) := \bigg( \frac{\chi_{\nu_x^1}^{-1} (s)}{\nu_x^1}, \ldots, \frac{\chi_{\nu_x^k}^{-1} (s)}{\nu_x^k} \bigg)_{1 \times k} .
\end{align*}
where $\chi_{\nu}^{-1}$ is the quantile distribution function of a univariate Chi-square random variable with $\nu$ degrees of freedom. By the change of variables $S_{\epsilon,t} = T^{-1}_\epsilon (\mathbf{U}_t)$ and $S_{x,t} = T^{-1}_x (\mathbf{V}_t)$, we can reduce the $d \times k$ dimensional integration problems from  \eqref{eq:proof_Estep_Q_GStS_missing_ind} to the two-dimensional integrations on the unit hypercube $[0,1]^2$. 

Firstly, we remark that the corresponding density functions $\pi_{\mathbf{U}_t| \Psi} (\mathbf{u}_t)$ and $\pi_{\mathbf{V}_t| \Psi} (\mathbf{v}_t)$ can be expressed by the density function of the random variables $S_{\epsilon,t}$ and $S_{x,t}$ defined in Subsection~\ref{ssec:PPCA_GStS_ind}, that is
\begin{align*} \notag
& \pi_{\mathbf{U}_t| \Psi} (\mathbf{u}_t) = \pi_{S_{\epsilon,t}| \Psi} (s_{\epsilon,t}) \Big|\frac{\partial T_\epsilon( s_{\epsilon,t})}{\partial s_{\epsilon,t}}\Big|^{-1} = \mathbf{1}_{[0,1]} (s_{\epsilon,t})  \Big|\frac{\partial T_\epsilon( s_{\epsilon,t})}{\partial s_{\epsilon,t}}\Big|^{-1}, \\
& \pi_{\mathbf{V}_t| \Psi} (\mathbf{v}_t) = \pi_{S_{x,t}| \Psi} (s_{x,t}) \Big|\frac{\partial T_x( s_{x,t})}{\partial s_{x,t}}\Big|^{-1} = \mathbf{1}_{[0,1]} (s_{x,t})  \Big|\frac{\partial T_x( s_{x,t})}{\partial s_{x,t}}\Big|^{-1}.
\end{align*}
The distribution of $\mathbf{Y}_t^o$ conditioned on $\mathbf{U}_{t}$ and $\mathbf{V}_{t}$ can also by expressed by $S_{\epsilon,t}$ and $S_{x,t}$ as 
$$ 
\pi_{\mathbf{Y}_t^o |T_\epsilon (S_{\epsilon,t}),T_x (S_{x,t}),\Psi}(\mathbf{y}_t^o) = \pi_{\mathbf{Y}_t^o |S_{\epsilon,t},S_{x,t},\Psi}(\mathbf{y}_t^o).
$$
Recall that the definitions of matrices $\mathbf{D}_{\epsilon,t}$ and $\mathbf{D}_{x,t}$ under the transformation $\mathbf{U}_t = T_\epsilon (S_{\epsilon,t})$ and $\mathbf{V}_t = T_x (S_{x,t})$ become
\begin{equation*}
\mathbf{D}_{\epsilon,t} =  \begin{pmatrix}
 \frac{\chi_{\nu_{\epsilon}^1}^{-1} (S_{\epsilon,t})}{\nu_{\epsilon}^1} &0&0 \\
0& \ddots &0 \\
0& 0&  \frac{\chi_{\nu_{\epsilon}^d}^{-1} (S_{\epsilon})}{\nu_{\epsilon}^d}
\end{pmatrix}_{d \times d} , \
\mathbf{D}_{x,t} =  \begin{pmatrix}
 \frac{\chi_{\nu_{x}^1}^{-1} (S_{x,t})}{\nu_{x}^1} &0&0 \\
0& \ddots &0 \\
0&0 &  \frac{\chi_{\nu_{x}^k}^{-1} (S_{x,t})}{\nu_{x}^k}
\end{pmatrix}_{k \times k}.
\end{equation*}
Hence, the E-step of the EM algorithm for GSt PPCA from  \eqref{eq:proof_Estep_Q_GStS_missing_ind} can be expressed by the sum of two-dimensional integration problems
\begin{equation*}
Q(\Psi, \Psi^* )= \frac{1}{\pi_{\mathbf{Y}_{1:N}^o| \Psi}(\mathbf{y}_{1:N}^o)} \sum_{t=1}^N \bigg\{ I_1(\mathbf{y}_{t}^o;\Psi,\Psi^*) \prod_{s=1, s \neq t}^N  I_2(\mathbf{y}_{s}^o;\Psi) \bigg\} ,
\end{equation*}
for the functions $I_1, \ I_2:\mathbb{R}^{d_o} \rightarrow \mathbb{R}$ 
\begin{align*}\notag
& I_1(\mathbf{y}_{t}^o;\Psi,\Psi^*) :=  \int_0^1 \int_0^1 v(\mathbf{y}_{t}^o,T_\epsilon(s_{\epsilon,t}),T_x(s_{x,t});\Psi,\Psi^*) \  m(\mathbf{y}_{t}^o,s_{\epsilon,t},s_{x,t};\Psi) \ d s_{\epsilon,t} \ d s_{x,t}, \\
& I_2(\mathbf{y}_{t}^o;\Psi) :=  \int_0^1 \int_0^1    m(\mathbf{y}_{t}^o,s_{\epsilon,t},s_{x,t};\Psi) \ d s_{\epsilon,t} \ d s_{x,t} ,
\end{align*}
and the function $m: \mathbb{R}^{d_o} \times [0,1]^2 \longrightarrow \mathbb{R}$
\begin{equation*}
m(\mathbf{y}_{t}^o,s_{\epsilon,t},s_{x,t};\Psi) = \pi_{\mathbf{Y}_t^o |S_{\epsilon,t},S_{x,t},\Psi}(\mathbf{y}_t^o) \big(\sigma^2\big)^{\frac{k}{2}} \Big| \mathbf{D}_{x,t}^{-1} \Big|^{- \frac{1}{2}} \Big| \mathbf{M}_t\Big|^{- \frac{1}{2}} \Big| \mathbf{N}_t\Big|^{-\frac{1}{2}} e^{ - \frac{1}{2} \bm{\delta}_x \big( \mathbf{D}_{x,t}^{-1} - \sigma^2  \mathbf{M}_t^{-1}\mathbf{W}^T \mathbf{N}_t^{-1}\mathbf{D}_{\epsilon,t} \mathbf{W}\mathbf{M}_t^{-1}  - \sigma^2  \mathbf{M}_t^{-1} \big)\bm{\delta}_x^T}.
\end{equation*}
\end{proof}

\subsection{Proof of Theorem~\ref{th:Mstep_explicitFormulas_missing_ind}} \label{proof:cor_Mstep_explicitFormulas_missing_ind}

\begin{proof}
Let us define a function $H:\mathbb{R}^{(N-1) \times d_o} \longrightarrow \mathbb{R}$ which is a product of $N-1$ integration problems on $[0,1]^2$, represented by the function $I_2$ from Theorem~\ref{th:Estep_GStS_missing_ind}, given by
\begin{equation}\label{eq:func_H_Mstep_ind}
H \big(\mathbf{y}_{1:N/t}^o;\Psi \big) := \prod_{s=1, s \neq t}^N  I_2(\mathbf{y}_s^o;\Psi).
\end{equation} 
The set $\mathbf{y}_{1:N/t}^o$ is a $N-1$ sequence of $d_o$-dimensional vectors with observed entries of $\mathbf{y}_t$ without the vector $\mathbf{y}^o_t$, that is
$$
\mathbf{y}_{1:N/t}^o = \Big\{ \mathbf{y}_1^o,\ldots,\mathbf{y}_{t-1}^o,\mathbf{y}_{t+1}^o, \ldots, \mathbf{y}_N^o \Big\}.
$$ 
The maximizers of the function $Q$ with respect to 
$\bm{\mu}^* \ ,\sigma^{*2}, \ \bm{\delta}^*_\epsilon$ and $\bm{\delta}_x^*$ have closed form solutions defined by sequences of various two-dimensional integration problem as follows
\begin{align*}
& \frac{\partial}{\partial \bm{\mu}^*} Q(\Psi,\Psi^*) = \mathbf{0} \Longleftrightarrow \sum_{t=1}^N \bigg\{ H \Big(\mathbf{y}^o_{1:N/t};\Psi \Big) \ \int_0^1 \int_0^1   \frac{\partial}{\partial \bm{\mu}^*}  \tilde{v}\big(\mathbf{y}_t^o, s_{\epsilon,t}, s_{x,t}; \Psi,\Psi^*\big)  \  m(\mathbf{y}_{t}^o,s_{\epsilon,t},s_{x,t};\Psi) \ d s_{\epsilon,t} \ d s_{x,t}  = \mathbf{0} \\\notag
&\Longleftrightarrow \bm{\mu}^* = \bigg[  \mathbf{A}_6 (\mathbf{y}_{1:N}^o;\Psi) -  \mathbf{A}_{10} (\mathbf{y}_{1:N}^o;\Psi,\Psi^*)  - \bm{\delta}_\epsilon^*  A_0 (\mathbf{y}_{1:N}^o;\Psi) + \bm{\mu}  \mathbf{A}_{13} (\mathbf{y}_{1:N}^o;\Psi,\Psi^*) + \big(\bm{\delta}_\epsilon \mathbf{W}  -  \sigma^2 \bm{\delta}_x\big) \mathbf{A}_{12} (\mathbf{y}_{1:N}^o;\Psi,\Psi^*)^T \bigg]\mathbf{A}_1(\mathbf{y}_{1:N}^o;\Psi)^{-1} ,\\\notag 
& \\\notag
&\frac{\partial}{\partial \bm{\delta}_\epsilon^*} Q(\Psi,\Psi^*) = \mathbf{0} \Longleftrightarrow \sum_{t=1}^N \bigg\{ H \Big(\mathbf{y}^o_{1:N/t};\Psi \Big) \ \int_0^1 \int_0^1  \frac{\partial}{\partial \bm{\delta}_\epsilon^*}  \tilde{v}\big(\mathbf{y}_t^o, s_{\epsilon,t}, s_{x,t}; \Psi,\Psi^*\big)  \  m(\mathbf{y}_{t}^o,s_{\epsilon,t},s_{x,t};\Psi) \ d s_{\epsilon,t} \ d s_{x,t}  = \mathbf{0} \\\notag
&\Longleftrightarrow \bm{\delta}_\epsilon^* =  \bigg[ \mathbf{A}_5(\mathbf{y}_{1:N}^o;\Psi) - \bm{\mu}^*A_0(\mathbf{y}_{1:N}^o;\Psi) - \mathbf{A}_8(\mathbf{y}_{1:N}^o;\Psi)\mathbf{W}^{*T}
 +  \bm{\mu} \mathbf{A}_{11}(\mathbf{y}_{1:N}^o;\Psi)  \mathbf{W}^{*T}+ \big( \bm{\delta}_\epsilon \mathbf{W} -  \sigma^2 \bm{\delta}_x \big) \mathbf{A}_4(\mathbf{y}_{1:N}^o;\Psi) \mathbf{W}^{*T} \bigg] \mathbf{A}_2(\mathbf{y}^{o}_{1:N};\Psi)^{-1}, \\\notag
& \\\notag
&\frac{\partial}{\partial \bm{\delta}_x^*} Q(\Psi,\Psi^*) = \mathbf{0} \Longleftrightarrow \sum_{t=1}^N \bigg\{ H \Big(\mathbf{y}^o_{1:N/t};\Psi \Big) \ \int_0^1 \int_0^1   \frac{\partial}{\partial \bm{\delta}_x^*}  \tilde{v}\big(\mathbf{y}_t^o,  s_{\epsilon,t}, s_{x,t}; \Psi,\Psi^*\big)  \  m\big(\mathbf{y}_{t}^o,s_{\epsilon,t},s_{x,t};\Psi\big) \ d s_{\epsilon,t} \ d s_{x,t}  = \mathbf{0} \\\notag
&\Longleftrightarrow \bm{\delta}_x^* = \bigg[ \mathbf{A}_8(\mathbf{y}_{1:N}^o;\Psi)- \bm{\mu} \mathbf{A}_{11}(\mathbf{y}_{1:N}^o;\Psi) - \big( \bm{\delta}_\epsilon \mathbf{W} -  \sigma^2 \bm{\delta}_x \big) \mathbf{A}_4(\mathbf{y}_{1:N}^o;\Psi) \bigg] \mathbf{A}_3 (\mathbf{y}_{1:N},\Psi)^{-1} ,\\\notag
& \\\notag
&\frac{\partial}{\partial \sigma^{*2}} Q(\Psi,\Psi^*) = \mathbf{0} \Longleftrightarrow \sum_{t=1}^N \bigg\{ H \Big(\mathbf{y}^o_{1:N/t};\Psi \Big) \ \int_0^1 \int_0^1  \frac{\partial}{\partial \sigma^{*2}}  \tilde{v}\big(\mathbf{y}_t^o, s_{\epsilon,t}, s_{x,t}; \Psi,\Psi^*\big)  \  m(\mathbf{y}_{t}^o,s_{\epsilon,t},s_{x,t};\Psi) \ d s_{\epsilon,t} \ d s_{x,t}  = \mathbf{0} \\\notag
&\Longleftrightarrow \sigma^{*2} = \frac{1}{d A_0(\mathbf{y}_{1:N}^o;\Psi)} \bigg[  A_{20}(\mathbf{y}_{1:N}^o;\Psi)  + \bm{\mu}^* \mathbf{A}_{1}(\mathbf{y}_{1:N}^o;\Psi) \bm{\mu}^{*T} + \bm{\delta}_\epsilon^{*} \mathbf{A}_{2}(\mathbf{y}_{1:N}^o;\Psi)\bm{\delta}_\epsilon^{*T} + 2   \mathbf{A}_{9}(\mathbf{y}_{1:N}^o;\Psi,\Psi^*) \bm{\mu}^T   \\\notag
& \hspace{1cm} + 2 \Big( \mathbf{A}_{10}(\mathbf{y}_{1:N}^o;\Psi,\Psi^*) -  \mathbf{A}_{6}(\mathbf{y}_{1:N}^o;\Psi)- \bm{\mu} \mathbf{A}_{13}(\mathbf{y}_{1:N}^o;\Psi,\Psi^*)\Big) \bm{\mu}^{*T} - 2  A_{21}(\mathbf{y}_{1:N}^o;\Psi,\Psi^*) +   \sigma^2 \tr \Big\{\mathbf{A}_{12}(\mathbf{y}_{1:N}^o;\Psi)^T \mathbf{W}^* \Big\}  \\\notag
&\hspace{1cm}  +   \tr \Big\{\mathbf{A}_{17}(\mathbf{y}_{1:N}^o;\Psi)^T \mathbf{W}^* \Big\} - \tr \Big\{\mathbf{A}_{18.1}(\mathbf{y}_{1:N}^o;\Psi)^T \mathbf{W}^* \Big\} - \tr \Big\{\mathbf{A}_{18.2}(\mathbf{y}_{1:N}^o;\Psi)^T \mathbf{W}^* \Big\} + \tr \Big\{\mathbf{A}_{19}(\mathbf{y}_{1:N}^o;\Psi)^T \mathbf{W}^* \Big\}   \\\notag
&\hspace{1cm} + 2 \Big( \bm{\mu}^{*}  A_{0}(\mathbf{y}_{1:N}^o;\Psi)-  \mathbf{A}_{5}(\mathbf{y}_{1:N}^o;\Psi) + \mathbf{A}_{8}(\mathbf{y}_{1:N}^o;\Psi) \mathbf{W}^{*T} -  \bm{\mu}  \mathbf{A}_{11}(\mathbf{y}_{1:N}^o;\Psi) \mathbf{W}^{*T} \Big) \bm{\delta}_\epsilon^{*T} \\\notag
&\hspace{1cm}+ 2 \Big(  \mathbf{A}_{7}(\mathbf{y}_{1:N}^o;\Psi,\Psi^*)  -  \bm{\mu}^{*} \mathbf{A}_{12}(\mathbf{y}_{1:N}^o;\Psi,\Psi^*) -   \bm{\delta}_\epsilon^{*}  \mathbf{W}^{*} \mathbf{A}_{4}(\mathbf{y}_{1:N}^o;\Psi) \Big) \big( \bm{\delta}_\epsilon \mathbf{W} - \sigma^2 \bm{\delta}_x \big)^T  \bigg],
\end{align*}
where 
\begin{align*}
& A_0 (\mathbf{y}_{1:N}^o;\Psi)_{1 \times 1} := \sum_{t=1}^N \bigg\{  H \Big(\mathbf{y}^o_{1:N/t};\Psi \Big) \  \int_0^1 \int_0^1 \  m(\mathbf{y}_{t}^o,s_{\epsilon,t},s_{x,t};\Psi) \ d s_{\epsilon,t} \ d s_{x,t} \bigg\} , \\\notag
&  \mathbf{A}_1 (\mathbf{y}_{1:N}^o;\Psi)_{d \times d} := \sum_{t=1}^N \bigg\{ H \Big(\mathbf{y}^o_{1:N/t};\Psi \Big) \ \int_0^1 \int_0^1  \mathbf{D}_{\epsilon,t} \  m(\mathbf{y}_{t}^o,s_{\epsilon,t},s_{x,t};\Psi) \ d s_{\epsilon,t} \ d s_{x,t} \bigg\}, \\\notag
& \mathbf{A}_2 (\mathbf{y}_{1:N}^o;\Psi)_{d \times d}:= \sum_{t=1}^N \bigg\{ H \Big(\mathbf{y}^o_{1:N/t};\Psi \Big) \ \int_0^1 \int_0^1  \mathbf{D}_{\epsilon,t}^{-1} \  m(\mathbf{y}_{t}^o,s_{\epsilon,t},s_{x,t};\Psi) \ d s_{\epsilon,t} \ d s_{x,t} \bigg\} ,\\\notag
&  \mathbf{A}_3 (\mathbf{y}_{1:N}^o;\Psi)_{k \times k}:= \sum_{t=1}^N \bigg\{ H \Big(\mathbf{y}^o_{1:N/t};\Psi \Big) \ \int_0^1 \int_0^1 \mathbf{D}_{x,t}^{-1} \  m(\mathbf{y}_{t}^o,s_{\epsilon,t},s_{x,t};\Psi) \ d s_{\epsilon,t} \ d s_{x,t} \bigg\} ,\\\notag
&  \mathbf{A}_4 (\mathbf{y}_{1:N}^o;\Psi)_{k \times k}:=  \sum_{t=1}^N \bigg\{ H \Big(\mathbf{y}^o_{1:N/t};\Psi \Big) \ \int_0^1 \int_0^1  \mathbf{M}_t^{-1} \  m(\mathbf{y}_{t}^o,s_{\epsilon,t},s_{x,t};\Psi) \ d s_{\epsilon,t} \ d s_{x,t} \bigg\} ,\\\notag
&  \mathbf{A}_5 (\mathbf{y}_{1:N}^o;\Psi)_{1 \times d}:=  \sum_{t=1}^N \bigg\{ H \Big(\mathbf{y}^o_{1:N/t};\Psi \Big) \ \int_0^1 \int_0^1  \mathbb{E}_{\mathbf{Y}_t^m|\mathbf{Y}_t^o,S_{\epsilon,t}, S_{x,t}, \Psi}\big[ \mathbf{Y}_t\big] \  m(\mathbf{y}_{t}^o,s_{\epsilon,t},s_{x,t};\Psi) \ d s_{\epsilon,t} \ d s_{x,t} \bigg\} ,\\\notag
&  \mathbf{A}_6 (\mathbf{y}_{1:N}^o;\Psi)_{1 \times d}:=  \sum_{t=1}^N \bigg\{ H \Big(\mathbf{y}^o_{1:N/t};\Psi \Big) \ \int_0^1 \int_0^1  \mathbb{E}_{\mathbf{Y}_t^m|\mathbf{Y}_t^o,S_{\epsilon,t}, S_{x,t}, \Psi}\big[ \mathbf{Y}_t\big] \mathbf{D}_{\epsilon,t}\  m(\mathbf{y}_{t}^o,s_{\epsilon,t},s_{x,t};\Psi) \ d s_{\epsilon,t} \ d s_{x,t} \bigg\} ,\\\notag
&  \mathbf{A}_7 (\mathbf{y}_{1:N}^o;\Psi,\Psi^*)_{1 \times k}:=  \sum_{t=1}^N \bigg\{ H \Big(\mathbf{y}^o_{1:N/t};\Psi \Big) \ \int_0^1 \int_0^1 \mathbb{E}_{\mathbf{Y}_t^m|\mathbf{Y}_t^o,S_{\epsilon,t}, S_{x,t}, \Psi}\big[ \mathbf{Y}_t\big]  \mathbf{D}_{\epsilon,t}  \mathbf{W}^* \mathbf{M}_t^{-1} \  m(\mathbf{y}_{t}^o,s_{\epsilon,t},s_{x,t};\Psi) \ d s_{\epsilon,t} \ d s_{x,t} \bigg\},  \\\notag
&  \mathbf{A}_8 (\mathbf{y}_{1:N}^o;\Psi)_{1 \times k}:=  \sum_{t=1}^N \bigg\{ H \Big(\mathbf{y}^o_{1:N/t};\Psi \Big) \ \int_0^1 \int_0^1  \mathbb{E}_{\mathbf{Y}_t^m|\mathbf{Y}_t^o,S_{\epsilon,t}, S_{x,t}, \Psi}\big[ \mathbf{Y}_t\big] \mathbf{D}_{\epsilon,t}  \mathbf{W} \mathbf{M}_t^{-1} \  m(\mathbf{y}_{t}^o,s_{\epsilon,t},s_{x,t};\Psi) \ d s_{\epsilon,t} \ d s_{x,t} \bigg\} ,\\\notag
& \mathbf{A}_9 (\mathbf{y}_{1:N}^o;\Psi,\Psi^*)_{1 \times d}:= \sum_{t=1}^N \bigg\{ H \Big(\mathbf{y}^o_{1:N/t};\Psi \Big) \ \int_0^1 \int_0^1 \mathbb{E}_{\mathbf{Y}_t^m|\mathbf{Y}_t^o,S_{\epsilon,t}, S_{x,t}, \Psi}\big[ \mathbf{Y}_t\big]  \mathbf{D}_{\epsilon,t}  \mathbf{W}^* \mathbf{M}_t^{-1}   \mathbf{W}^T \mathbf{D}_{\epsilon,t} \  m(\mathbf{y}_{t}^o,s_{\epsilon,t},s_{x,t};\Psi) \ d s_{\epsilon,t} \ d s_{x,t} \bigg\} ,\\\notag
& \mathbf{A}_{10} (\mathbf{y}_{1:N}^o;\Psi,\Psi^*)_{1 \times d}:= \sum_{t=1}^N \bigg\{ H \Big(\mathbf{y}^o_{1:N/t};\Psi \Big) \ \int_0^1 \int_0^1 \mathbb{E}_{\mathbf{Y}_t^m|\mathbf{Y}_t^o,S_{\epsilon,t}, S_{x,t}, \Psi}\big[ \mathbf{Y}_t\big] \mathbf{D}_{\epsilon,t}  \mathbf{W} \mathbf{M}_t^{-1} \mathbf{W}^{*T} \mathbf{D}_{\epsilon,t} \  m(\mathbf{y}_{t}^o,s_{\epsilon,t},s_{x,t};\Psi) \ d s_{\epsilon,t} \ d s_{x,t} \bigg\},  \\\notag
& \mathbf{A}_{11} (\mathbf{y}_{1:N}^o;\Psi)_{d \times k}:= \sum_{t=1}^N \bigg\{ H \Big(\mathbf{y}^o_{1:N/t};\Psi \Big) \ \int_0^1 \int_0^1 \mathbf{D}_{\epsilon,t} \mathbf{W} \mathbf{M}_t^{-1} \  m(\mathbf{y}_{t}^o,s_{\epsilon,t},s_{x,t};\Psi) \ d s_{\epsilon,t} \ d s_{x,t} \bigg\}, \\\notag
& \mathbf{A}_{12} (\mathbf{y}_{1:N}^o;\Psi,\Psi^*)_{d \times k}:=  \sum_{t=1}^N \bigg\{ H \Big(\mathbf{y}^o_{1:N/t};\Psi \Big) \ \int_0^1 \int_0^1 \mathbf{D}_{\epsilon,t} \mathbf{W}^{*} \mathbf{M}_t^{-1}  \  m(\mathbf{y}_{t}^o,s_{\epsilon,t},s_{x,t};\Psi) \ d s_{\epsilon,t} \ d s_{x,t} \bigg\}, \\\notag
& \mathbf{A}_{13} (\mathbf{y}_{1:N}^o;\Psi,\Psi^*)_{d \times d}:= \sum_{t=1}^N \bigg\{ H \Big(\mathbf{y}^o_{1:N/t};\Psi \Big) \ \int_0^1 \int_0^1 \mathbf{D}_{\epsilon,t} \mathbf{W} \mathbf{M}_t^{-1} \mathbf{W}^{*T} \mathbf{D}_{\epsilon,t} \  m(\mathbf{y}_{t}^o,s_{\epsilon,t},s_{x,t};\Psi) \ d s_{\epsilon,t} \ d s_{x,t} \bigg\} ,\\\notag
& \mathbf{A}_{14} (\mathbf{y}_{1:N}^o;\Psi)_{d \times k}:=  \sum_{t=1}^N \bigg\{  H \Big(\mathbf{y}^o_{1:N/t};\Psi \Big) \  \int_0^1 \int_0^1    \mathbf{D}_{\epsilon,t}\mathbb{E}_{\mathbf{Y}_t^m|\mathbf{Y}_t^o,S_{\epsilon,t}, S_{x,t}, \Psi}\big[ \mathbf{Y}_t^T \mathbf{Y}_t\big]  \mathbf{D}_{\epsilon,t}  \mathbf{W} \mathbf{M}_t^{-1} \  m(\mathbf{y}_{t}^o,s_{\epsilon,t},s_{x,t};\Psi) \ d s_{\epsilon,t} \ d s_{x,t} \bigg\} , \\\notag
& \mathbf{A}_{15} (\mathbf{y}_{1:N}^o;\Psi,\Psi^*)_{d \times k}:= \sum_{t=1}^N \bigg\{  H \Big(\mathbf{y}^o_{1:N/t};\Psi \Big) \  \int_0^1 \int_0^1 \mathbf{D}_{\epsilon,t}\bm{\mu}^{*T} \mathbb{E}_{\mathbf{Y}_t^m|\mathbf{Y}_t^o,S_{\epsilon,t}, S_{x,t}, \Psi}\big[ \mathbf{Y}_t\big] \mathbf{D}_{\epsilon,t}  \mathbf{W} \mathbf{M}_t^{-1} \  m(\mathbf{y}_{t}^o,s_{\epsilon,t},s_{x,t};\Psi) \ d s_{\epsilon,t} \ d s_{x,t} \bigg\} , \\\notag
& \mathbf{A}_{16} (\mathbf{y}_{1:N}^o;\Psi,\Psi^*)_{d \times k}:=\sum_{t=1}^N \bigg\{  H \Big(\mathbf{y}^o_{1:N/t};\Psi \Big) \  \int_0^1 \int_0^1 \mathbf{D}_{\epsilon,t}\Big( \mathbb{E}_{\mathbf{Y}_t^m|\mathbf{Y}_t^o,S_{\epsilon,t}, S_{x,t}, \Psi}\big[ \mathbf{Y}_t\big] - \bm{\mu}^{*}\Big)^T\Big( \bm{\mu}\mathbf{D}_{\epsilon,t}\mathbf{W} + \bm{\delta}_\epsilon\mathbf{W} - \sigma^2 \bm{\delta}_x\Big) \mathbf{M}_t^{-1} \\
&\hspace{0.3cm} \times  m(\mathbf{y}_{t}^o,s_{\epsilon,t},s_{x,t};\Psi) \ d s_{\epsilon,t} \ d s_{x,t} \bigg\} ,\\\notag
& \mathbf{A}_{17} (\mathbf{y}_{1:N}^o;\Psi,\Psi^*)_{k \times d}:= \sum_{t=1}^N \bigg\{  H \Big(\mathbf{y}^o_{1:N/t};\Psi \Big) \  \int_0^1 \int_0^1   \mathbf{D}_{\epsilon,t} \mathbf{W}^{*} \mathbf{M}_t^{-1} \mathbf{W} ^T\mathbf{D}_{\epsilon,t}  \mathbb{E}_{\mathbf{Y}_t^m|\mathbf{Y}_t^o,S_{\epsilon,t}, S_{x,t}, \Psi}\big[ \mathbf{Y}_t^T \mathbf{Y}_t\big]  \mathbf{D}_{\epsilon,t} \mathbf{W}\mathbf{M}_t^{-1} \\
&\hspace{0.3cm} \times  m(\mathbf{y}_{t}^o,s_{\epsilon,t},s_{x,t};\Psi) \ d s_{\epsilon,t} \ d s_{x,t} \bigg\}, \\\notag
&  \mathbf{A}_{18.1} (\mathbf{y}_{1:N}^o;\Psi,\Psi^*)_{k \times d}:= \sum_{t=1}^N \bigg\{  H \Big(\mathbf{y}^o_{1:N/t};\Psi \Big) \  \int_0^1 \int_0^1\mathbf{D}_{\epsilon,t}\mathbf{W}^{*}  \mathbf{M}_t^{-1} \mathbf{W} ^T\mathbf{D}_{\epsilon,t}   \mathbb{E}_{\mathbf{Y}_t^m|\mathbf{Y}_t^o,S_{\epsilon,t}, S_{x,t}, \Psi}\big[ \mathbf{Y}_t\big] ^T \Big(\bm{\mu}\mathbf{D}_{\epsilon,t}\mathbf{W}+ \bm{\delta}_\epsilon\mathbf{W} - \sigma^2 \bm{\delta}_x  \Big) \mathbf{M}_t^{-1} \\
&\hspace{0.3cm} \times  m(\mathbf{y}_{t}^o,s_{\epsilon,t},s_{x,t};\Psi) \ d s_{\epsilon,t} \ d s_{x,t} \bigg\},\\\notag
&  \mathbf{A}_{18.2} (\mathbf{y}_{1:N}^o;\Psi,\Psi^*)_{k \times d}:= \sum_{t=1}^N \bigg\{  H \Big(\mathbf{y}^o_{1:N/t};\Psi \Big) \  \int_0^1 \int_0^1\mathbf{D}_{\epsilon,t}\mathbf{W}^{*}  \mathbf{M}_t^{-1} \Big(\bm{\mu}\mathbf{D}_{\epsilon,t}\mathbf{W}+ \bm{\delta}_\epsilon\mathbf{W} - \sigma^2 \bm{\delta}_x  \Big)^T \mathbb{E}_{\mathbf{Y}_t^m|\mathbf{Y}_t^o,S_{\epsilon,t}, S_{x,t}, \Psi}\big[ \mathbf{Y}_t\big] \mathbf{D}_{\epsilon,t} \mathbf{W}\mathbf{M}_t^{-1} \\
&\hspace{0.3cm}\times  m(\mathbf{y}_{t}^o,s_{\epsilon,t},s_{x,t};\Psi) \ d s_{\epsilon,t} \ d s_{x,t} \bigg\}, \\\notag
& \mathbf{A}_{19} (\mathbf{y}_{1:N}^o;\Psi,\Psi^*)_{k \times d}:= \sum_{t=1}^N \bigg\{  H \Big(\mathbf{y}^o_{1:N/t};\Psi \Big) \  \int_0^1 \int_0^1 \mathbf{D}_{\epsilon,t}\mathbf{W}^{*} \mathbf{M}_t^{-1}\Big( \bm{\mu}\mathbf{D}_{\epsilon,t} \mathbf{W} + \bm{\delta}_\epsilon \mathbf{W} - \sigma^2 \bm{\delta}_x \Big)^T \Big( \bm{\mu}\mathbf{D}_{\epsilon,t} \mathbf{W} + \bm{\delta}_\epsilon \mathbf{W} - \sigma^2 \bm{\delta}_x \Big) \mathbf{M}_t^{-1} \\
&\hspace{0.3cm}\times  m(\mathbf{y}_{t}^o,s_{\epsilon,t},s_{x,t};\Psi) \ d s_{\epsilon,t} \ d s_{x,t} \bigg\} ,\\\notag
& A_{20} (\mathbf{y}_{1:N}^o;\Psi)_{1 \times 1}:=\sum_{t=1}^N \bigg\{ H \Big(\mathbf{y}^o_{1:N/t};\Psi \Big) \ \int_0^1 \int_0^1 \tr \Big\{  \mathbb{E}_{\mathbf{Y}_t^m|\mathbf{Y}_t^o,S_{\epsilon,t}, S_{x,t}, \Psi}\big[ \mathbf{Y}_t^T \mathbf{Y}_t\big]  \mathbf{D}_{\epsilon,t} \Big\} \  m(\mathbf{y}_{t}^o,s_{\epsilon,t},s_{x,t};\Psi) \ d s_{\epsilon,t} \ d s_{x,t} \bigg\},  \\\notag
& A_{21} (\mathbf{y}_{1:N}^o;\Psi,\Psi^*)_{1 \times 1}:= \sum_{t=1}^N \bigg\{ H \Big(\mathbf{y}^o_{1:N/t};\Psi \Big) \ \int_0^1 \int_0^1  \tr \Big\{ \mathbb{E}_{\mathbf{Y}_t^m|\mathbf{Y}_t^o,S_{\epsilon,t}, S_{x,t}, \Psi}\big[ \mathbf{Y}_t^T \mathbf{Y}_t\big]  \mathbf{D}_{\epsilon,t}  \mathbf{W} \mathbf{M}_t^{-1}\mathbf{W}^{*T} \mathbf{D}_{\epsilon,t} \Big\} \  m(\mathbf{y}_{t}^o,s_{\epsilon,t},s_{x,t};\Psi) \ d s_{\epsilon,t} \ d s_{x,t} \bigg\} .
\end{align*}
The maximizer that corresponds to the parameter $\mathbf{W}^*$ is more difficult to obtain since the partial derivative of the function $Q$ with respect to $\mathbf{W}^*$ requires integrating the matrix products which contain $\mathbf{W}^*$ and cannot be further simplified, that is
\begin{align*}
& \frac{\partial}{\partial \mathbf{W}^*} Q(\Psi,\Psi^*) = \mathbf{0} \Longleftrightarrow \sum_{t=1}^N \bigg\{ H \Big(\mathbf{y}^o_{1:N/t};\Psi \Big) \ \int_0^1 \int_0^1   \frac{\partial}{\partial \mathbf{W}^*}  \tilde{v}\big(\mathbf{y}_t^o, s_{\epsilon,t}, s_{x,t}; \Psi,\Psi^*\big)  \  m\big(\mathbf{y}_{t}^o,s_{\epsilon,t},s_{x,t};\Psi\big) \ d s_{\epsilon,t} \ d s_{x,t}  = \mathbf{0} \\\notag
& \Longleftrightarrow \mathbf{A}_{14} (\mathbf{y}_{1:N}^o;\Psi) - \mathbf{A}_{15} (\mathbf{y}_{1:N}^o;\Psi,\Psi^*) - \bm{\delta}_\epsilon^{*T} \mathbf{A}_{8} (\mathbf{y}_{1:N}^o;\Psi)- \mathbf{A}_{16} (\mathbf{y}_{1:N}^o;\Psi,\Psi^*)  + \bm{\delta}_\epsilon^{*T}\bm{\mu} \mathbf{A}_{11} (\mathbf{y}_{1:N}^o;\Psi) \\\notag
&\hspace{1cm}+ \bm{\delta}_\epsilon^{*T} \big(\bm{\delta} \mathbf{W} - \sigma^2 \mathbf{W} \big) \mathbf{A}_{4} (\mathbf{y}_{1:N}^o;\Psi)- \mathbf{A}_{17} (\mathbf{y}_{1:N}^o;\Psi,\Psi^*) +  \mathbf{A}_{18.1} (\mathbf{y}_{1:N}^o;\Psi,\Psi^*) + \mathbf{A}_{18.2} (\mathbf{y}_{1:N}^o;\Psi,\Psi^*) \\\notag
&\hspace{1cm}- \mathbf{A}_{19} (\mathbf{y}_{1:N}^o;\Psi,\Psi^*) - \sigma^2 \mathbf{A}_{12} (\mathbf{y}_{1:N}^o;\Psi,\Psi^*)= \mathbf{0}
\end{align*}
The maximizer with respect to $\mathbf{W}^*$ requires solving a root finding problem $\frac{\partial Q(\Psi,\Psi^*)}{\partial \mathbf{W}^*} = \mathbf{0}$. Recall that $\frac{\partial Q(\Psi,\Psi^*)}{\partial \mathbf{W}^*} $ contains elements of $\bm{\mu}^*$ and $\bm{\delta}_\epsilon^*$.
By denoting $f: \mathbb{R}^{d \times k} \rightarrow  \mathbb{R}^{d \times k} $ such that $f(\mathbf{W}^*;\Psi,\Psi^*): = \frac{\partial Q(\Psi,\Psi^*)}{\partial \mathbf{W}^*}$ is given by 
\begin{align*}
f(\mathbf{W}^*;\Psi,\Psi^*) &= \mathbf{A}_{14} (\mathbf{y}_{1:N}^o;\Psi) - \mathbf{A}_{15} (\mathbf{y}_{1:N}^o;\Psi,\Psi^*) - \bm{\delta}_\epsilon^{*T} \mathbf{A}_{8} (\mathbf{y}_{1:N}^o;\Psi)- \mathbf{A}_{16} (\mathbf{y}_{1:N}^o;\Psi,\Psi^*)  + \bm{\delta}_\epsilon^{*T}\bm{\mu} \mathbf{A}_{11} (\mathbf{y}_{1:N}^o;\Psi) \\\notag
&+ \bm{\delta}_\epsilon^{*T} \big(\bm{\delta} \mathbf{W} - \sigma^2 \mathbf{W} \big) \mathbf{A}_{4} (\mathbf{y}_{1:N}^o;\Psi)- \mathbf{A}_{17} (\mathbf{y}_{1:N}^o;\Psi,\Psi^*) +  \mathbf{A}_{18.1} (\mathbf{y}_{1:N}^o;\Psi,\Psi^*) + \mathbf{A}_{18.2} (\mathbf{y}_{1:N}^o;\Psi,\Psi^*) \\\notag
&- \mathbf{A}_{19} (\mathbf{y}_{1:N}^o;\Psi,\Psi^*) - \sigma^2 \mathbf{A}_{12} (\mathbf{y}_{1:N}^o;\Psi,\Psi^*)^T.
\end{align*}
the system of equations $\nabla_{\Psi^*} Q(\Psi,\Psi^*) = \mathbf{0}$ reduces to 
\begin{align*}
\begin{cases}
& \bm{\mu}^* = A_1(\mathbf{y}_{1:N}^o;\Psi)^{-1}\Big[ \mathbf{A}_6(\mathbf{y}_{1:N}^o;\Psi) - A_0(\mathbf{y}_{1:N}^o;\Psi) \bm{\delta}_\epsilon^* - \mathbf{A}_{10} (\mathbf{y}_{1:N}^o;\Psi) \mathbf{W}^{*T} + \bm{\mu}\mathbf{W}\mathbf{A}_{13}(\mathbf{y}_{1:N}^o;\Psi)\mathbf{W}^{*T} + \big(\bm{\delta}_\epsilon \mathbf{W} + \sigma^2 \bm{\delta}_x \big) \mathbf{A}_{11}(\mathbf{y}_{1:N}^o;\Psi)\mathbf{W}^{*T} \Big] ,\\\notag
& \bm{\delta}_\epsilon^* =   A_2(\mathbf{y}_{1:N}^o;\Psi)^{-1} \Big[ \mathbf{A}_5(\mathbf{y}_{1:N};\Psi) - A_0(\mathbf{y}_{1:N};\Psi)  \bm{\mu}^* - \mathbf{A}_8(\mathbf{y}_{1:N}^o;\Psi) \mathbf{W}^{*T} + \bm{\mu} \mathbf{W} \mathbf{A}_{11} (\mathbf{y}_{1:N}^o;\Psi) \mathbf{W}^{*T} +   \big( \bm{\delta}_\epsilon \mathbf{W} - \sigma^2 \bm{\delta}_x \big) \mathbf{A}_4(\mathbf{y}_{1:N}^o;\Psi) \mathbf{W}^{*T} \Big] ,\\\notag
& \bm{\delta}_x^* =  A_3(\mathbf{y}_{1:N}^o;\Psi)^{-1} \Big( \mathbf{A}_{8}(\mathbf{y}_{1:N};\Psi) - \bm{\mu} \mathbf{W} \mathbf{A}_{11}(\mathbf{y}_{1:N}^o;\Psi) - \big( \bm{\delta}_\epsilon\mathbf{W} - \sigma^2 \bm{\delta}_x \big)\mathbf{A}_{4}(\mathbf{y}_{1:N}^o;\Psi\big) \Big), \\\notag
& f(\mathbf{W}^*;\Psi,\Psi^*)  = 0 ,  \\\notag
& \sigma^{*2} = \frac{1}{dA_0(\mathbf{y}^o_{1:N};\Psi)} \Bigg[  A_{20}(\mathbf{y}^o_{1:N};\Psi) -  2\mathbf{A}_{6}(\mathbf{y}^o_{1:N};\Psi)\bm{\mu}^{*T} -2 \mathbf{A}_{5}(\mathbf{y}^o_{1:N};\Psi) \bm{\delta}_\epsilon^{*T} +  A_{1}(\mathbf{y}^o_{1:N};\Psi) \bm{\mu}^{*} \bm{\mu}^{*T} + 2A_{0}(\mathbf{y}^o_{1:N};\Psi)\bm{\mu}^* \bm{\delta}_\epsilon^{*T}  \\\notag
&\hspace{0.3cm}+ A_{2}(\mathbf{y}^o_{1:N};\Psi) \bm{\delta}_\epsilon^* \bm{\delta}_\epsilon^{*T} -  2 \tr \Big\{ \mathbf{A}_{14}(\mathbf{y}^o_{1:N};\Psi) \mathbf{W}^{*T} \Big\}+ 2\mathbf{A}_{9}(\mathbf{y}^o_{1:N};\Psi,\Psi^*)\mathbf{W}^T \bm{\mu}^T + 2 \mathbf{A}_{7}(\mathbf{y}^o_{1:N};\Psi,\Psi)\mathbf{W}^T\bm{\delta}_\epsilon^T \\\notag
&\hspace{0.3cm}  - 2 \sigma^2 \mathbf{A}_{7}(\mathbf{y}^o_{1:N};\Psi,\Psi)\bm{\delta}_x^T  + 2 \Big( \mathbf{A}_{10}(\mathbf{y}^o_{1:N};\Psi)  - \bm{\mu} \mathbf{W}\mathbf{A}_{13}(\mathbf{y}^o_{1:N};\Psi) -  \big(\bm{\delta}_\epsilon \mathbf{W} - \sigma^2 \bm{\delta}_x\big)\mathbf{A}_{11}(\mathbf{y}^o_{1:N};\Psi)\Big)\mathbf{W}^{*T} \bm{\mu}^{*T} \\\notag
&\hspace{0.3cm}  + 2 \Big( \mathbf{A}_{8}(\mathbf{y}^o_{1:N};\Psi) -  \bm{\mu} \mathbf{W}\mathbf{A}_{11}(\mathbf{y}^o_{1:N};\Psi)  -  \big(\bm{\delta}_\epsilon \mathbf{W} - \sigma^2 \bm{\delta}_x\big)\mathbf{A}_{4}(\mathbf{y}^o_{1:N};\Psi)\Big)\mathbf{W}^{*T} \bm{\delta}_\epsilon^{*T} +  \tr \Big\{ \mathbf{A}_{17}(\mathbf{y}^o_{1:N};\Psi)   \mathbf{W}^{*T}\mathbf{W}^{*}  \Big\}  \\\notag
& \hspace{0.3cm} -  \tr \Big\{ \Big(\mathbf{A}_{18}(\mathbf{y}^o_{1:N};\Psi) + \mathbf{A}_{18}(\mathbf{y}^o_{1:N};\Psi)^T\Big)   \mathbf{W}^{*T}  \mathbf{W}^{*}\Big\} +\tr \Big\{  \mathbf{A}_{19}(\mathbf{y}^o_{1:N};\Psi) \mathbf{W}^{*T}  \mathbf{W}^{*}  \Big\}  + \sigma^2  \tr \Big\{ \mathbf{A}_{11}(\mathbf{y}^o_{1:N};\Psi) \mathbf{W}^{*T}  \mathbf{W}^{*}  \Big\}\Bigg] .
\end{cases}
\end{align*}
\end{proof}

\section{Simplified Skew-t GSt PPCA}\label{appendix:skewX_syn_case_studies}
We calculate the steps of the EM algorithm for a simplified model of the Skew-t GSt PPCA than introduced in Subsection~\ref{ssec:PPCA_GStS_ind}. The following algorithm is derived using the assumption that there is no missing data and the skewness is only present in the representation of $\mathbf{X}_t$, that is
\begin{align}\label{eq:appendix_skewX_model} 
& \mathbf{X}_{t, \ 1\times k} = \bm{\delta}_x V_t^{-1} +\sqrt{ V_t^{-1} }\mathbf{Z}_{x,t} , \ \bm{\epsilon}_{t \ 1\times d} = \sqrt{\sigma^2 U_t^{-1}} \mathbf{Z}_{\epsilon,t} , \  \mathbf{Y}_{t, \ 1\times d} = \bm{\mu } + \mathbf{X}_t \mathbf{W}^T + \bm{\epsilon}_t,
\end{align}
for
\begin{align*}
&U_t  := T_\epsilon(s) := \frac{\chi_{\nu_{\epsilon}^1}^{-1} (s)}{\nu_{\epsilon}}  \mbox{ and } \ V_t:=T_x(s) :=  \frac{\chi_{\nu_x^1}^{-1} (s)}{\nu_x}.
\end{align*}
In the follow-up publication to this work, we show that the Skew-t GSt PPCA is subject to identification problems that arise from the joint estimation of $\bm{\mu}$ and $\bm{\delta}_x$ via the EM algorithm. To address this problem, we propose to exclude the estimation of both parameters from the EM algorithm. The new set of normal equations that defines the iterative maximisation of $Q$ does not include the step that maximizes the objective function with respect to $\bm{\mu}$ or $\bm{\delta}_x$. 

We want to avoid making an assumption that the observation process is zero- mean.  Therefore, we argue that if there is no presence of missing values, we can introduce a simple correction that provides good accuracy of estimates for all parameters, $\bm{\mu},\mathbf{W},\sigma^2$ and $\bm{\delta}_x$, when parameters $\mathbf{W}, \ \sigma^2$ are specified by the EM algorithm that assumes no intercept term.  We determine $\bm{\delta}_x$ by a grid search, and the intercept term $\bm{\mu}$ is obtained by an iterative adjustment. We use the fact that the first moment of the marginal distribution of $\mathbf{X}_t$ equals
\begin{align*}\notag
& \mathbb{E}_{\mathbf{X}_t|\Psi}\big[ \mathbf{X}_t \big] = \mathbb{E}_{V_t|\Psi}\Big[ \mathbb{E}_{\mathbf{X}_t|V_t, \Psi }\big[ \mathbf{X}_t \big]  \Big]  =  \frac{v_x}{v_x-2} \bm{\delta}_x ,
\end{align*}
and consequently
\begin{align*}
& \mathbb{E}_{ \mathbf{Y}_{t}|\Psi } \big[ \mathbf{Y}_{t}\big] = \bm{\mu } +  \frac{v_x}{v_x-2} \bm{\delta}_x  \mathbf{W}^T.
\end{align*}
Therefore, we specify the update of $\bm{\mu}$, $\bm{\mu }^*$, over each iteration of the EM algorithm as
\begin{align*}
&  \bm{\mu }^* = \frac{1}{N} \sum_{t = 1}^N \mathbf{Y}_{t} -  \frac{v_x}{v_x-2} \bm{\delta}_x  \mathbf{W}^T, 
\end{align*}
for fixed $\bm{\delta}_x$, and calculate the maximizers $ \mathbf{W}^*$ and $\sigma^{*2}$ given the centred realisations $\tilde{\mathbf{Y}}_t = \mathbf{Y}_{t} - \bm{\mu }^* $. Instead of the sample average, one can use more robust estimators of the first moment. 

The steps of the EM algorithm for the centred data (no $\bm{\mu}$) with $\bm{\delta}_x$ specified on a grid assumes the following stochastic representation of the observation process
\begin{align}\label{eq:appendix_skewX_model3}
& \tilde{\mathbf{Y}}_{t, \ 1\times d} =  \mathbf{X}_t \mathbf{W}^T + \bm{\epsilon}_t,
\end{align}
and the corresponding conditional distribution $\tilde{\mathbf{Y}}_{t}  |\mathbf{X}_{t}, \mathbf{U}_{t},\mathbf{V}_{t} , \Psi$
\begin{align*}
\pi (\tilde{\mathbf{y}}_{t}  |\mathbf{x}_{t}, u_{t},v_{t} , \Psi)  =   \big(2 \pi \big)^{-\frac{d}{2}}  \big( \sigma^2\big)^{-\frac{d}{2}}  U_t^{\frac{d}{2}}  \exp \Big\{ -\frac{1}{2 \sigma^2} \Big( U_t \tilde{\mathbf{y}}_t\tilde{ \mathbf{y}}_t ^T - 2 U_t \tilde{\mathbf{y}}_t \mathbf{W}  \mathbf{x}_t ^T + U_t  \mathbf{x}_t \mathbf{W}^T\mathbf{W}  \mathbf{x}_t^T\Big)\Big\} .
\end{align*}
The steps of the algorithm are detailed in Algorithm \ref{algo_adj} given Lemma \ref{lemma:Estep_GStS_skewX_c} and Lemma \ref{lemma:Mstep_GStS_skewX_c}.

\begin{algorithm}[H]\label{algo_adj}
\caption{Algorithm of the EM algorithm for Skew-t GSt PPCA with the parameter $\bm{\delta}_x$ being specified on a grid of values and the adjustment for the intercept $\bm{\mu}$.}
\BlankLine
\addtolength\linewidth{-12ex}

\KwIn{Define $v_\epsilon, v_x$ and the grid for $\bm{\delta}_x^{grid}$; }
\KwIn{Define initial  parameters for the EM algorithm, $\mathbf{W}^{(0)},\sigma^{(0)2}$;}

Calculate $\bar{\bm{\mu}}  = \frac{1}{N} \sum_{t = 1}^N \mathbf{y}_{t}$; \\
\For{$\bm{\delta}_x^{grid} = \big\{a,\ldots,b\big\}$}{
Initialize $\mathbf{W}^{(0)},\sigma^{(0)2}$;\\
\For{$i = 1, \ldots, M$}{
Specify 
$\bm{\mu }^{(i) }= \bar{\bm{\mu}} -  \frac{v_x}{v_x-2} \bm{\delta}_x^{grid}  \mathbf{W}^{(i-1)T} 
$; \\
Specify $\tilde{\mathbf{Y}}_{N \times d}^{(i)} = \mathbf{Y}_{N \times d}- \bm{\mu }^{(i) } $; \\
Calculate maximizers $\mathbf{W}^{(i) }, \sigma^{(i)2}$ as in Lemma \ref{lemma:Mstep_GStS_skewX_c} for the centred data set $\tilde{\mathbf{Y}}_{1:N}^{(i)}$;  \\

}
}
Select the optimal $\bm{\delta}^*_x$ with the highest log-likelihood for the sample $\mathbf{y}_{1:N}$.

\end{algorithm}

\begin{lemma}\label{lemma:Estep_GStS_skewX_c}
Let the observation vector $\tilde{\mathbf{Y}}_t$ and the latent variables $\bm{\epsilon}_t, \ \mathbf{X}_t,  \ \mathbf{U}_t$ and $\mathbf{V}_t$ be modelled as in \eqref{eq:appendix_skewX_model} and \eqref{eq:appendix_skewX_model3}. Given $N$ realisations of the random vector $\tilde{\mathbf{Y}}_t$ , the E-step of the Expectation-Maximisation algorithm for centred Skew-t GSt PPCA in the complete data setting is specified as
\begin{align*} 
Q(\Psi, \Psi^* )& = \frac{1}{\pi_{\tilde{\mathbf{Y}}_{1:N}| \Psi}(\mathbf{y}_{1:N})} \sum_{t=1}^N \Big\{ I_1(\tilde{\mathbf{y}}_{t};\Psi,\Psi^*) \prod_{s=1, s \neq t}^N  I_2(\tilde{\mathbf{y}}_s;\Psi) \Big\}. 
\end{align*}
The functions $I_1, \ I_2:\mathbb{R}^{d} \rightarrow \mathbb{R}$ are defined as 
\begin{align*}\notag
& I_1(\tilde{\mathbf{y}_{t}};\Psi,\Psi^*) :=  \int_0^1 \int_0^1 w(\tilde{\mathbf{y}}_{t},s_{\epsilon,t},s_{x,t};\Psi,\Psi^*) \  m(\tilde{\mathbf{y}}_{t},s_{\epsilon,t},s_{x,t};\Psi) \ d s_{\epsilon,t} \ d s_{x,t} ,\\
& I_2(\tilde{\mathbf{y}_{t}};\Psi) :=  \int_0^1 \int_0^1    m(\tilde{\mathbf{y}}_{t},s_{\epsilon,t},s_{x,t};\Psi) \ d s_{\epsilon,t} \ d s_{x,t},
\end{align*}
where $\tilde{w}\big(\tilde{\mathbf{y}}_{t},s_{\epsilon,t},s_{x,t};\Psi,\Psi^*\big):\mathbb{R}^{d} \times [0,1]^2 \longrightarrow \mathbb{R}$ is defined as  $\tilde{w}\big(\tilde{\mathbf{y}}_{t}, s_{\epsilon,t},s_{x,t};\Psi,\Psi^*\big): = w\big(\tilde{\mathbf{y}}_{t},T_\epsilon(s_{\epsilon,t}),T_x(s_{x,t});\Psi,\Psi^*\big)$ for 
\begin{align*}
w\big(\tilde{\mathbf{y}}_{t},s_{\epsilon,t},s_{x,t};\Psi,\Psi^*\big) & = \int_{ \mathbb{R}^{k} }  \log \pi (\tilde{\mathbf{y}}_{t},\mathbf{x}_{t},u_{t},v_{t} ; \Psi^*) \pi (\mathbf{x}_{t}|\tilde{\mathbf{y}}_{t},u_{t},v_{t} ; \Psi)   d \mathbf{x}_{t},
\end{align*}
for $ \pi (\mathbf{x}_{t}|\mathbf{y}_{t},u_{t},v_{t} | \Psi)$ being a probability density function of the conditional random vector 
$$
\mathbf{X}_t|\tilde{ \mathbf{Y}}_{t},U_{t},V_{t}, \Psi \sim \mathcal{N} \bigg( \big( U_t \tilde{\mathbf{Y}}_t\mathbf{W}  + \sigma^2 \bm{\delta}_x \big)  \mathbf{M}_t^{-1},  \frac{1}{ \sigma^2} \mathbf{M}_t   \bigg) ,
$$
$\mathbf{M}_t = T_\epsilon(s_{\epsilon,t})  \mathbf{W}^T\mathbf{W}  + \sigma^2 T_x(s_{x,t}) \mathbb{I}_k$  and the function $m: \mathbb{R}^{d_o} \times [0,1]^2 \longrightarrow \mathbb{R}$ is given by 
\begin{align*}
m(\tilde{\mathbf{y}}_t,s_{\epsilon,t}, s_{x,t} , \Psi)& :=\big(2 \pi \big)^{-\frac{d}{2}}  \big( \sigma^2\big)^{\frac{k-d}{2}}  T_\epsilon(s_{\epsilon,t})^{\frac{d}{2}}   T_x(s_{x,t})^{ \frac{k}{2}} \Big|   \mathbf{M}_t \Big|^{-\frac{1}{2}}  \exp \Big\{ -\frac{ 1}{2\sigma^2}\Big( T_\epsilon(s_{\epsilon,t}) \tilde{\mathbf{y}}_t \tilde{\mathbf{y}}_t ^T \Big\}\\
& \times \exp \Big\{ -\frac{ 1}{2\sigma^2}\Big(\sigma^2 T_x(s_{x,t})^{-1} \bm{\delta}_x  \bm{\delta}_x^T  -  \big( T_\epsilon(s_{\epsilon,t}) \tilde{\mathbf{y}}_t \mathbf{W}  + \sigma^2 \bm{\delta}_x \big)  \mathbf{M}_t^{-1}    \big( T_\epsilon(s_{\epsilon,t}) \tilde{\mathbf{y}}_t\mathbf{W}  + \sigma^2 \bm{\delta}_x \big) ^T \Big)\Big\}.
\end{align*}
\end{lemma}
\begin{proof}
The proof of Lemma \ref{lemma:Estep_GStS_skewX_c} follows the same steps as the proof of Theorem~\ref{th:Estep_GStS_missing_ind} in Subsection \ref{proof:th_Estep_GStS_missing_ind} in Appendix \ref{appendix:proofs_GStS} but assuming no $\bm{\delta}_\epsilon$ and no $\bm{\mu}$.
\end{proof}

\begin{lemma} \label{lemma:Mstep_GStS_skewX_c}
The solution to the system of equation $\nabla_{\Psi^*} Q(\Psi, \Psi^* ) = \mathbf{0}$  which determines the maximizers of the function $Q$ from Lemma \ref{lemma:Estep_GStS_skewX_c} with respect to the parameters $ \sigma^2$ and $\mathbf{W}$ are given by explicit formulas defined by two-dimensional integration problems on the hypercube $[0,1]^2$ as follows
{\footnotesize 
\begin{align*}
\begin{cases}
&\mathbf{W}^{*} = \Big[ \mathbf{A}_{14}  + \sigma^2  \mathbf{A}_{16}   \Big] \Big(\sigma^2\mathbf{A}_{11} + \mathbf{A}_{171819} \Big)^{-1}, \\\notag 
& \sigma^{*2} = \frac{1}{d}A_0^{-1} \bigg[  A_{20}(\mathbf{y}_{1:N};\Psi) - 2 \tr \Big\{ \mathbf{A}_{14} \mathbf{W}^{*T} \Big\}  - 2 \sigma^2 \mathbf{A}_{22} \bm{\delta}_x^T  + \tr \bigg\{ \Big( \sigma^2 \mathbf{A}_{11} +\mathbf{A}_{171819} \Big) \mathbf{W}^{*T} \mathbf{W}^*\bigg\} \bigg],
\end{cases}
\end{align*}
}
where the two-dimensional integrals $\mathbf{A}_i$ on the hypercube $[0,1]^2$, for $i \in \Big\{ 0, \ldots, 22 \Big\}$, are defined in Subsection \ref{proof:cor_Mstep_explicitFormulas_missing_ind} in Appendix \ref{appendix:proofs_GStS} with the exceptions
\begin{align*}
& \mathbf{A}_{16} (\mathbf{y}_{1:N};\Psi)_{d \times k}:=\sum_{t=1}^N \bigg\{  H \big(\mathbf{y}_{1:N/t};\Psi \big) \  \mathbf{y}_t^T \int_0^1 \int_0^1 T_\epsilon(s_{\epsilon,t}) \bm{\delta}_x  \mathbf{M}_t^{-1} \   m(\mathbf{y}_{t},s_{\epsilon,t},s_{x,t};\Psi) \ d s_{\epsilon,t} \ d s_{x,t} \bigg\}; \\
& \mathbf{A}_{171819}(\mathbf{y}_{1:N};\Psi)_{ k \times k} := \sum_{t=1}^N \bigg\{  H \big(\mathbf{y}_{1:N/t};\Psi \big) \  \int_0^1 \int_0^1 T_\epsilon(s_{\epsilon,t}) \mathbf{M}_t^{-1}  \Big(  T_\epsilon(s_{\epsilon,t}) \mathbf{y}_t \mathbf{W} + \sigma^2 \bm{\delta}_x \Big) ^T  \\
& \hspace{3cm}\times \Big(  T_\epsilon(s_{\epsilon,t}) \mathbf{y}_t \mathbf{W} + \sigma^2 \bm{\delta}_x \Big) \mathbf{M}_t^{-1}   m(\mathbf{y}_{t},s_{\epsilon,t},s_{x,t};\Psi) \ d s_{\epsilon,t} \ d s_{x,t} \bigg\} .
\end{align*}

\end{lemma}
\begin{proof}
The proof of Lemma \ref{lemma:Mstep_GStS_skewX_c}  follows the steps of the proof to Theorem~\ref{th:Mstep_explicitFormulas_missing_ind} in Subsection \ref{proof:cor_Mstep_explicitFormulas_missing_ind} in Appendix \ref{appendix:proofs_GStS} but assuming no $\bm{\delta}_\epsilon$ and $\bm{\mu}$ and no maximizers of the function $Q$ with respect to $\bm{\delta}_x$ .
\end{proof}

\section{Figures \& Tables}

\begin{figure}[h]
  \centering   
\includegraphics[scale = 0.10]{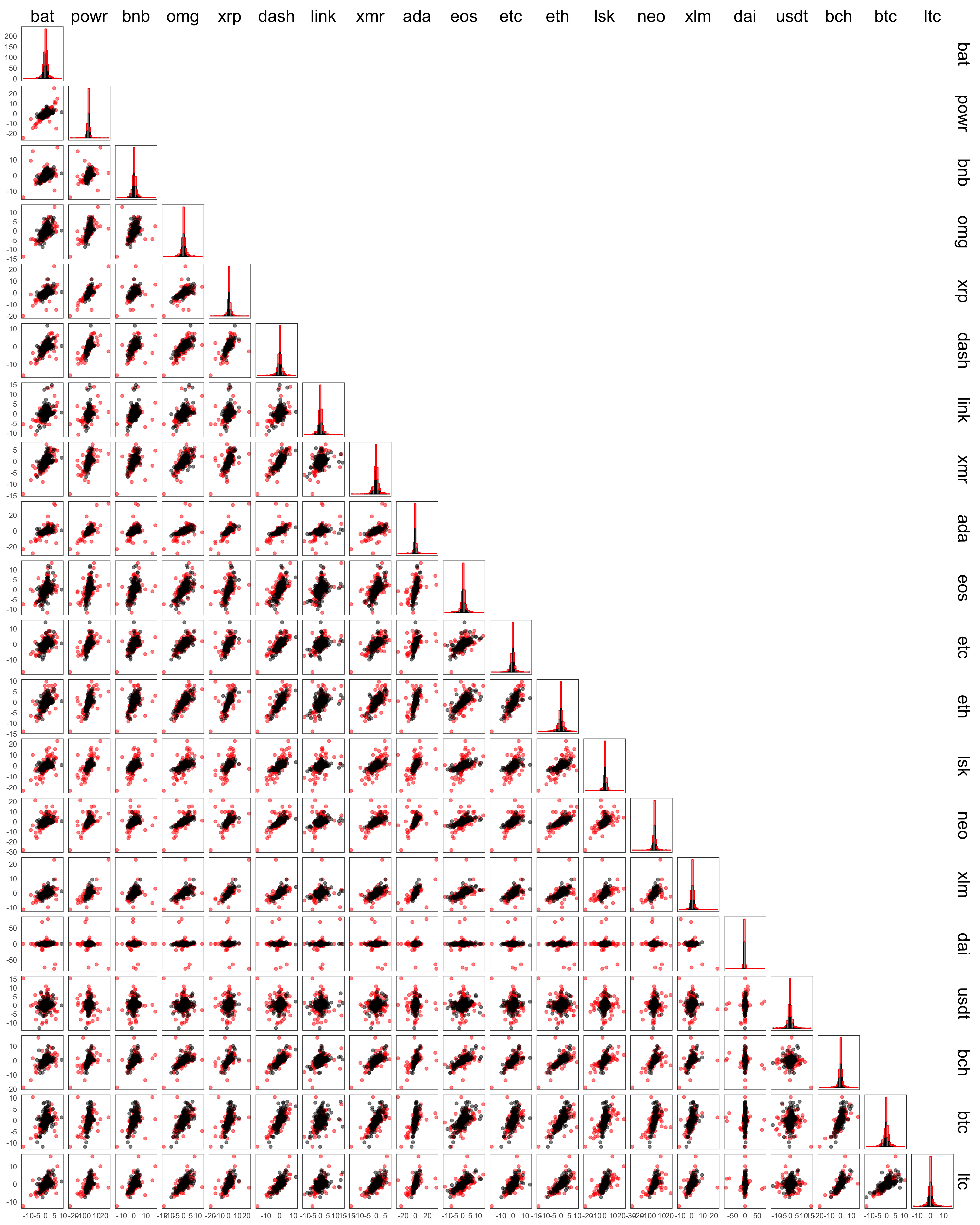}     
   \caption{The pair plots of linear interactions for standardized daily returns of $20$ crypto assets (rows and columns names of the panels) listed in Table \ref{tab:crypto_list}. The colors corresponds to the different sample periods, 2018 (red) and 2019 (black).}\label{fig:5crypto_pairplots_cat_2019}
\end{figure}

\begin{table}[htb]
\caption{The model choices (selected degrees of freedom, if applicable) and resulting log-likelihood ($\log L$) for $5$ PPCA franeworks: the Gaussian PPCA, Student-t PPCA, Student-t GSt PPCA, Grouped-t GSt PPCA and Skew-t GSt PPCA for standardized daily returns of $20$ crypto assets listed in Table \ref{tab:crypto_list} over the two sample periods, 2018 and 2019.  }
\centering \small
\begin{tabular}{lll}
\hline
  & 2018 &  2019 \\ 
  \hline
Gaussian PPCA& \\ 
  \hline
$\log L$  & -5451.735  & -2743.588 \\ 
   \hline
Student-t PPCA & \\ 
   \hline
   $\log L$ & -10455.642  &  -9989.970 \\ 
  $\nu$   & 2 & 2 \\ 
   \hline  
Student-t GSt PPCA & \\     \hline
  $\log L$  & -10385.006  & -10003.074 \\ 
  $\nu_{\epsilon}$ &  2 & 2 \\ 
  $\nu_{x}$  & 2 & 2 \\ 
   \hline 
Skew-t GSt PPCA & \\    \hline 
  $\log L$  & -11005.430  &  -9841.640 \\ 
  $\nu_{\epsilon}$  & 4 & 4 \\ 
  $\nu_{x}$  & 4 & 4 \\
   \hline 
Grouped-t GSt PPCA & \\      \hline
  $\log L$  & -10366.933  &  -9987.246 \\ 
  $\nu_{\epsilon,Advertising}$  & 2 & 2 \\ 
  $\nu_{\epsilon,Exchange}$ & 2 & 2 \\ 
  $\nu_{\epsilon,Technology}$  & 2 & 2 \\ 
  $\nu_{\epsilon,Energy}$  & 2  & 2 \\ 
  $\nu_{\epsilon,Smart Contracts}$ & 2 & 2 \\ 
  $\nu_{\epsilon,Interoperability}$  & 2 & 100 \\ 
  $\nu_{\epsilon,Governance}$  & 2 & 2 \\ 
  $\nu_{\epsilon,Privacy}$  & 2 & 2 \\ 
   $\nu_{\epsilon,Financial Service}$  & 2 & 2 \\ 
  $\nu_{\epsilon,Stablecoin}$ & 100  & 2 \\ 
  $\nu_{x,1}$  & 2 & 2 \\ 
  $\nu_{x,2}$  & 2 & 2 \\ 
  $\nu_{x,3}$ & 2 & 2 \\ 
   \hline
\end{tabular}\label{tab:20ccy_res}
\end{table}

\end{document}